\documentclass[reqno,11pt,letter]{amsart}

\usepackage[full]{textcomp}
\usepackage{newpxtext}
\usepackage{times}
\usepackage{amsmath,amssymb}
\usepackage{cabin} 
\usepackage[varqu,varl]{inconsolata} 
\usepackage[bigdelims,vvarbb]{newpxmath} 
\usepackage[bottom=.8in,top=.83in,left=1.15in,right=1.15in]{geometry}

\usepackage{lineno}

\usepackage{savesym}
\let\savedegree\bigtimes
\let\bigtimes\relax
\usepackage{mathabx}
\let\bigtimes\savedegree
\usepackage[usenames,dvipsnames]{color}
\usepackage{hyperref}
\hypersetup{
 colorlinks,
 linkcolor={blue!95!black},
 citecolor={blue!95!black},
 urlcolor={blue!50!black}
}
\usepackage{tikz}
\usetikzlibrary{decorations.markings}
\usepackage{mathrsfs}
\usepackage[all,cmtip]{xy}
\usepackage{mleftright}
\usepackage{booktabs}
\usepackage{mathdots}
\usepackage[foot]{amsaddr}

\usepackage{cite}
\usepackage{enumitem}

\usepackage{graphicx}
\usepackage{subfigure}
\usepackage{caption}
\usepackage{placeins}
\usepackage{comment}
\usepackage[normalem]{ulem}

\newcommand{\la}{\langle}
\newcommand{\ra}{\rangle}
\newcommand{\sgn}{\operatorname{sgn}}

\newcommand{\phys}{\text{phys}}

\newcommand{\vspan}{\operatorname{span}}

\newcommand{\per}{\text{per}}
\newcommand{\qua}{\text{qua}}

\newcommand{\bif}{\text{bif}}
\newcommand{\guess}{\text{guess}}

\newcommand{\td}{\tilde}

\newcommand{\mfq}{\mathfrak{q}}

\newcommand{\Vp}[1]{\mc V_{#1}^\per}

\setlist[enumerate]{labelsep=*, leftmargin=1.5pc}
\setlist[enumerate]{label=\normalfont(\roman*), ref=\roman*}

\newtheorem{theorem}{Theorem}[section]
\newtheorem{lem}[theorem]{Lemma}

\newtheorem{definition}[theorem]{Definition}
\theoremstyle{definition}

\newtheorem{remark}[theorem]{Remark}

\numberwithin{equation}{section}

\newcommand{\IGNORE}[1]{}
\newcommand{\ignore}[1]{}
\newcommand{\veps}{\varepsilon}
\newcommand{\opn}{\operatorname}

\newcommand{\diag}{\operatorname{diag}}
\newcommand{\phm}{\phantom{-}}
\newcommand{\mbb}[1]{\mathbb{#1}}

\newcommand{\mc}[1]{\mathcal{#1}}

\newcommand{\jt}{\textstyle}

\newcommand{\der}[2]{\frac{\partial #1}{\partial #2}}

\newcommand{\pa}{\partial}

\newcommand{\e}[1]{{(#1)}}

\usepackage{xcolor}
\usepackage{etoolbox}
\makeatletter
\pretocmd\@bibitem{\color{black}\csname keycolor#1\endcsname}{}{\fail}
\newcommand\citecolor[1]{\@namedef{keycolor#1}{\color{blue}}}
\makeatother

\begin{document}

\author[Jon Wilkening and Xinyu Zhao]{Jon Wilkening and Xinyu Zhao}
\address{Department of Mathematics, University of California, Berkeley,
  Berkeley, CA 94720, USA}
\email{wilkening@berkeley.edu}
\address{Department of Mathematics and Statistics, McMaster University,
  Hamilton, Ontario, Canada L8S 4K1}
\email{zhaox171@mcmaster.ca}


\subjclass[]{}

\title{Spatially quasi-periodic bifurcations from periodic traveling
  water waves and a method for detecting bifurcations using signed
  singular values}

\begin{abstract}
  We present a method of detecting bifurcations by locating zeros of a
  signed version of the smallest singular value of the Jacobian.  This
  enables the use of quadratically convergent root-bracketing
  techniques or Chebyshev interpolation to locate bifurcation points.
  Only positive singular values have to be computed, though the method
  relies on the existence of an analytic or smooth singular value
  decomposition (SVD). The sign of the determinant of the Jacobian,
  computed as part of the bidiagonal reduction in the SVD algorithm,
  eliminates slope discontinuities at the zeros of the smallest
  singular value.  We use the method to search for spatially
  quasi-periodic traveling water waves that bifurcate from
  large-amplitude periodic waves. The water wave equations are
  formulated in a conformal mapping framework to facilitate the
  computation of the quasi-periodic Dirichlet-Neumann operator.  We
  find examples of pure gravity waves with zero surface tension and
  overhanging gravity-capillary waves. In both cases, the waves have
  two spatial quasi-periods whose ratio is irrational.  We follow the
  secondary branches via numerical continuation beyond the realm of
  linearization about solutions on the primary branch to obtain
  traveling water waves that extend over the real line with no two
  crests or troughs of exactly the same shape.  The pure gravity wave
  problem is of relevance to ocean waves, where capillary effects can
  be neglected. Such waves can only exist through secondary
  bifurcation as they do not persist to zero amplitude. The
  gravity-capillary wave problem demonstrates the effectiveness of
  using the signed smallest singular value as a test function for
  multi-parameter bifurcation problems. This test function becomes
  mesh independent once the mesh is fine enough.
\end{abstract}

\maketitle
\markboth{J. WILKENING AND X. ZHAO}{QUASI-PERIODIC TRAVELING WATER WAVES}



\section{Introduction}
\label{sec:intro}



This paper has the dual purpose of carrying out a detailed
computational study of new families of spatially quasi-periodic
traveling water waves that bifurcate from finite-amplitude periodic
waves and developing a new test function for detecting bifurcations
in general. We begin with a discussion of water waves.

The study of traveling solutions of the free-surface water wave
problem has a long history.  Stokes \cite{stokes1880theory} first
studied two-dimensional periodic traveling solutions in the
gravity-driven case without surface tension.  He constructed an
asymptotic expansion of the solution in powers of an amplitude
parameter and conjectured that the highest-amplitude solution
possesses a wave crest with a sharp $120^\circ$ interior corner angle.  This was proved 100 years
later by Amick, Fraenkel and Toland \cite{amick:82}. Longuet-Higgins
and Fox \cite{lhf:78} carried out a matched asymptotic analysis of the
almost-highest traveling water wave and discovered that interesting
oscillatory structures develop near the wave crest as the wave height
approaches that of the sharply-crested wave of greatest height. They
also showed that the wave speed ceases to increase monotonically as
the wave crest sharpens, and instead possesses an infinite number of
turning points.

The problem of traveling gravity waves in two dimensions can be
studied as a bifurcation problem with a one-dimensional kernel
\cite{buffoni2000regularity, buffoni2000sub,
  toland1996stokes}. In this setting, nonlinear solutions form a
bifurcation branch from the zero amplitude solution. This branch is
called the primary branch. Plotnikov \cite{plotnikov1992nonuniqueness}
proved that there are infinitely many critical points, either turning
points or bifurcation points on this primary branch. In
\cite{buffoni2000sub}, Buffoni, Dancer and Toland showed that for each
sufficiently large value of the integer $m$, there exists a secondary
bifurcation branch of solutions of period $2\pi m$ bifurcating from a
$2\pi$-periodic solution near the highest wave. These results build on
earlier work of Chen and Saffman \cite{chen1980numerical}, who
computed subharmonic bifurcations corresponding to $m=2$ and $m=3$.
Zufiria computed a sequence of 3 bifurcations within an $m=6$
framework, the third being a symmetry-breaking bifurcation that leads
to a branch of non-symmetric traveling gravity waves \cite{zufiria87}.
Vanden-Broeck
\cite{vanden2014periodic} further extended Chen and Saffman's results
to $m = 9$ and provided numerical evidence that the bifurcated
solutions approach non-periodic waves when $m$ approaches infinity.
Our goal in this paper is to study such secondary bifurcations with
infinite spatial period, which requires new techniques beyond the
usual setting of periodic waves.

In \cite{bridges1996spatially}, as an alternative to imposing periodic
boundary conditions, Bridges and Dias propose a quasi-periodic (QP)
framework to study weakly nonlinear traveling solutions of the
gravity-capillary water wave problem. The solution of the linearized
problem can be written as a superposition of two cosine waves whose
wave numbers $k_1$ and $k_2$ are both solutions of the dispersion
relation
\begin{equation}\label{eq:disp:rel}
  c^2=\frac gk + \tau k.
\end{equation}
Here $g$ is the acceleration of gravity, $\tau$ is the surface tension
coefficient, and $c$ is the wave speed, which must be the same for
$k=k_1$ and $k=k_2$. This is a generalization of the Wilton ripple
problem \cite{akers2012wilton, trichtchenko:16, wilton1915} to the
case that $k_1$ and $k_2$ are irrationally related.  In
\cite{quasi:trav}, the present authors propose a conformal mapping
approach \cite{tanveer91, dyachenko1996analytical, dyachenko1996nonlinear,
  choi1999exact, dyachenko2001dynamics, ruban:2005,
  dyachenko2016branch} to generalize the QP framework of Bridges and
Dias to the fully nonlinear water wave regime, numerically confirming
the existence of spatially quasi-periodic traveling water waves. The
conformal map simplifies the computation of the Dirichlet-Neumann
operator to a QP variant of the Hilbert transform.

The solutions in the two-parameter family computed in
\cite{quasi:trav} persist to zero amplitude as QP waves, where they
include both branches of the dispersion relation \eqref{eq:disp:rel}
as special cases. The left and right branches are classified by
Djordjevic and Redekopp \cite{djordjevic77} as gravity waves and
capillary waves, respectively. At the scale of gravity waves in the
ocean, the wave number of capillary waves is typically $10^7$ times
larger than that of gravity waves. Such a large wave number ratio is
computationally out of reach in our framework; moreover, we would not
expect to find interesting nonlinear interactions between waves
of such vastly different length scales.  Thus, if we wish to find QP
traveling waves that might be found in the ocean, we must look beyond
families of solutions that persist to zero-amplitude and consider the
secondary bifurcation problem. In this regime, one may as well set
$\tau=0$, which motivates our study of QP solutions of the pure
gravity wave problem. For simplicity, and for comparison to the
gravity-capillary waves found in \cite{quasi:trav}, we focus on the
case $k_2/k_1=1/\sqrt2$.

A useful feature of the conformal mapping framework is the possibility
of studying overhanging waves \cite{dyachenko:newell:16}, where the
wave profile is not the graph of a single-valued function.  In
\cite{quasi:ivp}, the present authors consider the spatially
quasi-periodic initial value problem for water waves and compute a
solution that begins at $t=0$ with a periodic wave profile but a QP
velocity potential. Each of the infinite number of wave crests evolves
differently as time advances, with some waves overturning and others
flattening out. We now pose the question of whether overhanging
quasi-periodic traveling water waves exist. Clearly surface tension or
hydroelastic forces will be needed to balance the force of gravity in
a steady, overhanging state. Crapper \cite{crapper} discovered a
family of exact overhanging periodic traveling solutions of the pure
capillary wave problem (with no gravity). Kinnersley \cite{kinnersley76}
  found analogues of these exact solutions in finite depth, expressed in
  terms of elliptic functions, again for the pure capillary wave problem.
Schwartz and Vanden-Broeck
\cite{schwartz:79} computed and classified several families of
traveling gravity-capillary waves, providing several examples of
overhanging waves.  Guyenne and Parau \cite{guyenne:12} and Wang
\emph{et~al.} \cite{wang2013two} computed overhanging traveling
solutions of the flexural gravity wave problem of an ice sheet over
deep water. Akers, Ambrose and Wright \cite{akers2014gravity} show that
Crapper's pure capillary wave solution can be perturbed to account for
gravity. In the present work, we search for QP bifurcations from the
two-parameter family of periodic waves referred to by Schwartz and
Vanden-Broack as ``type 1 waves.'' We obtain a new two-parameter
family of QP gravity-capillary waves. The largest-amplitude waves in
this family exhibit an infinite, non-repeating pattern of erratically-spaced
overhanging waves, each with a different shape.

Because QP waves are represented by periodic functions on
higher-dimen\-sional tori, computing large-amplitude QP traveling
waves is a high-dimen\-sional nonlinear optimization problem.  We
build on the basic framework of Wilkening and Yu
\cite{wilkening2012overdetermined}, who formulated the search for
standing waves as an overdetermined nonlinear least squares problem.
In the present setting of QP traveling waves, we optimize over a
two-dimensional array of Fourier coefficients to represent a
one-dimensional QP function. The 1D water wave equations are imposed
in the characteristic direction $(k_1,k_2)$ at each point of a uniform
grid overlaid on the two-dimensional QP torus in physical (as opposed
  to Fourier) space. The problem is overdetermined since we zero-pad
the Fourier representation of the solution so that the nonlinear least
squares solver only has access to the lower-frequency modes of
the solution. This reduces aliasing errors and improves the efficiency
of the computation by reducing the number of degrees of freedom for a
given grid spacing. We wrote a custom Levenberg-Marquardt solver that
employs ScaLapack on a supercomputer to carry out the linear least
squares problems that govern the trust region search steps
\cite{nocedal}.

We also present a new test function\footnote[1]{Here we use
  terminology from the computational dynamical systems literature
  \cite{bindel14,kuznetsov:book}, where a test function changes sign
  at a simple bifurcation; it is not related to test functions from
  the theory of distributions.}  for locating bifurcation points in
finite or infinite dimensional equilibrium problems. We combine the best
features of the singular value decomposition (SVD) approach
\cite{chow:svd:88, shen:bif:97}; the Jacobian determinant approach
\cite{kuznetsov:book, allgower2012numerical, bindel14, pde2path}; the
minimally augmented systems approach \cite{griewank:84, allgower:97,
  kuznetsov:book, govaerts:book}; and the continuation of invariant
subspaces (CIS) approach \cite{demmel01,dieci01,friedman01,beyn01,
  bindel08,bindel14}.  To find branch points in equilibrium problems,
say $f(q)=0$, where $q=(u,s)$ with $u\in\mbb R^n$ and $s\in\mbb R$,
one searches for parameter values $s$ on a path $\mfq(s)$ of solutions
at which the dimension of the kernel of the Jacobian $\mc
J(s)=f_q(\mfq(s))$ increases by one. Often one defines an augmented
(or extended) Jacobian, $\mc J^e(s)$, consisting of $\mc J(s)$ with an
extra row consisting of a multiple of the transpose of $\mfq'(s)$,
which is tangent to the primary branch of solutions.  In bifurcation
problems arising in low-dimensional dynamical systems, the determinant
of $\mc J^e(s)$ can be used as a test function that changes sign at
simple bifurcations. In this case, the method presented here has
  the same theoretical justification; see, e.g., \cite{allgower2012numerical}.
In high dimensions, e.g.~after discretizing a
continuous problem, it is preferable for many reasons (discussed
  in Section~\ref{sec:detect:bif}) to search for
zeros of the smallest singular value $\sigma_\text{min}(s)$ of $\mc
J^e(s)$ rather than using the determinant to locate bifurcation
points.  But singular values are usually computed as non-negative
numbers, which leads to slope discontinuities in
$\sigma_\text{min}(s)$ at its zeros and precludes the use of
root-bracketing methods. As shown by Shen \cite{shen:bif:97}, it is
still possible to devise a Newton-type method, but this involves
approximations of the second derivative operator, $f_{qq}$, applied in
certain directions using finite difference approximations, which we
aim to avoid.

We instead take advantage of the existence of an analytic singular
value decomposition \cite{kato,mehrmann:svd}, or ASVD,
$\mc J^e(s)=U(s)\Sigma(s)V(s)^T$, where the diagonal entries of
$\Sigma(s)$, denoted $\sigma_i(s)$, can change sign and do not
necessarily remain monotonically ordered. A smooth SVD \cite{dieci99}
is sufficient if $\mc J^e(s)$ is not analytic. Sign changes in
$\sigma_\text{min}(s)$ eliminate the slope discontinuities at its
zeros. The standard SVD transfers negative signs in
$\sigma_\text{min}(s)$ to the corresponding column of $U(s)$ or row of
$V(s)^T$, which changes the sign of one of their determinants.  We can
recover the signed version by defining $\chi(s)=\det(U(s))\det(V(s))
\sigma_\text{min}(s) = (\sgn\det \mc J^e(s))\sigma_\text{min}(s)$ as a
test function whose zeros coincide with bifurcation points, where now
the standard SVD with positive singular values is being used. This
opens the door to using quadratically convergent derivative free
methods such as Brent's method \cite{brent:73} to locate the zeros of
$\chi(s)$. We show that once enough Fourier modes are employed to
resolve both the underlying periodic traveling wave and the left and
right singular vectors corresponding to $\sigma_\text{min}(s)$, then
$\chi(s)$ becomes independent of $N$, the Fourier cutoff index. Other
methods we are familiar with make stronger use of the
finite-dimensional status of the truncated problem, leading to
test functions with no infinite-dimensional limit. Having a
globally defined test function $\chi(s)$ that does not change on
refining the mesh is particularly useful in multi-parameter problems,
with $s\in\mbb{R}^d$. We demonstrate this with $d=2$ in
Section~\ref{sec:num:cap}.

Another effective method of locating bifurcations is to border
$\mc J^e(s)$ with a carefully chosen additional row and column to obtain a
matrix $\mc J^{ee}(s)$ and solve
\begin{equation}\label{eq:J:ee}
  \mc J^{ee}(s)\begin{pmatrix} r \\ \psi \end{pmatrix} =
  \begin{pmatrix} 0 \\ 1 \end{pmatrix}, \qquad
  r,0\in\mbb R^{n+1}, \quad \psi,1\in\mbb R.
\end{equation}
The scalar function $\psi(s)$ can then be used as a test function
whose zeros coincide with the desired branch points \cite{griewank:84,
  allgower:97, beyn01, bindel14, kuznetsov:book, govaerts:book}.  We
discuss this approach further in Section~\ref{sec:detect:bif} and
compare the relative merits of $\chi(s)$ and $\psi(s)$, one being that
$\psi(s)$ will change discontinuously if the mesh is refined
adaptively as $s$ changes while $\chi(s)$ will not.

To reduce the cost of searching for quasi-periodic bifurcation points,
we take advantage of Bloch-Fourier theory for diagonalizing linear
operators over periodic potentials \cite{kittel:book}. This technique
has proved useful for studying subharmonic stability of water waves
\cite{longuet:78, mclean:82, mackay:86, oliveras:11, tiron2012linear,
  trichtchenko:16, murashige:20}, but requires reformulation to fit
with our quasi-periodic torus framework. Decomposing the Fr\'echet
derivative of the traveling water wave equations in the QP torus
representation into a direct sum of Bloch-Fourier operators allows us
to focus on a single Bloch frequency when searching for bifurcations.
The dimension of the restricted operator corresponds to points in a
one-dimensional subset of the two-dimensional Fourier lattice, which
makes it possible to locate QP bifurcations from very large-amplitude
periodic traveling waves. Tracking the bifurcation curves beyond
linearization about traveling waves then brings in the full 2D array
of Fourier modes for the torus representation of the solution.

This paper is organized as follows. In Section~\ref{sec:quasi:prelim}
we review the equations governing spatially quasi-periodic traveling
water waves \cite{quasi:trav} and describe the numerical continuation
algorithm we use to compute both periodic and quasi-periodic traveling
waves. In Section~\ref{sec:bif}, we introduce spaces of real-analytic
torus functions, work out the Bloch-Fourier theory of quasi-periodic
bifurcations from traveling waves, define and analyze the test
function $\chi(s)$ for identifying bifurcation points, and show how to
compute the sign of the determinant of a matrix efficiently along
with its singular values.  In Section~\ref{sec:num}, we present
numerical results for the QP gravity wave problem and study a
two-parameter bifurcation problem leading to examples of overhanging
QP traveling gravity-capillary waves. In the appendices, we discuss
the effects of floating-point arithmetic on $\chi(s)$ and prove that
the Fr\'echet derivative for this problem is a bounded operator
between spaces of real analytic torus functions when the parameters of
these spaces are chosen appropriately. Concluding remarks are given in
Section~\ref{sec:conclusion}.

\section{Spatially Quasi-periodic Water Waves}
\label{sec:quasi:prelim}

\subsection{Governing equations for traveling waves}
\label{sec:gov:eqs}
We study the problem of traveling gravity-capillary water waves over a
two-dimensional, irrotational, inviscid fluid of infinite depth.  We
adopt a conformal mapping formulation \cite{quasi:trav, tanveer91, choi1999exact,
  sergey:I, dyachenko2001dynamics} of the problem; specifically, we
consider a conformal map
\begin{equation}
\td z(w) = \td x(w) + i\td y(w), \qquad \qquad w = \alpha + i\beta
\end{equation}
that maps the lower half plane 
\begin{equation}
  \mbb C^- := \{\alpha+i\beta: \quad \alpha\in\mbb R,\quad \beta<0\}
\end{equation}
to the fluid domain in physical space. Here time has been frozen at
$t=0$ and dropped from the notation, and we place a tilde over
functions defined on the real line to simplify the notation for
higher-dimensional torus representation of quasi-periodic
functions. This is the opposite convention of
\cite{quasi:trav,quasi:ivp}, but seems more natural in hindsight. When
the free surface is single-valued in physical space, the fluid domain
has the form
\begin{equation}
\Omega:= \{(x, y): -\infty < y < \td\eta^\phys(x),  \quad x \in \mbb R\},
\end{equation}
where $\td\eta^\phys$ is the free surface elevation. 
To fix the map, we assume that $\td z$ satisfies
\begin{equation}\label{eq:assump:z}
  \lim_{\beta\to-\infty} \td z_w = 1
  \qquad \text{and} \qquad \td z(0) = 0.
\end{equation}
We also assume that $\td z(w)$ can be extended continuously to
$\overline{\mbb C^-}$ and maps the real line $\beta = 0$ to the free
surface.  We introduce the notation $\td\zeta=\td z\vert_{\beta=0}$,
$\td\xi=\td x\vert_{\beta=0}$ and $\td\eta=\td y\vert_{\beta=0}$ so
that the free surface is parameterized by
\begin{equation}\label{eq:zeta:xi:eta}
  \td\zeta (\alpha) = \td\xi(\alpha) + i \td\eta(\alpha),
  \qquad \td\eta(\alpha) = \td\eta^\phys(\td\xi(\alpha)),
  \qquad \alpha\in\mathbb R.
\end{equation}
If the free-surface is not single-valued in physical space, one may
drop the condition that $\td\eta(\alpha) = \td\eta^\phys(\td\xi(\alpha))$ and
simply require that $\td\zeta(\alpha)$ does not self-intersect; see
\cite{quasi:trav}.

In this paper, we focus on the cases where $\td\eta(\alpha)$ is
periodic or quasi-periodic with two quasi-periods. As defined in
\cite{moser1966theory, dynnikov2005topology}, such a function
$\td\eta$ is of the form
\begin{equation}\label{eq:quasi_eta}
  \td\eta(\alpha) = \eta(k_1\alpha, k_2\alpha), \qquad 
  \eta(\alpha_1, \alpha_2) = \sum_{(j_1, j_2)\in\mathbb{Z}^2}
  \hat{\eta}_{j_1, j_2} e^{i(j_1\alpha_1 + j_2\alpha_2)},
\end{equation}
where $\eta$ is defined on the torus $\mathbb{T}^2 =
\mathbb{R}^2/(2\pi\mbb Z)^2$.  After non-dimensionalization, we may
assume that the two basic frequencies of $\td\eta$ are
\begin{equation}
  k_1=1, \qquad\quad k_2=k,
\end{equation}
where $k$ is irrational.  We refer to such functions $\eta$ as
\emph{torus functions} and $\td\eta$ as having been \emph{extracted}
or \emph{reconstructed} from $\eta$.  One can observe that the form
(\ref{eq:quasi_eta}) still applies when $\td\eta$ is periodic; in this
case, the corresponding torus function satisfies
\begin{equation}\label{eq:periodic:eta}
  \eta(\alpha_1, \alpha_2) = \td\eta(\alpha_1), \qquad \qquad
  \alpha_1, \alpha_2 \in \mbb T.
\end{equation}
This allows us to use $\eta(\alpha_1,\alpha_2)$ to represent both
quasi-periodic and periodic functions $\td\eta(\alpha)$.

In \cite{quasi:trav}, quasi-periodic traveling gravity-capillary waves
on deep water are formulated in terms of $\eta(\alpha_1,\alpha_2)$.
The governing equations for $\eta$ read
\begin{equation}\label{eq:govern}
\begin{gathered}
  P\left[\frac{b}{2J} + g\eta -\tau\kappa\right] = 0,
  \qquad b = c^2, \qquad \xi = H[\eta], \\[3pt]
  J = (1+\partial_\alpha\xi)^2 + (\partial_\alpha\eta)^2,
  \qquad
  \kappa = \frac{(1+\partial_\alpha\xi)(\partial_\alpha^2\eta) 
- (\partial_\alpha\eta)(\partial_\alpha^2\xi)}
{J^{3/2}},
\end{gathered}
\end{equation}
where $c$ is the wave speed; $g$ is the gravitational acceleration;
$\tau$ is the surface tension coefficient; and $\xi$, $J$ and $\kappa$
are auxiliary torus functions representing the quasi-periodic part of
the horizontal parameterization of the free surface, the square of the
arclength element, and the curvature, respectively. For gravity waves,
$\tau$ is zero.  The operators $P$, $H$ and $\partial_\alpha$ are
defined by
\begin{equation}\label{eq:hilbert}
\begin{gathered}
  P = \operatorname{id} - P_0, \qquad 
  P_0[f] 
  =  \frac{1}{(2\pi)^2}\int_{\mbb{T}^2}f(
    \alpha_1, \alpha_2)\,d\alpha_1\,d\alpha_2, \\
  H[f](\alpha_1, \alpha_2) = \sum_{j_1, j_2\in\mathbb{Z}}
  (-i)\sgn(j_1 + k j_2) \hat{f}_{j_1, j_2} e^{i(j_1\alpha_1+ j_2\alpha_2)}, \\
  \partial_\alpha f(\alpha_1, \alpha_2) =
          (\partial_{\alpha_1} + k\partial_{\alpha_2})f(\alpha_1, \alpha_2).
\end{gathered}
\end{equation}
Here $\partial_\alpha = (1, k)^T \cdot \nabla$ is the directional
derivative in the $(1,k)$ direction on the torus; $H$ is the
``quasi-periodic Hilbert transform'' \cite{quasi:trav}; and
\begin{equation}\label{eq:sgn:def}
  \sgn(a) = \begin{cases} 1, & a>0, \\[-3pt]
    0, & a=0, \\[-3pt]
    -1, & a<0. \end{cases}
\end{equation}
Note that $\partial_\alpha$ and $H$ act on torus functions in such a
way that extracting the 1d function from the result is equivalent to
first extracting the function and then applying the 1d derivative or
Hilbert transform operators:
\begin{equation}
  \begin{aligned}
  (\pa_\alpha f)(\alpha,k\alpha) &= \pa_\alpha\big[ f(\alpha,k\alpha)\big], \\[-8pt]
  (Hf)(\alpha,k\alpha) &= H\big[f(\cdot,k\cdot)\big](\alpha) =
  \frac1\pi\,PV\!\int_{-\infty}^\infty \frac{f(\beta,k\beta)}{\alpha-\beta}\, d\beta.
  \end{aligned}
\end{equation}
Both $\pa_\alpha$ and $H$ are diagonal in 2D Fourier space, with
Fourier multipliers
\begin{equation}\label{eq:da:H:fmult}
  \widehat{\pa_{j_1, j_2}\!\!}\, = i(j_1 + k j_2), \qquad\quad
  \hat H_{j_1, j_2} = (-i) \sgn(j_1 + k j_2),
\end{equation}
respectively. Different choices of $k$ lead to different operators.

One can check that if $\eta(\alpha_1, \alpha_2)$ is a solution of
(\ref{eq:govern}), then $\eta(-\alpha_1, -\alpha_2)$ is also a
solution.  In this paper we will focus on real-valued traveling
solutions with even symmetry: $\eta(-\alpha_1,
  -\alpha_2)=\eta(\alpha_1, \alpha_2)$. Equivalently, we assume the
Fourier coefficients of $\eta$ satisfy
\begin{equation}\label{eq:assump:eta}
\hat{\eta}_{j_1, j_2} = \hat{\eta}_{-j_1, -j_2} = \overline{\hat{\eta}}_{j_1, j_2},
\qquad \qquad j_1, j_2 \in \mbb{Z}. 
\end{equation}
Since adding a constant to $\eta$ will not change (\ref{eq:govern}),
we assume $P_0[\eta] = 0$ when computing traveling waves.  Under these
assumptions, we reconstruct $\td\xi$ in (\ref{eq:zeta:xi:eta}) from
$\xi=H[\eta]$ using
\begin{equation}\label{eq:xi:def}
\td\xi(\alpha) = \alpha + \xi(\alpha, k\alpha),
\end{equation}
which is an odd function. For most torus functions, adding a tilde
denotes evaluation at $(\alpha,k\alpha)$, but we treat $\xi$ as a
special case and include the linear growth term $\alpha$ in
\eqref{eq:xi:def}. This is why we refer to $\xi$ as the quasi-periodic
part of the horizontal parameterization.

\begin{remark}\label{rmk:hat:eta00}
  It is preferable to report solutions with zero mean in physical
  space rather than in conformal space. Let us briefly use a
  superscript 0 to denote a traveling wave with the above properties,
  which satisfies
\begin{equation}
  \hat{\eta}^\e0_{0,0} = 0.
\end{equation}
The desired solution only requires adjusting the $(0,0)$ Fourier mode:
\begin{equation}
  \hat\eta_{0,0} = -P_0\big[ (\eta^\e0)(1+\xi_\alpha^\e0)\big], \qquad
  \hat\eta_{j_1,j_2} = \hat\eta^\e0_{j_1,j_2},\qquad (j_1,j_2)\ne(0,0),
\end{equation}
where $\lim_{a\to\infty}\frac1{2a}\int_{-a}^a
\td\eta(\alpha)\td\xi_\alpha(\alpha)\,d\alpha=0$ is the zero mean
condition, and we make use of $\td\xi^\e0_\alpha=(1+\xi_\alpha^\e0)$,
from \eqref{eq:xi:def}. Thus, we may assume $\hat\eta_{0,0}=0$ when
computing periodic waves, quasi-periodic waves, and the bifurcation
points where they meet; we can then adjust the mean of each wave
computed as a simple post-processing step.
\end{remark}

\begin{remark} \label{rmk:symmetry}
One can check that if $\eta(\alpha_1, \alpha_2) = \sum_{j_1,
  j_2\in\mbb{Z}}\hat{\eta}_{j_1, j_2} e^{i(j_1\alpha_1 +
    j_2\alpha_2)}$ is a real-valued solution of (\ref{eq:govern}) with
even symmetry, then the following three functions are also real-valued
solutions with even symmetry:
\begin{equation}\label{eq:sym:trans}
\begin{gathered}
  \eta(\alpha_1 + \pi, \alpha_2) = \sum_{j_1, j_2\in\mbb{Z}}
   (-1)^{j_1}\hat{\eta}_{j_1, j_2} e^{i(j_1\alpha_1 + j_2\alpha_2)}, \\
   \eta(\alpha_1, \alpha_2 + \pi) = \sum_{j_1, j_2\in\mbb{Z}}
   (-1)^{j_2}\hat{\eta}_{j_1, j_2} e^{i(j_1\alpha_1 + j_2\alpha_2)},\\
   \eta(\alpha_1 + \pi, \alpha_2 + \pi) =
   \sum_{j_1, j_2\in\mbb{Z}}(-1)^{j_1 + j_2}\hat{\eta}_{j_1, j_2}
   e^{i(j_1\alpha_1 + j_2\alpha_2)}.
\end{gathered}
\end{equation}
\end{remark}

\subsection{Numerical Algorithm}
\label{sec:num:alg}

Following \cite{quasi:trav}, we formulate (\ref{eq:govern}) as a
nonlinear least-squares problem and define objective and residual
functions
\begin{equation}\label{eq:FR:def}
  \mc{F}[\eta, \tau, b] = \frac{1}{8\pi^2}
  \int_{\mbb{T}^2} \mc{R}^2[\eta,\tau,b]\,d\alpha_1\, d \alpha_2,
  \qquad
  \mc{R}[\eta,\tau,b] = P\left[\frac{b}{2J} + g\eta
    -\tau\kappa\right].
\end{equation}
We use square brackets for the functional $\mc F$ and operator $\mc R$
so that $\mc R(\alpha_1,\alpha_2)$ can be short-hand for $\mc
R[\eta,\tau,b](\alpha_1,\alpha_2)$.  We represent a torus function $f$
in two ways, either through its values on a uniform $M_1\times M_2$
grid on $\mbb T^2$, or via the fast Fourier transform coefficients of
these sampled values:
\begin{equation}\label{eq:fft:2d}
  \begin{aligned}
    f_{m_1, m_2} &= f(2\pi m_1/M_1, 2\pi m_2/M_2)
    = \sum_{j_2=0}^{M_2-1}\sum_{j_1=0}^{M_1-1} \breve f_{j_1,j_2}
    e^{2\pi i(j_1m_1/M_1 + j_2m_2/M_2)}, \\
    \breve{f}_{j_1, j_2} &=
    \frac{1}{M_1M_2} \sum_{m_2 = 0}^{M_2-1} \sum_{m_1 = 0}^{M_1-1}
    f_{m_1, m_2} e^{-2\pi i (j_1m_1/M_1 + j_2m_2/M_2)} =
    \sum_{n_1,n_2\in\mbb Z} \hat f_{j_1+n_1M_1,j_2+n_2M_2}.
  \end{aligned}
\end{equation}
Here $\hat f_{j_1,j_2}=(4\pi^2)^{-1}\iint_{\mbb T^2}
f(\alpha_1,\alpha_2)e^{-i(j_1\alpha_1+j_2\alpha_2)}d\alpha_1\,d\alpha_2$
are the actual Fourier modes of $f(\alpha_1,\alpha_2)$, which are
related to the FFT modes $\breve f_{j_1,j_2}$ by the above aliasing
formula. We only store the values of the periodic arrays
$\{f_{m_1,m_2}\}$ and $\{\breve f_{j_1,j_2}\}$ with indices
\begin{equation}\label{eq:mj:ranges}
  0\le m_1<M_1, \qquad
  0\le m_2<M_2, \qquad
  0\le j_1\le M_1/2, \qquad
  0\le j_2<M_2,
\end{equation}
where we take advantage of $\breve f_{-j_1,-j_2}=\overline{\breve
  f_{j_1,j_2}}$ when $f(\alpha_1,\alpha_2)$ is real-valued to avoid
having to store modes with index $j_1<0$. We assume $M_1$ and
$M_2$ are sufficiently large and $|\hat f_{j_1,j_2}|$
decays sufficiently fast as $|j_1|+|j_2|\to\infty$ that
\begin{equation}\label{eq:dft:approx}
  \breve f_{j_1,j_2} \approx \begin{cases}
    \hat f_{j_1,j_2} & 0\le j_2\le M_2/2, \\
    \hat f_{j_1,j_2-M_2} & M_2/2<j_2<M_2.
  \end{cases}
\end{equation}
When evaluating $\mc R[\eta,\tau,b]$, we only vary $b$, $\tau$
and the leading Fourier coefficients of $\eta$,
\begin{equation}\label{eq:N1N2:range}
\hat{\eta}_{j_1, j_2}, \qquad 
\left(-N_1\leq j_1\leq N_1,  \quad -N_2\leq j_2\leq N_2\right),
\end{equation}
where $N_1$ and $N_2$ are cutoff frequencies typically taken to be
around $M_1/3$ and $M_2/3$, respectively. The higher-frequency Fourier
modes $\hat\eta_{j_1,j_2}$ with $|j_1|>N_1$ or $|j_2|>N_2$ are
set to zero. This means that the FFT modes $\breve\eta_{j_1,j_2}$ in the
range \eqref{eq:mj:ranges} with $j_1>N_1$ or $N_2<j_2<M_2-N_2$ are set
to zero.  Since $\breve\eta_{0,0}$ is also set to zero at this stage
of the computation (and later adjusted via
  Remark~\ref{rmk:hat:eta00}), and since $\eta$ is real-valued and
even, satisfying \eqref{eq:assump:eta}, the number of independent
FFT coefficients $\breve\eta_{j_1,j_2}$ is
\begin{equation}\label{eq:N:tot}
  N_{\text{tot}} = N_1(2N_2+1) + N_2.
\end{equation}

\begin{remark}\label{rmk:hat:breve}
For simplicity, from now on we focus on the Fourier modes of the torus
functions that arise (e.g., $\hat\eta_{j_1,j_2}$), with the
understanding that in the numerical implementation they map to FFT
modes (e.g., $\breve\eta_{j_1,j_2}$) with indices in the ranges
\eqref{eq:mj:ranges} via the assumption \eqref{eq:dft:approx}.
\end{remark}

We evaluate $R[\eta,\tau,b](\alpha_1,\alpha_2)$ on the $M_1\times M_2$
grid. Using the trapezoidal rule on the integral \eqref{eq:FR:def},
which is a spectrally accurate approximation, we obtain an
overdetermined nonlinear least-squares problem from $\mbb
R^{N_\text{tot}}$ to $\mbb R^{M_1M_2}$, namely
\begin{equation}\label{eq:min:F}
  \text{minimize} \; F[\{\hat\eta_{j_1,j_2}\};\tau,b] = \frac12r^Tr, \qquad
  r_{m_2M_1+m_1} = \frac{\mc R(2\pi m_1/M_1,2\pi m_2/M_2)}{\sqrt{M_1M_2}}.
\end{equation}
Here we have written the objective function to suggest that $\tau$ and
$b=c^2$ are prescribed parameters and the independent Fourier modes
$\{\hat\eta_{j_1,j_2}\}$ are the unknowns over which the objective
function is minimized. For small-amplitude traveling waves, as
explained in \cite{quasi:trav}, it is better to prescribe two Fourier
modes, say $\hat\eta_{1,0}$ and $\hat\eta_{0,1}$, and include $\tau$
and $b$ among the unknowns to be determined by solving
\eqref{eq:min:F}.  In the present work, as explained below, we study
large-amplitude waves and use a hybrid choice in which $\tau$ and
either $\hat\eta_{1,0}$, $\hat\eta_{0,1}$ or the wave height (defined
  below) are prescribed while $b$ and the other Fourier modes of
$\eta$ are found by the solver.

One of the examples presented in \cite{quasi:trav} involves computing
a two-parameter family of quasi-periodic traveling waves with
$k=1/\sqrt2$ held fixed and $\hat\eta_{1,0}$ and $\hat\eta_{0,1}$ prescribed
to vary over the interval $I=[-0.01,0.01]$.  Each of these amplitude
parameters is assigned values on a 16-point Chebyshev-Lobatto grid
over $I$, and polynomial interpolation is used to express
the surface tension coefficient $\tau$ in the form
\begin{equation}\label{eq:tau:c:expand}
  \tau(\hat\eta_{1,0},\hat\eta_{0,1})
  = \sum_{m=0}^{15}\sum_{n=0}^{15}
  \check\tau_{mn}
  T_m(100\hat\eta_{1,0})T_n(100\hat\eta_{0,1}),
  \qquad\quad
  (\hat\eta_{1,0},\hat\eta_{0,1})\in I^2,
\end{equation}
where $\{T_m(x)\}_{m=0}^\infty$ are the Chebyshev polynomials.
The wave speed $c$ is similarly interpolated.  The tensor product
Chebyshev coefficients $\check\tau_{mn}$ and $\check c_{mn}$ are found
to decay below $10^{-15}$ in amplitude for $m+n\ge10$, suggesting that
\eqref{eq:tau:c:expand} is accurate to double-precision accuracy
throughout the parameter region $I^2$. Setting both amplitude
parameters to zero gives $\tau=\tau_\text{lin}=g/(k_1k_2)$ and
$c^2=c_\text{lin}^2 = (k_1+k_2)\tau$, as predicted by linear theory,
with $g=1$, $k_1=1$, $k_2=k=1/\sqrt2$ in this case.

The above approach sidesteps the difficulty of finding bifurcation
points within the family of periodic traveling waves. Instead, all the
waves in the polynomial interpolation leading to
\eqref{eq:tau:c:expand} are genuinely quasi-periodic, with non-zero
values of both $\hat\eta_{1,0}$ and $\hat\eta_{0,1}$. Indeed, by using
an even number of Chebyshev-Lobatto nodes, 0 is not among the
interpolation points in either direction. After the expansion
coefficients $\check\tau_{mn}$ or $\check c_{mn}$ have been identified
via interpolation, we can set $\hat\eta_{1,0}$ or $\hat\eta_{0,1}$ to
zero to find the curves $\tau(\hat\eta_{1,0},0)$ or
$\tau(0,\hat\eta_{0,1})$ where a bifurcation exists from periodic
traveling waves of wave number 1 or $k=1/\sqrt2$ to quasi-periodic
waves of wave numbers $\vec k=(1,k)$. These curves are the
intersection of a two-parameter family of quasi-periodic solutions
with a two-parameter family of periodic solutions.

In the present paper, we address the difficulty sidestepped above.  We
begin by computing families of large-amplitude periodic traveling
waves, which can be parameterized by surface tension and one amplitude
parameter.  We may assume without loss of generality that
$\eta(\alpha_1,\alpha_2)=\td\eta(\alpha_1)$ is independent of
$\alpha_2$, as in \eqref{eq:periodic:eta}. All the Fourier modes
$\hat\eta_{j_1,j_2}$ with $j_2\ne0$ are then zero, so we may drop
the $j_2$ subscript and view $\{\hat\eta_j\}_{j\in\mbb Z}$ as the
coefficients of the 1D Fourier expansion of $\td\eta(\alpha)$.
The unknowns for the periodic problem are then
\begin{equation}\label{eq:p:trav:1d}
  \vec p=(b\,,\,\hat\eta_{1}\,,\,\hat\eta_{2}\,,
    \,\hat\eta_{3}\,,\,\cdots\,,\,\hat\eta_{N_1-1}\,,\,\hat\eta_{N_1}),
\end{equation}
where $\hat\eta_{0}$ is set to 0 as discussed in
Remark~\ref{rmk:hat:eta00}, $\hat\eta_{-j}=\hat\eta_{j}$ due to
\eqref{eq:assump:eta}, and $\tau$ is a prescribed parameter.  We also
define a wave amplitude by introducing coefficients $\nu_i$ and
setting
\begin{equation}\label{eq:amp:mu}
  \vec\nu\cdot\vec p = \nu_0b+\sum_{j=1}^{N_1}\nu_jp_j = \mu,
\end{equation}
where $\mu$ is the prescribed amplitude parameter. The two cases we consider
are
\begin{equation}\label{eq:nu:mu}
  \begin{alignedat}{2}
    \text{case 1: } \; \vec\nu &= (0,1,0,\dots,0), &\quad \mu&=\hat\eta_{1}
    \quad (=\hat\eta_{1,0}), \\
    \text{case 2: } \; \nu_j&=0 \; (j\text{ even})\;,\;\; \nu_j=4 \; (j\text{ odd}), &
    \quad \mu &= h =\td\eta(0)-\td\eta(\pi).
  \end{alignedat}
\end{equation}
In case 2, the reason for $\nu_j=4$ when $j$ is odd is that
$\hat\eta_{-j}=\hat\eta_j$  and together they contribute $4\hat\eta_{j}$
to the wave height, $h=\td\eta(0)-\td\eta(\pi)$, when $j$ is odd. Note that
$b$ ($j=0$) and the even modes with $j\ge2$ have no effect on
the wave height due to cancellation. For this periodic sub-problem, both
$\eta(\alpha_1, \alpha_2)$ and $\mc R[\eta,\tau,b](\alpha_1,\alpha_2)$
are independent of $\alpha_2$, so one can simplify \eqref{eq:min:F} to
\begin{equation}\label{eq:min:F:1d}
  \text{minimize} \; F[\vec p;\mu,\tau] = \frac12r^Tr, \qquad
  r_{m} = \begin{cases}
    \frac{\mc R(2\pi m/M_1,0)}{\sqrt{M_1}}, & \; 0\le m< M_1, \\
    \vec\nu\cdot\vec p - \mu, & \; m=M_1,
  \end{cases}
\end{equation}
where we have added an equation to enforce \eqref{eq:amp:mu} and
re-organized the argument list of $F$ to separate the prescribed
parameters from the unknowns.

\begin{remark}\label{rmk:eta10}
  When solutions of the periodic problem are embedded in the 2D torus
  representation, it is often preferable to employ the double-index
  Fourier notation, $\hat\eta_{j_1,j_2}=\hat\eta_{j_1}\delta_{j_2,0}$,
  where $\delta_{ij}$ is the Kronecker delta. The amplitude parameter
  in case 1 of \eqref{eq:nu:mu} is then 
  $\hat\eta_{1,0}$.
\end{remark}

\begin{remark}\label{rmk:drop:mu}
In case 1 of \eqref{eq:nu:mu}, we can alternatively drop the last
component of $\vec r$ in \eqref{eq:min:F:1d} and remove
$\hat\eta_{1}$ from the vector $\vec p$ of unknowns over which the
minimization is performed. This mode is set equal to $\mu$ externally
and not varied by the solver.
\end{remark}

Given an initial guess $(\vec p^0;\mu,\tau)$ for a periodic traveling
wave, we use the Levenberg-Marquardt method \cite{nocedal} to minimize
$F$ in \eqref{eq:min:F:1d} over the unknowns $\vec p$ holding
$(\mu,\tau)$ fixed.  We employ the delayed Jacobian update strategy
proposed by Wilkening and Yu \cite{wilkening2012overdetermined} in the
context of computing standing water waves.  For the initial guess, we
use linear theory to get started on one or several straight-line paths
through parameter space, i.e., through the $(\mu,\tau)$-plane. Once
two solution on such a path have been computed, we use linear
extrapolation for the starting guesses of successive solutions on the
numerical continuation path. We increase $N_1$ and $M_1$ adaptively as
we go to maintain spectral accuracy of the computed traveling waves.

We search for bifurcation points along the numerical continuation path
of periodic traveling waves using the methods of Section~\ref{sec:bif}
below. When a bifurcation branch is found, we follow it using the same
strategy as for periodic traveling waves, but with $\hat\eta_{0,1}$
replacing $\mu$ as the first numerical continuation parameter.  We use
$\tau$ as the second parameter in both cases.  On this branch,
$\hat\eta_{0,1}=0$ corresponds to the periodic traveling wave, and the
list of unknowns, $\vec p$, is expanded to include the modes
$\hat\eta_{j_1,j_2}$ with $j_2\ne0$:
\begin{equation}\label{eq:p:trav:2d}
  \vec p = (b,\{\hat\eta_{j_1,j_2}\}), \qquad
  \left(
    \begin{gathered}
      1\le j_1\le N_1 \\
      -N_2\le j_2\le N_2
    \end{gathered}
    \right) \; \text{or} \;
  \left(
    \begin{gathered}
      j_1=0 \\
      2\le j_2\le N_2
    \end{gathered}
    \right).
\end{equation}
Here we follow the strategy of Remark~\ref{rmk:drop:mu} and remove
$\hat\eta_{0,1}$ from the list $\vec p$ rather than add a component to
$\vec r$ to govern the amplitude.  Since $b$ has replaced
$\hat\eta_{0,1}$ in the list, we see that in the quasi-periodic
problem, the number of degrees of freedom of the nonlinear least
squares problem \eqref{eq:min:F} is $N_\text{tot}$ from
\eqref{eq:N:tot}. For the numerical continuation path, we hold $\tau$
fixed with its value at the bifurcation point and vary
$\hat\eta_{0,1}$ with progressively larger values, increasing $N_1$,
$N_2$, $M_1$ and $M_2$ as needed to maintain spectral accuracy. We
stop when we run out of computational resources to further increase
the problem size. Details will be given in Sections~\ref{sec:num:grav}
and~\ref{sec:num:cap} below.

The Levenberg-Marquardt algorithm requires the evaluation of the
Jacobian $\mc J_{ij} = \pa r_i/\pa p_j$, which can be carried out
analytically or with finite differences.  We take the analytical
approach. Let us denote the Fr\'echet derivative of $\mc{R}$ by
\begin{equation}\label{eq:R:frechet}
  D_q\mc{R} = (D_\eta\mc{R}, D_\tau\mc{R}, D_b\mc{R}), \qquad\quad
  q = (\eta,\tau,b),
\end{equation}
and employ ``dot notation'' \cite{benj1,wilkening2012overdetermined}
for the variational derivative of a quantity at $q$ in the $\dot q$
direction:
\begin{equation}\label{eq:dot:R:def}
  \dot{ \mc R}(q, \dot q) = D_q\mc{R}[q]\dot q =
  \frac{d}{d\veps}\bigg\vert_{\veps=0}\mc R[q+\veps\dot q].
\end{equation}
We will not use a dot for time derivatives in this paper.
Explicitly, we have
\begin{equation}\label{eq:lin:R}
\begin{gathered}
  \dot {\mc{R}} = 
  P\bigg[ \frac{1}{2J} \dot b - 
    \frac{b}{2J^2} \dot {J} + g \dot{\eta}
    -\tau \dot{\kappa} -\kappa\dot \tau
    \bigg], \\[3pt]
  \dot{\xi} = H\big[\dot{\eta}\big],\qquad \qquad
  \dot{J} = 2\Big\{\big(1+\partial_\alpha
    \xi\big)\big(\partial_\alpha \dot{\xi}\big) + 
  \big(\partial_\alpha \eta\big)\big(
    \partial_\alpha\dot{ \eta}\big)\Big\}, \\[3pt]
  \dot{\kappa} = -\frac{3\kappa}{2J}\, \dot{J} + 
  \frac{1}{J^{3/2}}\Big\{\big(\partial_\alpha^2 \eta\big)
  \big(\partial_\alpha \dot{ \xi}\big)
  + \big(1+\partial_\alpha \xi\big)\big(
    \partial_\alpha^2 \dot{ \eta}\big)
  - \big(\partial_\alpha^2 \xi\big)\big(
    \partial_\alpha \dot{\eta}\big)
  - \big(\partial_\alpha \eta\big)\big(
    \partial_\alpha^2 \dot{\xi}\big)\Big\},
\end{gathered}
\end{equation}
where $\eta$, $\xi$, $J$, $\kappa$, $\dot\eta$, $\dot\xi$, $\dot J$
and $\dot\kappa$ are torus functions; $b$, $\tau$, $\dot b$ and
$\dot\tau$ are scalars; and $\xi$ represents only the quasi-periodic
part of $\td\xi$, via \eqref{eq:xi:def}. With these formulas, it
is easy to evaluate the entries of the Jacobian
\begin{equation}\label{eq:Jij:levmar}
  \mc J_{i,j} = \der{r_i}{p_j} =
  \frac{\dot R(2\pi m_1/M_1,2\pi m_2/M_2)}{\sqrt{M_1M_2}}, \qquad
  0\le i = m_2M_1+m_1 < M_1M_2,
\end{equation}
where $m_1,m_2$ are in the ranges \eqref{eq:mj:ranges} and $j$
enumerates the entries of $\vec p$ in \eqref{eq:p:trav:2d}.  For
example, $j=0$ corresponds to $b$, so one sets $\dot
q=(\dot\eta,\dot\tau,\dot b)=(0,0,1)$ in \eqref{eq:lin:R} to compute
the zeroth column of $\mc J$ via \eqref{eq:Jij:levmar}. Since $\tau$ is
treated as a fixed parameter, we set $\dot\tau=0$ for each column of
the Jacobian in the present work; however, in \cite{quasi:trav},
$\tau$ is computed by the solver, just like $b$, so one of the
Jacobian columns corresponds to $\dot q=(\dot\eta,\dot\tau,\dot
  b)=(0,1,0)$. Each of the remaining columns corresponds to varying
one of the Fourier mode degrees of freedom. Since we make use of the
symmetry \eqref{eq:assump:eta}, these columns correspond to
variations of the form
\begin{equation}\label{eq:q:j1j2}
  \dot q_{j_1,j_2}=\big(
    e^{i(j_1\alpha_1+j_2\alpha_2)}+e^{-i(j_1\alpha_1+j_2\alpha_2)},0,0\big),
\end{equation}
where $j_1$ and $j_2$ range over the values listed in
\eqref{eq:p:trav:2d} to enumerate columns 1 through $(N_\text{tot}-1)$
of $\mc J$ in a zero-based numbering convention.  More details on the
form of $\dot R$ for variations of the form \eqref{eq:q:j1j2} will be
given in Section~\ref{sec:bif} below.  In the periodic sub-problem,
$\mc J_{i,j}$ in \eqref{eq:Jij:levmar} is modified in the obvious way
to account for the change from \eqref{eq:min:F} to \eqref{eq:min:F:1d}
and \eqref{eq:p:trav:2d} to \eqref{eq:p:trav:1d}.

In the process of finding a bifurcation point $q^\bif$ from periodic
to quasi-periodic traveling waves in Section~\ref{sec:bif} below, we
will obtain a null vector $\dot q^\qua$ of $D_q\mc R[q^\bif]$ that is
transverse to the family of traveling waves. In the numerical
continuation algorithm, we take the bifurcation point $q^\bif$ as the
zeroth point on the path. The first point on the path, which, unlike
the zeroth point, will be genuinely quasi-periodic, is obtained using
the Levenberg-Marquardt algorithm with initial guess $q^\bif+\veps\dot
q^\qua$.  Here $\veps$ is a suitably small number that we choose by
trial and error to make progress in escaping the family of periodic
waves while still resembling the zeroth solution. In the minimization,
$\tau$ is held fixed with its value at the zeroth solution and
$\hat\eta_{0,1}$ is held fixed with the value
$\veps\dot{\hat\eta}_{0,1}^\qua$ (since
  $\hat\eta_{0,1}^\bif=0$). After the zeroth and first solution on the
path are computed, we continue along a straight line through parameter
space (the $(\hat\eta_{0,1},\tau)$-plane) using linear extrapolation
for the initial guess for the next quasi-periodic solution. The
straight line involves holding $\tau$ fixed and incrementing
$\hat\eta_{0,1}$ for successive solutions. The increment is initially
$\veps\dot{\hat\eta}_{0,1}$, but can be changed adaptively, if
needed. This numerical continuation strategy is easy to implement and
requires only a few iterations to find quasi-periodic solutions that
deviate significantly from the periodic solution at the bifurcation
point as long as $\veps$ is chosen to be large enough to make progress
along the path (but small enough for linear extrapolation to be
  effective).

\section{Quasi-periodic bifurcations from periodic traveling waves}
\label{sec:bif}
In \cite{quasi:trav}, we compute small-amplitude quasi-periodic
traveling waves that bifurcate from the zero-amplitude wave.
In this section, we consider quasi-periodic bifurcations
from finite-amplitude periodic traveling waves that can be far
beyond the linear regime of the zero solution. In particular,
we wish to study genuinely quasi-periodic traveling waves with zero
surface tension, which do not exist at small amplitude.

Before discussing bifurcation theory, it is convenient to frame the
problem in a Hilbert space setting.  Recall that a torus function
$\eta:\mbb T^2\to\mbb C$ is real-analytic if and only if (iff)
it has a convergent power series in a neighborhood of each
$\vec\alpha\in\mbb T^2$. Equivalently \cite{quasi:ivp,broer:book},
  $\eta$ is real-analytic iff its Fourier modes $\hat \eta_{j_1,j_2}$
  in \eqref{eq:quasi_eta} decay exponentially,
  i.e., there exist positive constants $C$ and $\sigma$ such that
  $|\hat\eta_{j_1,j_2}|\le C e^{-\sigma{(|j_1|+|j_2|)}}$ for all
  $(j_1,j_2)\in\mbb Z^2$. We follow
the standard convention \cite{krantz} that real-analytic functions can
be complex-valued.  Although $\eta$ in \eqref{eq:quasi_eta} must be
real-valued for \eqref{eq:zeta:xi:eta} to make sense, it is useful to
allow complex-valued torus functions when considering the effect of
perturbations in Fourier space. Ultimately, linear combinations will
be taken to keep the result real-valued. Similarly, while $b=c^2$
must be positive and $\tau$ must be non-negative, perturbations of
these quantities can have either sign, and can even be complex as
long as linear combinations are eventually taken to make them real.

\begin{definition}\label{def:V:sig}
  For $\sigma\ge0$, let $\mc V_\sigma$ be the Hilbert space of
  real-analytic torus functions of finite norm induced by the inner
  product
  \begin{equation}\label{eq:sigma:norm}
    \la f,g \ra = \sum_{(j_1,j_2)\,\in\mbb Z^2}
    \hat f_{j_1,j_2} \,\, \overline{\hat g_{j_1,j_2}} \,\, e^{2\sigma(|j_1|+|j_2|)}.
  \end{equation}
  We also define the subspaces
  \begin{equation}
    \begin{aligned}
      \mc V_\sigma^\e{l} &= \{\, f\in\mc V_\sigma\;:\;
      \hat f_{j_1,j_2}=0 \;\; \text{if} \;\; j_2\ne l\,\},
      \qquad (l \in \mbb Z), \\
      \mc V_\sigma^\per &= \mc V_\sigma^\e{0}, \qquad
      \mc V_\sigma^\qua = (\mc V_\sigma^\per)^\perp =
      \jt\bigoplus_{l\ne0}\mc V_\sigma^\e{l}
    \end{aligned}
  \end{equation}
  and write, e.g.,
  $\big(\mc V_\sigma^\e1,0,\mbb C\big)$ and
  $\big(\mc V_\sigma^\per,\mbb C,\mbb C\big)$ as shorthand for
  $\big\{(f,0,b)\,:\, f\in\mc V_\sigma^\e1\,,\,b\in\mbb C\big\}$ 
  and $\mc V_\sigma^\per\times\mbb C^2$, respectively, with the
  product Hilbert space norms.
\end{definition}

Note that $\mc V_\sigma^\per$ consists precisely of the torus
functions $f(\alpha_1,\alpha_2)$ in $\mc V_\sigma$ that do not depend
on $\alpha_2$.  We think of functions in $\mc V_\sigma\setminus\mc
V_\sigma^\per$ as being genuinely quasi-periodic even though this set
includes functions $f(\alpha_1,\alpha_2)$ that are independent of
$\alpha_1$. We adopt this viewpoint as our focus is on bifurcations
from $2\pi$-periodic traveling waves. The case of bifurcations from
$(2\pi/k)$-periodic traveling waves can be investigated within this
framework by rescaling space by a factor of $k$ to make the wavelength
of these waves $2\pi$, and then replacing $k$ by $1/k$ as the second
basic frequency.

\subsection{Linearization about periodic traveling waves}
\label{sec:lin:per}

Recall from \eqref{eq:FR:def}
that the governing equations \eqref{eq:govern} for traveling water
waves are equivalent to solving
\begin{equation}\label{eq:Rq}
  \mc{R}[q] = P\left[\frac{b}{2J} + g\eta
    -\tau\kappa\right] = 0, \qquad \qquad q = (\eta, \tau, b),
\end{equation}
where $J$ and $\kappa$ depend on $\eta$ via \eqref{eq:govern}.  We
computed the Fr\'echet derivative of $\mc R$ in \eqref{eq:lin:R} using
``dot notation,'' defined in \eqref{eq:dot:R:def}. The following
theorem is proved in Appendix~\ref{sec:proof}:

\begin{theorem}\label{thm:bdd}
  Suppose $q^\per=(\eta,\tau,b)$ with $\eta\in\Vp{\sigma}$ for some
  $\sigma>0$.  Suppose also that $\eta$ is real-valued and the
  resulting $J(\alpha_1,\alpha_2)$ in \eqref{eq:govern}, which is
  independent of $\alpha_2$, is non-zero for every $\alpha_1\in\mbb
  T$. Then there exists $\rho\in(0,\sigma)$ such that $D_q\mc
  R[q^\per]$ is a bounded operator from $\big(\mc V_\sigma,\mbb C,\mbb
    C\big)$ to $\mc V_\rho$.
\end{theorem}

We do not assume $q^\per$ is a solution of $\mc R[q]=0$ in this
theorem, though that is the case of interest. Our next goal is to show
that when linearized about a periodic solution, variations in
$(\tau,b)$ and periodic perturbations of $\eta$ lead to periodic
changes in $\mc R$ while quasi-periodic perturbations lead to
quasi-periodic changes in $\mc R$. We combine the discussion of the
numerical computation with the derivation since the only difference is
whether infinite Fourier series of real analytic functions are
considered or whether these functions are approximated via the FFT on
a uniform grid.

Let $q^\per=(\eta,\tau,b)$ satisfy the hypotheses of
Theorem~\eqref{thm:bdd}.  Then clearly
\begin{equation}\label{eq:dot:R:tau:b}
  D_q \mc R \big[q^\per \big]\big(0,\dot \tau, \dot b \big) =
  P\Big[ \dot b/\big(2J\big)  - \kappa \dot \tau
    \Big]
\end{equation}
is periodic, i.e., a torus function independent of $\alpha_2$.
Moreover, if $\dot q$ is of the form
\begin{equation}\label{eq:dot:q:form}
  \dot q^{(l_1, l_2)}= \big(e^{il_1\alpha_1}e^{il_2 \alpha_2},0,0\big),
  \qquad\qquad   l_1,\,  l_2\in\mbb Z,
\end{equation}
then \eqref{eq:lin:R} simplifies to
\begin{equation}\label{eq:lin:l1l2}
  \begin{aligned}
    \dot{\mc R} &= P\bigg[
      -\frac{b}{2J^2} \dot J + g \dot{\eta}
      -\tau \dot{\kappa}\bigg], \quad
    \dot{\eta} = e^{il_1\alpha_1}e^{il_2\alpha_2}, \quad
    \dot{\xi} = -i\opn{sgn}(l_1+kl_2)e^{il_1\alpha_1}e^{il_2\alpha_2}, \\
    \dot{J} &= 2\Big\{\,\big|l_1+kl_2\big|\big( 1 + \pa_\alpha\xi\, \big)
    + i\big(l_1+kl_2\big)\,\pa_\alpha\eta\,\Big\}e^{il_1\alpha_1}e^{il_2\alpha_2}, \\
    \dot{\kappa} &= -\frac{3\kappa}{2J} \dot{J}
    + \frac{1}{J^{3/2}}\Big\{
    \big|l_1+kl_2\big|\pa_\alpha^2\eta - \big(l_1+kl_2\big)^2
    \big( 1 + \pa_\alpha \xi \,\big) \\
    & \hspace*{1.5in}-i\big(l_1+kl_2\big)\big(
      \pa_\alpha^2\xi + \big|l_1+kl_2\big|\,\pa_\alpha\eta\big)\Big\}
    e^{il_1\alpha_1}e^{il_2\alpha_2}.
  \end{aligned}
\end{equation}
The terms in braces, which we denote by $\td A^\e1(\alpha_1)$ and $\td
A^\e2(\alpha_1)$, respectively, are independent of $\alpha_2$. Next we
expand $\td\eta(\alpha)$ in \eqref{eq:periodic:eta} as a 1d Fourier
series, $\sum_j\hat\eta_je^{ij\alpha}$, which gives
\begin{equation}\label{eq:fourier:xi:eta}
  \big(1+\pa_\alpha\xi\,\big)=1+\sum_{j}|j|\hat\eta_{j}e^{ij\alpha_1}, \quad\;
  \pa_\alpha^r\eta = \sum_{j}(ij)^r\hat\eta_{j}e^{ij\alpha_1}, \quad\;
  \pa_\alpha^2\xi = \sum_{j}ij|j|\hat\eta_{j}e^{ij\alpha_1},
\end{equation}
where $r\in\{1,2\}$. The 1d inverse FFT can then be used to compute
$\td A^\e1(\alpha_1)$ and $\td A^\e2(\alpha_1)$ on a uniform grid in the
$\alpha_1$ variable that is fine enough to resolve the Fourier modes
to the desired accuracy. We then write
\begin{equation}\label{eq:A34}
  \begin{aligned}
    \dot\kappa(\alpha_1,\alpha_2) &=
    \td A^\e3(\alpha_1)e^{il_1\alpha_1}e^{il_2\alpha_2}, & \quad
    \td A^\e3 &= -3\frac\kappa J \td A^\e1 + \frac1{J^{3/2}}\td A^\e2, \\
    -\frac b{2J^2}\dot J + g\dot\eta - \tau\dot\kappa &=
    \td A^\e4(\alpha_1)e^{il_1\alpha_1}e^{il_2\alpha_2}, &
    \td A^\e4 &= -\frac b{J^2}\td A^\e1 + g - \tau \td A^\e3,
  \end{aligned}
\end{equation}
where $\td A^\e{3}(\alpha_1)$ and $\td A^\e4(\alpha_1)$ are computed
pointwise on the grid. For each $m\in\{1,2,3,4\}$, we note that
$\td A^\e{m}(\alpha_1)$ depends on $(l_1,l_2)$, and will be written
$\td A^\e{l_1,l_2,m}(\alpha_1)$ when the dependence needs to be shown
explicitly. Finally, we obtain
\begin{equation}\label{eq:dot:R:pqua}
  \big(D_q\mc R\big[q^\per\big]\dot q^{(l_1, l_2)}\big)(\alpha_1,\alpha_2) =
  P\big[ \td A^\e{l_1,l_2,4}(\alpha_1)e^{il_1\alpha_1}
    e^{il_2\alpha_2} \big] = \td u^{(l_1, l_2)}(\alpha_1)e^{il_2\alpha_2},
\end{equation}
where the projection $P$ was defined in \eqref{eq:hilbert} above.
The Fourier expansion
\begin{equation}\label{eq:uhat:l1l2}
  \td u^\e{l_1,l_2}(\alpha) = \sum_j
  \hat u^\e{l_1,l_2}_j e^{ij\alpha}, \qquad\quad
  \hat u^\e{l_1,l_2}_j =
  \begin{cases}
    \hat A^\e{l_1,l_2,4}_{j-l_1}, & (j,l_2)\ne(0,0), \\
    0, & (j,l_2)=(0,0),
    \end{cases}
\end{equation}
is easily read off from the FFT of $\td A^\e{l_1,l_2,4}(\alpha)$,
where we used the fact that multiplication by
$e^{il_1\alpha_1}$ in \eqref{eq:dot:R:pqua} simply shifts the Fourier
index by $l_1$.  Of course, by Remark~\ref{rmk:hat:breve}, $\hat
  A^\e{l_1,l_2,4}_{j-l_1}$ will be computed via a one-dimensional
  de-aliasing formula analogous to \eqref{eq:dft:approx}.  Since
$\tau$ and $b$ are real and
$\eta(\alpha_1,\alpha_2)=\td\eta(\alpha_1)$ is real-valued, inspection
of \eqref{eq:lin:l1l2}--\eqref{eq:dot:R:pqua} shows that
\begin{equation}\label{eq:Au:neg:l}
  \td A^\e{-l_1,-l_2,m}(\alpha)=\overline{\td A^\e{l_1,l_2,m}(\alpha)}, \qquad
  \td u^\e{-l_1,-l_2}(\alpha)=\overline{\td u^\e{l_1,l_2}(\alpha)}, \qquad
  \left(\begin{gathered} m=1,2,3,4 \\[-3pt] l_1,l_2 \in \mbb Z
      \end{gathered}\right).
\end{equation}
This shows that $D_q\mc R\big[q^\per\big]\overline{\dot q^{(l_1,
      l_2)}}= \overline{D_q\mc R\big[q^\per\big]\dot q^{(l_1, l_2)}},$
which is also evident from \eqref{eq:lin:l1l2}.  If, moreover, $\eta$ has even
symmetry, then
\begin{equation}
  \td A^\e{l_1,l_2,m}(-\alpha)=\overline{\td A^\e{l_1,l_2,m}(\alpha)}, \qquad
  \td u^\e{l_1,l_2}(-\alpha)=\overline{\td u^\e{l_1,l_2}(\alpha)}, \qquad
  \left(\begin{gathered} m=1,2,3,4 \\[-3pt] l_1,l_2 \in \mbb Z
      \end{gathered}\right),
\end{equation}
which implies that the Fourier coefficients of these complex-valued
functions are real.

\begin{remark}
  Evaluation of \eqref{eq:lin:l1l2} and \eqref{eq:dot:R:pqua} along
  the characteristic line $\alpha_1=\alpha$, $\alpha_2=k\alpha$ gives
  the real-line version of these equations, which can be derived
  directly via a Fourier-Bloch analysis commonly used in
  the study of subharmonic stability of traveling waves
  \cite{oliveras:11, tiron2012linear}. However, by posing the problem
  in a quasi-periodic torus framework, it becomes possible to follow
  bifurcation branches beyond linearization about periodic traveling
  waves.
\end{remark}

\begin{remark}\label{rmk:block:diag}
  In summary, we have shown that $D_q\mc R[q^\per]$ has a block
  structure, mapping $(\dot\eta,\dot\tau,\dot b)\in
  \big(\Vp{\sigma},\mbb C,\mbb C\big)$ to $\Vp{\rho}$ and
  $(\dot\eta,0,0) \in \big( \mc V_{\sigma}^\e{l_2},0,0
    \big)$ to $\mc V_{\rho}^\e{l_2}$ for $l_2\in\mbb
  Z\setminus\{0\}$.
\end{remark}

\begin{remark} \label{rmk:crandall}
  The spaces $\mc V_\sigma$ are convenient for identifying the
  block structure of $D_q\mc R[q^\per]$, which leads us to a
  numerical algorithm for computing quasi-periodic bifurcation
  points and perturbation directions to switch to the new branch;
  however, we are not able to apply rigorous bifurcation theorems
  such as the Crandall-Rabinowitz theorem \cite{Crandall:bifur:1971}
  to prove existence of genuinely spatially quasi-periodic water
  waves in this framework since $D_q\mc R[q^\per]$ is not a Fredholm
  operator from $\big(\mc V_\sigma,\mbb C,\mbb C\big)$ to $\mc
  V_\rho$.  Indeed, when $\rho$ and $\sigma$ are chosen as in the
  proof of Theorem B.1 in Appendix B, the algebraic codimension is
  infinite since $D_q\mc R[q^\per]$ is also bounded if the range is
  decreased slightly to $\mc V_{\rho+\veps}$ for sufficiently small
  $\veps$, and the embedding of $\mc V_{\rho+\veps}$ into $\mc
  V_\rho$ has infinite algebraic co-dimenison.  Proofs of existence
  \cite{toland1996stokes,buffoni2000sub,buffoni2000regularity} of
  bifurcations from $2\pi$-periodic traveling waves to $2\pi
  m$-periodic traveling waves for sufficiently large integers $m$
  employ a variant of Nekrasov's equation \cite{nekrasov1921steady}
  instead of \eqref{eq:Rq} for the governing equations. Adapting
  these proofs to the quasi-periodic case is an interesting avenue
  of future research, and may require employing Nash-Moser theory
  \cite{plotnikov01,berti2016quasi,baldi2018time,berti2020traveling}
  to overcome small divisors, whose effects can be seen in the
  numerical results presented in Section~\ref{sec:num:grav} below.
\end{remark}

It is convenient at this point to introduce alternative basis
functions and subspaces that more clearly exhibit the even,
real-valued nature of the solutions we seek. Let
\begin{equation}\label{eq:varphi:def}
  \begin{aligned}
    \varphi_{\vec l}(\vec\alpha) &= 2\cos(\vec l\cdot\vec\alpha), \\
    \psi_{\vec l}(\vec\alpha) &= -2\sin(\vec l\cdot\vec\alpha),
  \end{aligned}
  \qquad
  \vec l\in \Lambda = \Big\{(l_1,l_2)\in\mbb Z^2\,:\, l_2>0
  \text{\, or\, } \big(l_2=0 \text{ and } l_1>0\big)\Big\},
\end{equation}
where $\vec\alpha=(\alpha_1,\alpha_2)$. We also define
$\varphi_{0,0}(\vec\alpha)=1$. Then
\begin{equation}
  \Big( \varphi_{\vec l}(\vec\alpha)\,,\,
    \psi_{\vec l}(\vec\alpha) \Big) =
  \Big( e^{i\vec l\cdot \vec\alpha}\,,\, e^{-i\vec l\cdot \vec\alpha} \Big)
  \begin{pmatrix} 1 & i \\ 1 & -i \end{pmatrix}, \qquad
  (\vec l\in \Lambda)
\end{equation}
and the torus function expansions of an arbitrary function
\begin{equation}
  f(\vec\alpha)=\sum_{j_1,j_2}\hat f_{j_1,j_2}e^{i(j_1\alpha_1+j_2\alpha_2)} =
 a_{\vec 0}\varphi_{\vec 0}(\vec\alpha) +
  \sum_{\vec l\in\Lambda} \big( a_{\vec l}\varphi_{\vec l}(\vec\alpha)
    + b_{\vec l}\psi_{\vec l}(\vec\alpha)\big)
\end{equation}
are related by
\begin{equation}
  a_{\vec0} = \hat f_{\vec0},
  \quad\;\;
  \begin{pmatrix} a_{\vec l} \\ b_{\vec l} \end{pmatrix} =
  \begin{pmatrix} \frac12 & \frac12 \\[4pt] \frac1{2i} & \frac{-1}{2i}
  \end{pmatrix}
  \begin{pmatrix} \hat f_{\vec l} \\ \hat f_{-\vec l} \end{pmatrix},
  \quad\;\;
  \begin{pmatrix} \hat f_{\vec l} \\ \hat f_{-\vec l} \end{pmatrix} =
  \begin{pmatrix} 1 & i \\ 1 & -i \end{pmatrix}
  \begin{pmatrix} a_{\vec l} \\ b_{\vec l} \end{pmatrix},
  \quad\;\; (\vec l\in\Lambda).
\end{equation}
Note that $f(\vec\alpha)$ is real-valued precisely when all the
$a_{\vec l}$ and $b_{\vec l}$ are real. In this case, these
coefficients are the real and imaginary parts of $\hat f_{\vec l}$ for
$\vec l\in\Lambda$, and $\hat f_{-\vec l}= \overline{\hat f_{\vec
    l}}$. Similarly, $f(\vec\alpha)$ is even with respect to
$\vec\alpha\in\mbb T^2$ precisely when all the $b_{\vec l}$ are zero.
We also define the subspaces
\begin{equation}\label{eq:XY:spaces}
  \begin{aligned}
    &\mc X_\sigma^\text{const} = \vspan_\sigma\{\varphi_{0,0}\} =
    \{\text{constant functions on $\mbb T^2$}\}, \\
    &\mc X_\sigma^\per = \mc X_\sigma^\e0 =
    \vspan_\sigma\{\varphi_{l_1,0} \,:\, l_1\ge1\}, \qquad
    \mc Y_\sigma^\per = \mc Y_\sigma^\e0 =
    \vspan_\sigma\{\psi_{l_1,0} \,:\, l_1\ge1\}, \\
    &\mc X_\sigma^\e{l_2} =
    \vspan_\sigma\{\varphi_{l_1,l_2} \,:\, l_1\in\mbb Z\}, \qquad
    \mc Y_\sigma^\e{l_2} =
    \vspan_\sigma\{\psi_{l_1,l_2} \,:\, l_1\in\mbb Z\}, \qquad
        (l_2\ge1),
  \end{aligned}
\end{equation}
where $\vspan_\sigma$ of a list of functions is the closure
of the set of finite linear combinations of the functions with respect
to the $\mc V_\sigma$ norm from \eqref{eq:sigma:norm}. We note that
\begin{align}
  \Vp{\sigma} &= \mc X_\sigma^\text{const} \oplus
  \mc X_\sigma^\per\oplus\mc Y_\sigma^\per, \qquad
  \mc V_\sigma^\e{l_2}\oplus \mc V_\sigma^\e{-l_2} =
  \mc X_\sigma^\e{l_2}\oplus \mc Y_\sigma^\e{l_2}, \quad (l_2\ge1), \\
  \label{eq:VXY}
  \mc V_\sigma &= \mc X_\sigma^\text{const} \oplus
  \mc X_\sigma^\per \oplus
  \mc Y_\sigma^\per\oplus \mc X_\sigma^\qua \oplus \mc Y_\sigma^\qua,
  \qquad \mc X_\sigma^\qua = \bigoplus_{l_2=1}^\infty \mc X_\sigma^\e{l_2},
\end{align}
with a similar formula for $\mc Y_\sigma^\qua$.  Since the basis
functions $\varphi_{\vec l}(\vec\alpha)$ and $\psi_{\vec
  l}(\vec\alpha)$ are real-valued, the spaces in \eqref{eq:VXY} may be
regarded as complex or real Hilbert spaces.  It is also useful to
define
\begin{equation}\label{eq:Xsig:def}
  \mc X_\sigma = \mc X_\sigma^\per \oplus \mc X_\sigma^\qua,
\end{equation}
which contains the even functions of $\mc V_\sigma$ of zero mean on
$\mbb T^2$.

As explained above, since we assume in \eqref{eq:assump:eta} that
$\eta$ is real-valued and even, the Fourier coefficients of
$\td u^\e{l_1,l_2}(\alpha)$ in \eqref{eq:uhat:l1l2} are real. Thus, from
\eqref{eq:dot:R:pqua} and \eqref{eq:Au:neg:l} we have
\begin{equation}\label{eq:DqR:phi:psi}
  \begin{aligned}
    \jt
    D_q \mc R\big[q^\per\big]\big(\varphi_{\vec l},0,0\big) =
    \sum_{j_1\in\mbb Z} \hat u_{j_1}^\e{l_1,l_2}\varphi_{j_1,l_2}(\vec\alpha), \\
    \jt
    D_q \mc R\big[q^\per\big]\big(\psi_{\vec l},0,0\big) =
    \sum_{j_1\in\mbb Z} \hat u_{j_1}^\e{l_1,l_2}\psi_{j_1,l_2}(\vec\alpha),
  \end{aligned} \qquad\quad
  \big(\,\vec l = (l_1,l_2) \in \Lambda\,\big).
\end{equation}
The $j_1=0$ term can be omitted from these sums when $l_2=0$ since
$\hat u^\e{l_1,0}_0=0$, and must be omitted in the second formula as
  $\psi_{0,0}(\vec\alpha)$ is not defined.  We also see directly from
\eqref{eq:dot:R:tau:b} and \eqref{eq:lin:l1l2} that $D_q \mc
R\big[q^\per\big]\big(0, \dot\tau,\dot b\big)$ is an even function and
$D_q \mc R\big[q^\per\big]\big(\varphi_{0,0},0,0\big)=0$.

\begin{remark} \label{rmk:X:block}
  Since $\eta$ is assumed real-valued and even and $D_q \mc
  R\big[q^\per\big]$ maps even perturbations to even functions and odd
  perturbations to odd functions, we can restrict attention to even
  functions and perturbations. In light of Remark~\ref{rmk:hat:eta00}
  and the presence of the projection $P$ in the definition
  \eqref{eq:Rq} of $\mc R$, we may further restrict attention to
  functions and perturbations of zero mean.  It follows from
  \eqref{eq:DqR:phi:psi} that $D_q \mc R\big[q^\per\big]$ has a block
  structure with respect to the decomposition \eqref{eq:Xsig:def},
  mapping $(\dot\eta,\dot\tau,\dot b) \in \big(\mc X_\sigma^\per,\mbb R,
    \mbb R\big)$ to $\mc X_\rho^\per$ and $(\dot\eta,0,0) \in
  \big(\mc X_\sigma^\e{l_2},0,0\big)$ to $\mc
  X_\rho^\e{l_2}$ for $l_2\ge1$.
\end{remark}

\begin{figure}
  \begin{center}
    \includegraphics[width=0.54\textwidth]{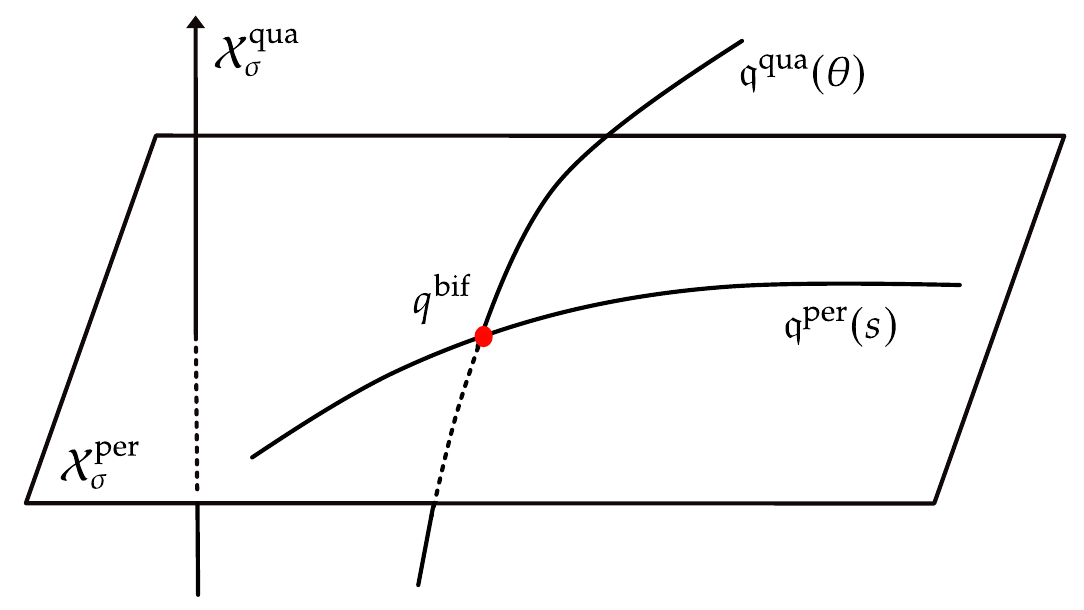}
  \end{center}
  \caption{\label{fig:simple:bifur}  Simple quasi-periodic bifurcation diagram}
\end{figure}

If the surface tension $\tau$ is held fixed, which will be the case
for the pure gravity wave problem in section~\ref{sec:num:grav} below,
then there will be a one-parameter ``primary'' branch of periodic
traveling waves, which we denote by $q=\mfq^\per(s)\in (\mc
X_\sigma^\per,\mbb R,\mbb R\big).$ Here $s$ is any convenient amplitude
parameter such as $\hat\eta_{1,0}$ or the crest-to-trough height of
the wave, $h=\td\eta(0)-\td\eta(\pi)$.  Following \cite{Crandall:bifur:1971,
  chen1980numerical, chen1979steady, zeidlerAFA, allgower2012numerical}, we are
interested in finding simple bifurcation points $q^\bif$ where a
second solution curve $q=\mfq^\qua(\theta)\in\big(\mc X_\sigma,\mbb R,
\mbb R\big)$ intersects the first non-tangentially, with $\theta$ another
amplitude parameter such as $\hat\eta_{0,1}$; see Figure
\ref{fig:simple:bifur}. Suppose such an intersection occurs at $s=s_0$
and $\theta=\theta_0$. Differentiating $\mc R\big[\mfq^\per(s)\big] =
0$ and $\mc R\big[\mfq^\qua(\theta)\big] = 0$, we obtain
\begin{equation}
D_q\mc R\big[q^\bif\big]\Big(\big(\mfq^\per\big)'(s_0)\Big) = 0 =
D_q\mc R\big[q^\bif\big]\Big(\big(\mfq^\qua\big)'(\theta_0)\Big).
\end{equation}
At a simple bifurcation \cite{Crandall:bifur:1971}, these null vectors
span the kernel of $D_q\mc R\big[q^\bif\big]$,
\begin{equation}\label{eq:kDq:span}
  \ker D_q\mc R\big[q^\bif\big]  =
  \vspan\Big\{\,
  \big(\mfq^\per\big)'(s_0)\;,\, \big(\mfq^\qua\big)'(\theta_0)\,
  \Big\}.
\end{equation}
Because $\tau$ is frozen and we restrict attention to even
perturbations of zero mean, the domain of $D_q\mc R\big[q^\bif\big]$
is taken to be $\mc D_\sigma = \big(X_\sigma,0,\mbb R\big)$
when computing the kernel. We decompose
\begin{equation}\label{eq:mcD}
  \mc D_\sigma=\mc D_\sigma^\per\oplus\left(\bigoplus_{l=1}^\infty
    \mc D_\sigma^\e{l}\right), \qquad
  \begin{aligned}
    \mc D_\sigma^\per &= \{(\dot\eta,0,\dot b)\,:\, \dot\eta\in\mc
    X_\sigma^\per,\dot b\in\mbb R\}, \\
    \mc D_\sigma^\e{l} &=
    \{(\dot\eta,0,0)\,:\, \dot\eta\in\mc X_\sigma^\e{l}\}.
    \end{aligned}
\end{equation}
Let $\dot q^\per = \big(\mfq^\per\big)'(s_0)$, which belongs to $\mc
D_\sigma^\per$ since solutions on this branch are periodic.  By
\eqref{eq:mcD}, we can decompose $\big(\mfq^\qua\big)'(\theta_0)= \dot
q_1^\per + \sum_{l=1}^\infty \dot q^\e{l}$ with $\dot q_1^\per\in\mc
D_\sigma^\per$ and $\dot q^\e{l}\in \mc D_\sigma^\e{l}$. Our
assumption that the bifurcation is simple implies that precisely one
of the $\dot q^\e l$ is non-zero. In more detail, because $D_q\mc
R\big[q^\bif\big]$ has a block structure, each non-zero component
$\dot q_1^\per$ and $\dot q^\e{l}$ will also be in the kernel. At
least one of the $\dot q^\e{l}$, say with $l=l_0$, must be non-zero
since solutions on the path $\mfq^\qua(\theta)$ are supposed to be
genuinely quasi-periodic, and we require that this occurs at linear
order in the perturbation.  But then all the other $\dot q^\e{l}$ must
be zero and $\dot q_1^\per$ must be a multiple of $\dot q^\per$, or
else the dimension of $\ker D_q\mc R\big[q^\bif\big]$ would be greater
than two. For simplicity, and without loss of generality, we may
assume $l_0=1$. Indeed, observing the way $k$ and $l_2$ appear in the
formulas of \eqref{eq:lin:l1l2}, we see that if $\dot q^\e{l}$ belongs
to the kernel for $l>1$, then $\dot q^\e1$ will be in the kernel for
an auxiliary problem with $k$ replaced by $kl$. Renaming $\dot q^\e1$
by $\dot q^\qua$, we have shown that the kernel should take the form
\begin{equation}\label{eq:span:dot:q}
  \ker D_q\mc R\big[q^\bif\big]  =
  \vspan\big\{\, \dot q^\per\,,\, \dot q^\qua \big\}, \qquad
  \dot q^\per \in \mc D_\sigma^\per, \qquad
  \dot q^\qua \in \mc D_\sigma^\e1,
\end{equation}
where $\big(\mfq^\per\big)'(s_0)=\dot q^\per$ and
$\big(\mfq^\qua\big)'(\theta_0) = C\dot q^\per + \dot q^\qua$ for some
$C\in\mbb R$. As explained in Section~\ref{sec:num:grav} below, a
non-zero value of $C$ would break a symmetry that arises on the
solution branches we have found. Thus, $C$ turns out to be zero and
$q^\bif + \veps \dot q^\qua$ can be used as a natural initial
guess to switch from the periodic branch to the quasi-periodic branch,
where $\veps$ is a suitably small number chosen empirically.

When the surface tension is allowed to vary, which will be the case in
the gravity-capillary wave problem of section~\ref{sec:num:cap}, there
will be a two-parameter family of periodic traveling waves that
contains one-dimensional bifurcation curves where the periodic waves
intersect with two-dimensional sheets of quasi-periodic traveling
waves.  In this formulation, one can increase the domain of $D_q\mc
R\big[q^\bif\big]$ from \eqref{eq:mcD} to $\mc D_\sigma= \big(\mc
  X_\sigma,\mbb R,\mbb R\big)$ and the dimension of the kernel will
increase from two to three. But computationally, we can sweep through
these manifolds of solutions with $\tau$ or an amplitude parameter
fixed, which reduces the problem to searching for isolated bifurcation
points along one-parameter solution curves, as described above.  The
general theory of multi-parameter bifurcation theory is presented in
\cite{antman:81,antman:book}, for example.

\begin{remark}\label{rmk:eta:mean}
  If we had not separated $\mc X_\sigma^\text{const}$ from $\mc
  X_\sigma^\per$ in \eqref{eq:XY:spaces}, the dimension of the
  manifolds of solutions would increase by one, as would the kernel,
  but in a trivial way: one can add a constant to $\eta$ for any of
  the solutions to obtain another solution.  One could add an equation
  to the nonlinear system $\mc R[q]=0$ to select the physical solution
  with zero mean in physical space, but it is simpler to just hold
  $\hat\eta_{0,0}=0$ in this search phase of the problem and compute
  the correction afterwards, as explained in
  Remark~\ref{rmk:hat:eta00}.
\end{remark}

\subsection{Detecting Quasi-Periodic Bifurcation Points}
\label{sec:detect:bif}
In this section, we discuss how to detect quasi-periodic bifurcation
points on $\mfq^\per(s)$ and compute the corresponding bifurcation
directions. Because we seek bifurcation points $q^\bif$ such that
$D_q\mc R\big[q^\bif\big]$ has a null vector $\dot q^\qua\in\mc
D_\sigma^\e{1}$, it suffices to compute the restriction of $D_q\mc
R\big[\mfq^\per(s)\big]$ to $\mc D_\sigma^\e{1}=\big(\mc
  X_\sigma^\e1,0,0\big)$ and search for values of $s$ for which this
operator has a non-trivial kernel. By Remark~\ref{rmk:X:block}, the
range of this restriction may be taken to be $\mc X_\rho^\e1$.

Recall from Section~\ref{sec:num:alg} that we compute periodic
traveling waves numerically by specifying $\hat\eta_{1,0}$ and $\tau$
as given parameters and minimizing the objective
function~\eqref{eq:min:F:1d} to find the square of the wave speed,
$b=c^2$, and the remaining leading Fourier modes
$(\hat\eta_{2,0},\dots,\hat\eta_{N_1,0})$. The computations are done
on a uniform $M_1$-point grid in the $\alpha_1$ variable, where
$M_1\approx 3N_1$ is generally sufficient to achieve spectrally
accurate solutions with minimal effects of aliasing errors. Fourier
modes $\hat\eta_{j_1,0}$ with $|j_1|>N_1$ are taken to be zero, and
$\eta$ is assumed real and even so that
$\hat\eta_{-j_1,0}=\hat\eta_{j_1,0}\in\mbb R$.  $N_1$ is chosen large
enough that the Fourier modes decay to machine precision by the time
$|j_1|$ reaches $N_1$.

For each of these computed periodic solutions, $q^\per=\mfq^\per(s)$,
we form the matrix $\mc J^\qua\big[q^\per\big]$ representing the
restriction of $D_q\mc R\big[q^\per\big]$ to $\mc D_\sigma^\e{1}$,
using the $\{\varphi_{l_1,1}\}$ basis in both the domain and range
of the restricted operator, up to a cutoff frequency
$N$. We order the basis functions via
$l_1=(0,1,-1,2,-2,3,-3,\dots,N,-N)$ and use \eqref{eq:DqR:phi:psi} to
obtain
\begin{equation}\label{eq:J:qua}
  \renewcommand{\arraystretch}{1.6}
  \mc J^\qua \big[q^\per\big] =
  \begin{pmatrix}
    \hat{u}^{(0,1)}_0 & \hat{u}^{(1,1)}_0 & \hat{u}^{(-1,1)}_0 &
    \cdots &\hat{u}^{(N,1)}_0 
    & \hat{u}^{(-N, 1)}_0 \\
    \hat{u}^{(0,1)}_1 & \hat{u}^{(1,1)}_1 & \hat{u}^{(-1,1)}_1 &
    \cdots &\hat{u}^{(N,1)}_1 
    & \hat{u}^{(-N, 1)}_1 \\
    \hat{u}^{(0,1)}_{-1} & \hat{u}^{(1,1)}_{-1} & \hat{u}^{(-1,1)}_{-1} &
    \cdots &\hat{u}^{(N,1)}_{-1} 
    & \hat{u}^{(-N, 1)}_{-1}\\
    \vdots & \vdots & \vdots & & \vdots &\vdots \\
    \hat{u}^{(0,1)}_{N} & \hat{u}^{(1, 1)}_{N} &
    \hat{u}^{(-1,1)}_{N} & \cdots 
    &\hat{u}^{(N,1)}_{N} & \hat{u}^{(-N, 1)}_{N}\\
    \hat{u}^{(0,1)}_{-N} & \hat{u}^{(1,1)}_{-N} &
    \hat{u}^{(-1,1)}_{-N} 
    & \cdots &\hat{u}^{(N, 1)}_{-N} & \hat{u}^{(-N, 1)}_{-N}
  \end{pmatrix},
\end{equation}
where $\td u^\e{l_1,l_2}(\alpha_1)$ is defined above, in
\eqref{eq:dot:R:pqua}, and a formula for its Fourier modes $\hat
u^\e{l_1,l_2}_j$ is given in \eqref{eq:uhat:l1l2}. We generally choose
$N$ in the range $N_1\le N\le (3/2)N_1$. The goal is to have enough
rows and columns in $\mc J^\qua$ that when the singular value
decomposition
\begin{equation}\label{eq:J:svd}
  \mc J^\qua[\mfq^\per(s)] = U(s)\Sigma(s) V(s)^T, \qquad
  \Sigma(s)=\diag\Big( \sigma_1(s)\,,\,\sigma_2(s)\,,\,
    \dots\,,\,\sigma_{2N+1}(s) \Big)
\end{equation}
is computed, the left and right singular vectors corresponding to the
smallest singular value have expansions in the $\{\varphi_{l,1}\}$
basis with coefficients $\{a_{l,1}\}$ that decay in amplitude to
machine precision by the time $|l|$ reaches $N$. Here the
  singular values $\sigma_i(s)$ are not related to the parameter
  $\sigma$ in the weighted spaces $\mc V_\sigma$, $\mc D_\sigma$ and
  $\mc X_\sigma$. Our reasons for using the unweighted Fourier basis
  $\{\varphi_{l_1,1}\}$ when computing the matrix representation of
  $D_q\mc R\big[q^\per\big]$ from $\mc D_\sigma^\e1$ to $\mc
  X_\rho^\e1$ are explained in Remark~\ref{rmk:J:unweighted} below.

Because this is a discretization of an infinite-dimensional problem,
the smallest singular values are of physical interest while the
largest are under-resolved and inaccurately computed. Thus, we order
the singular values in \emph{ascending order}
\begin{equation}\label{eq:sigma:ord}
  0\le \sigma_1(s) \le \sigma_2(s) \le \cdots \le \sigma_{2N+1}(s).
\end{equation}
Increasing $N$ further does not change the smallest singular values
and corresponding singular vectors (up to floating-point arithmetic
  effects) as they are already fully resolved without using the
high-frequency columns and rows of $\mc J^\qua$ that are added.
Standard computational routines, of course, return them in descending
order. But for this discussion, we reverse the columns of $U(s)$ and
the rows of $V(s)^T$ from the computation to match the convention
\eqref{eq:sigma:ord}. In the end, the matrices $U(s)$ and $V(s)^T$ do
not have to be computed at all; see Section~\ref{sec:svd:det}.

As the parameter $s$ changes, we compute $\Sigma(s)$ in
\eqref{eq:J:svd} and search for zeros $s_0$ of the smallest singular
value, $\sigma_{1}(s_0)=0$.  Since singular values are
returned as non-negative quantities, $\sigma_{1}(s)$ has a slope
discontinuity at each of its zeros. This is a challenge for
root-finding algorithms that rely on bracketing or polynomial
approximation. But if the matrix entries of $\mc J^\qua$ depend
analytically on $s$, there is an analytic SVD in which the singular
values and vectors are analytic functions of $s$; see Bunse-Gerstner
\emph{et al.} \cite{mehrmann:svd} and Kato (\cite{kato}, pp.~120--122,
  392--393). If $\mc J^\qua$ is not analytic with respect to $s$,
a smooth SVD \cite{dieci99}
often exists and can be used instead. When the SVD is computed
numerically, instead of changing the sign of $\sigma_{1}(s)$ when $s$
crosses $s_0$, the corresponding left or right singular vector will
change sign. If we transfer this sign change to $\sigma_{1}(s)$, the
slope discontinuity is eliminated and $\sigma_{1}(s)$ becomes
analytic or smooth.

One idea we experimented with is to multiply $\sigma_{1}(s)$ by
$\opn{sgn}(\la u_{1},v_{1}\ra)$, where $u_{1}$ and $v_{1}$ are the
first columns of $U$ and $V$ in \eqref{eq:J:svd} under the convention
\eqref{eq:sigma:ord}; $\sgn(a)$ is defined above in
\eqref{eq:sgn:def}; and $\la u,v\ra=\sum_j u^jv^j$, where we use
superscripts for the components of a column vector that has been
extracted from a matrix. If $u_{1}(s_0)$ and $v_{1}(s_0)$ are nearly
orthogonal, one can instead use $\opn{sgn}(\la Tu_{1},v_{1}\ra)$,
where $T$ is a Householder reflection that aligns $u_{1}(s_1)$ with
$v_{1}(s_1)$ at some point $s_1$ near $s_0$. Although this works fine
in the present problem (even without introducing $T$), it is clear
that in general, $\la Tu_{1},v_{1}\ra$ might change discontinuously if
$\sigma_{1}(s)$ and $\sigma_{2}(s)$ ever cross, possibly leading to a
sign change that does not correspond to a zero crossing of
$\sigma_{1}(s)$.

We propose, instead, to use the signs of the determinants of $U$ and
$V$ to track orientation changes in the singular vectors when $s$
crosses a zero of $\sigma_{1}(s)$. So we define
\begin{equation}\label{eq:chi}
  \begin{aligned}
    \chi(s) &= \sgn\big( \det U(s) \big)\sgn\big( \det V(s) \big) \sigma_{1}(s) \\
    &= \big( \sgn\det \mc J^\qua\big[\mfq^\per(s)\big]\, \big)\,
    \sigma_{1}(s).
    \end{aligned}
\end{equation}
The zeros $s_0$ of $\chi(s)$ will be used to identify bifurcation
points $q^\bif= \mfq^\per(s_0)$. Since $U(s)$ and $V(s)$ are
orthogonal, their determinants are equal to 1 or $-1$ and can be
computed accurately by $LU$ or $QR$ factorization to determine which.
(Including $\sgn$ just rounds the numerical result to the exact
  value).  The second formula of \eqref{eq:chi} has to be treated with
care, but can be computed faster than the first formula as an
intermediate step of computing $\sigma_{1}(s)$, without having
to actually form the matrices $U(s)$ or $V(s)$ or compute their
determinants; see Section~\ref{sec:svd:det}.

To find a zero $s_0$ of $\chi(s)$, we can use a root bracketing
technique such as Brent's method \cite{brent:73} to reduce $|\chi(s)|$
to the point that floating-point errors corrupt the smallest singular
value of the SVD algorithm, which is typically below $10^{-13}$ in
double-precision. Alternatively, one can compute $\chi(s)$ at a set of
Chebyshev nodes on an interval $[s_1,s_2]$ for which $\chi(s_1)$ and
$\chi(s_2)$ have opposite signs. One can then use Newton's method or
Brent's method on the Chebyshev interpolation polynomial to find a
zero $s_0$ of $\chi(s)$. We demonstrate both techniques in
Section~\ref{sec:num} below.  Once $q^\bif=\mfq^\per(s_0)$ has been
found, the null vector $\dot q^\qua$ in \eqref{eq:span:dot:q} is given
by
\begin{equation}\label{eq:q:qua:from:v}
  \dot q^\qua = (\dot\eta^\qua,0,0), \qquad
  \dot \eta^\qua = v^1\varphi_{0,1} + \sum_{j=1}^N
  \big[ v^{2j}\varphi_{j,1} + v^{2j+1}\varphi_{-j,1} \big] =
  \sum_{j=-N}^N a_{j,1}\varphi_{j,1},
\end{equation}
where $v=v_{1}$ is the first column of $V$ and we make use of the row
and column ordering of $\mc J^\qua$ in \eqref{eq:J:qua}. Here
$a_{0,1}=v^1$, $a_{j,1}=v^{2j}$ and $a_{-j,1}=v^{2j+1}$ for $1\le
j\le N$. Since $V$ is orthogonal, $\|v_{1}\|=1$, so $\dot q^\qua$ is
already normalized sensibly and $q^\bif + \veps\dot q^\qua$ serves as
a useful initial guess for computing solutions on the secondary branch
$\mfq^\qua(\theta)$. The null vector only needs to be computed when
$s=s_0$ is a zero of $\chi(s)$, and can be computed efficiently
without forming $U(s)$ or the other columns of $V(s)$ if the problem
size is large enough to make these calculations expensive; see
Section~\ref{sec:svd:det}.

\begin{remark} \label{rmk:chi}
An important feature of $\chi(s)$ is that once $N$ is large enough
that the left and right singular vectors $u_{1}$ and $v_{1}$ have
entries that decay to zero in floating-point arithmetic, further
increases in $N$ do not change the numerical value of $\chi(s)$.
\end{remark}

To explain Remark~\ref{rmk:chi}, we first consider the zero-amplitude
case. Setting $\eta=0$, $\xi=0$, $J=1$ and $\kappa=0$ in
\eqref{eq:lin:l1l2}, we find that
\begin{equation}\label{eq:uhat:zro:amp}
  \hat u^\e{l_1,l_2}_j = \begin{cases}
    -b|l_1+kl_2|+g+\tau(l_1+kl_2)^2, & j=l_1, \\
    0, & \text{otherwise.}
  \end{cases}
\end{equation}
Thus, $\mc J^\qua$ in \eqref{eq:J:qua} is diagonal, and the last two
diagonal entries are
\begin{equation}
  \begin{aligned}
  \mc J^\qua_{2N,2N} &= (1-\tau|N+k|)(1-|N+k|), \\
  \mc J^\qua_{2N+1,2N+1} &= (1-\tau|-N+k|)(1-|-N+k|),
  \end{aligned}
\end{equation}
where we used $b=c^2=(g/k_1)+\tau k_1$ for the square of the wave
speed of the zero-amplitude traveling wave of dimensionless wave
number $k_1=1$ and gravitational acceleration $g=1$.  Assuming $k>0$
and $N>k+1$, we have $1-|\pm N+k|\le 1-(N-k)<0$, so both of these
final diagonal entries are negative if $\tau=0$.  If $\tau>0$, then
both diagonal entries will be positive once $\tau N>\tau k+1$ and
$N>k+1$. Thus, with or without surface tension, once $N$ is large
enough, increasing $N$ by one does not change the sign of the
determinant of $\mc J^\qua$ in the linearization about $\eta=0$. With
$N$ fixed, because an analytic SVD of the form \eqref{eq:J:svd}
exists, the sign of $\det\mc J^\qua[\mfq^\per(s)]$ will only change
when $s$ passes through a zero of $\chi(s)$, which is the signed
version of the smallest singular value $\sigma_{1}(s)$.
If such zero crossings correspond to well-resolved singular vectors in
the kernel and $N_1$ and $N_2$ are large enough,
then upon replacing $N=N_1$ by $N=N_2$, the same crossings will be
encountered and the sign of the determinant at a given $s$ will not
change.  This argument would break down if truncating the matrix leads
to a spurious null vector at some $s$ for either $N_1$ or $N_2$, but
we find that only the large singular values are sensitive to where the
matrix is truncated. When a null vector is found, it is easy to check
\emph{a-posteriori} that the entries $v^{2j}$ and $v^{2j+1}$ of
$v=v_{1}$ in \eqref{eq:q:qua:from:v} decay to machine precision by
the time $j$ reaches $N$.

\begin{remark}\label{rmk:bad:det}
  In finite-dimensional bifurcation problems, say $f(q)=0\in\mbb R^n$
  with $q=(u,s)\in\mbb R^{n+1}$ and primary branch parameterized by
  $q=\mfq(s)$, one can use
  \begin{equation}\label{eq:det:test:fcn}
    \det \mc J^e(s), \qquad \mc J^e(s) =
    \left( \begin{array}{c} f_q[\mfq(s)] \\
        \mfq'(s)^T \end{array} \right)
  \end{equation}
  as a test function that changes sign at simple bifurcations.  This
  is the numerical approach advocated in \cite{allgower2012numerical},
  for example, and is one of the test functions implemented in
  \textsc{Matcont} \cite{dhooge:03,bindel14}. In our case, using
  Bloch's theorem, we only have to consider quasi-periodic
  perturbations $\dot q\in\mc D_\sigma^\e{l_2}$ with $l_2=1$, so we
  can replace $\det \mc J^e(s)$ above by $\det \mc
  J^\qua[\mfq^\per(s)]$ in \eqref{eq:J:qua}. The most common way to
  compute the determinant is as the product of the diagonal entries of
  the LU factorization.  But these products can be very large or very
  small, potentially leading to overflow or underflow in
  floating-point arithmetic, and it is difficult to know what order of
  magnitude of the determinant constitutes a zero crossing. One can
  look for zeros among the diagonal entries of the LU factorization,
  but there are many cases where a nearly singular matrix has diagonal
  entries all bounded away from zero. For example, the bidiagonal
  matrix with $1$'s on the diagonal and $2$'s on the superdiagonal has
  unit determinant but is effectively singular once the matrix size
  exceeds 50. Moreover, unlike our $\chi(s)$ function, the numerical
  value of the determinant will change when the matrix truncation
  parameter $N$ changes. For all these reasons, the determinant itself
  is not a suitable function to identify bifurcation points in this
  problem, though its sign is effective at removing the slope
  discontinuities of $\sigma_{1}(s)$, enabling the use of root-finding
  algorithms to rapidly locate its zeros.
\end{remark}

\begin{remark}\label{rmk:dyn:sys} In the context of dynamical
  systems, $du/dt=f(u,s)$, many test functions have been devised to
  identify fold points, Hopf points, and branch points
  \cite{friedman01,dhooge:03, bindel08,bindel14,govaerts:book}. For
  large-scale equilibrium problems arising from discretized PDEs,
  Bindel \emph{et al.} \cite{bindel14} reached the same conclusion we
  did above in Remark~\ref{rmk:bad:det} on the unsuitability of
  \eqref{eq:det:test:fcn} as a test function.  Instead, in
  \cite{bindel14}, minimally augmented systems \cite{griewank:84,
    allgower:97, beyn01,bindel08, govaerts:book} are used together
  with Newton's method to locate branch points. Studying the details
  of this approach, e.g., Algorithm~5 of \cite{bindel14}, the Newton
  iteration involves solving $f=0$ simultaneously with driving
  $\psi(s)$ in \eqref{eq:J:ee} to zero. As a result, intermediate
  Newton iterations will not involve states $u$ that lie precisely on
  the primary bifurcation curve. This causes a problem for us as we
  have to linearize about a periodic solution to use Bloch-Fourier
  theory. The assumption that $\mc J^{ee}$ in \eqref{eq:J:ee} is
  square is also incompatible with our formulation of the traveling
  wave problem as an overdetermined nonlinear least squares problem,
  where $\mc J$ in \eqref{eq:Jij:levmar} has more rows than columns to
  reduce aliasing errors and improve the accuracy of the computed
  periodic or quasi-periodic traveling waves. We prefer to treat the
  two stages of finding traveling waves and studying their behavior
  under perturbation as separate infinite dimensional problems that we
  solve with spectral methods using as many modes as necessary to
  achieve double-precision accuracy. One could still devise a
  minimally augmented systems approach within this philosophy to
  search for changes in the dimension of the kernel of $\mc
  J^\qua[\mfq^\per(s)]$, but it would require a custom implementation
  and the resulting test function $\psi(s)$ analogous to
  \eqref{eq:J:ee} would not be much cheaper to compute than our
  $\chi(s)$. Moreover, $\psi(s)$ is only locally defined near each
  bifurcation point due to various choices of vectors that are made
  when augmenting the Jacobian. It also does not have the mesh
  independence feature of $\chi(s)$, so $\psi(s)$ will change
  discontinuously if the mesh is refined adaptively as $s$ changes.
\end{remark}

\begin{remark}
  The closest test function we have found in the literature to
  \eqref{eq:chi} is a signed version of the magnitude of the smallest
  eigenvalue of $f_u$, denoted $|\lambda_\text{min}(s)|$, (see
    equation (67) of \cite{bindel14}), which is proposed as an
  alternative to $\det(f_u(s))$ for detecting zero-Hopf points. We use
  $\sigma_\text{min}$ instead of $|\lambda_\text{min}|$ and broaden
  the scope of the test function to search for branch points. This has
  the advantage that it can be applied to an equation $g(u,s)=0$ that
  is equivalent to $f(u,s)=0$ but is no longer in dynamical systems
  form. For example, one does not obtain the dynamic water wave
  equations \cite{choi1999exact, dyachenko1996analytical,
    dyachenko1996nonlinear, dyachenko2001dynamics, ruban:2005,
    quasi:ivp} by setting $\eta_t=\mc R[q]$ in \eqref{eq:Rq}, since
  the velocity potential has been eliminated in the traveling wave
  equations. While the eigenvalues of $f_u$ at an equilibrium point
  give information about the dynamics of $u$ under perturbation, only
  the kernel of $g_u$ (and changes in its dimension) are relevant,
  making $\sigma_\text{min}(s)$ more natural than
  $|\lambda_\text{min}(s)|$ in a test function based on solving $g=0$.
\end{remark}

\begin{remark}\label{rmk:J:unweighted}
    Since $D_q\mc R\big[q^\per\big]$ maps $\big(\mc
      X_\sigma^\e1,0,0\big)$ to $\mc X_\rho^\e1$ with $0<\rho<\sigma$,
    it would be natural to use
    $\big\{2^{-1/2}e^{-\sigma(|l_1|+1)}\varphi_{l_1,1}\big\}_{l_1\in\mbb Z}$
    and
    $\big\{2^{-1/2}e^{-\rho(|j_1|+1)}\varphi_{j_1,1}\big\}_{j_1\in\mbb Z}$
    as orthonormal bases for $\mc X_\sigma^\e1$ and $\mc X_\rho^\e1$
    in the domain and range.  This would cause the rows and columns of
    $\mc J^\qua\big[q^\per\big]$ to be rescaled so that entry $\hat
    u^\e{l_1,1}_{j_1}$ in \eqref{eq:J:qua} is multiplied by
    $e^{-\sigma|l_1|}e^{\rho|j_1|}$. This yields a two-parameter
    family of matrices $\mc J^\qua\big[q^\per\big]$, parameterized by
    $\rho$ and $\sigma$, that ultimately predict the same bifurcation
    points and perturbation directions to switch branches. Recall that
    $\sigma$ is defined by the requirement that $\eta\in\mc
    V_\sigma^\per$, so any smaller positive value can also be used
    without violating the hypotheses. If $\sigma$ is small, the
    corresponding $\rho$ will also be small since
    $\rho\in(0,\sigma)$. So we are effectively considering the
    $\sigma\to0^+$ limit, with
    $e^{-\sigma|l_1|}e^{\rho|j_1|}\approx1$, in the formula
    \eqref{eq:J:qua} for $\mc J^\qua\big[q^\per\big]$.

    This is the most suitable choice for the numerical algorithm
    for three reasons. First, if $\rho N$ were large, floating-point
    errors would be amplified in the matrix entries in the lower-left
    corner of $\mc J^\qua\big[q^\per\big]$, where $|j_1|\gg|l_1|$,
    possibly reducing the accuracy of $\sigma_1(s)$ in
    \eqref{eq:sigma:ord} and the corresponding right singular vector,
    $v_1$, which is the desired null vector predicting the bifurcation
    direction at the zeros of $\sigma_1(s)$. Second, $N$ may need to
    be increased to fully resolve this null vector for the rescaled
    version of $\mc J^\qua\big[q^\per\big]$. Indeed, if $v_1^j$ are
    the entries of $v_1$ for the unscaled version of $\mc
    J^\qua\big[q^\per\big]$, then $Ce^{\sigma|j|}v_j$ will be the
    entries of $v_1$ for the rescaled version, which decay slower and
    therefore need a larger $N$ to decay below the roundoff-error
    threshold. ($C$ is a normalizing constant.)  And third, rescaling
    the matrix $\mc J^\qua\big[q^\per\big]$ will change its singular
    values, possibly leading to new small singular values that do not
    correspond to bifurcation directions but instead to high-frequency
    inputs to $D_q\mc R\big[q^\per\big]$ that are compressed due to
    the change in norm from the input space $(\mc X^\e1_\sigma,0,0)$
    to the output space $\mc X^\e1_\rho$. By considering the
    $\sigma\to0^+$ limit, only the well-resolved singular values are
    small.
\end{remark}

\subsection{Computing the sign of the determinant of a matrix
  along with its singular values}
\label{sec:svd:det}

For simplicity, since the results of this section are not tied to the
water wave problem, we revert to standard numerical linear algebra
notation: $\mc J^\qua$ will be denoted by $A$; its dimension $2N+1$
will be denoted by $n$; and the singular values will be ordered so
that $\sigma_1\ge\sigma_2\ge\cdots\sigma_n\ge0$. In floating-point
arithmetic, we \emph{define} the sign of the determinant as a single
function (without computing $\det A$ as an intermediate result) to be
\begin{equation}\label{eq:sgn:det:A}
  \sgn\det A = (\det U)(\det V), \qquad
  A = U\Sigma V^T, \qquad \Sigma = \diag(\sigma_1,\dots,\sigma_n).
\end{equation}
This is a procedural definition: compute the SVD of $A$ numerically to
obtain $U$ and $V$, which are orthogonal. Then compute their
determinants by LU or QR factorization, round to 1 or $-1$, and
multiply them together.  As explained in Section~\ref{sec:detect:bif},
if $A$ depends analytically on a parameter $s$, then $\chi(s)=(\det
  U(s))(\det V(s))\sigma_n(s)$ will be a real analytic function that
does not have slope discontinuities at the zeros of
$\sigma_n(s)$. This conclusion relies on the \emph{existence} of an
analytic SVD, but it is only necessary to compute the standard SVD
with non-negative singular values. Our goal now is to show how to
compute $(\det U)(\det V)$ without actually forming the matrices $U$
and $V$ or computing their determinants explicitly.

Recall that the first step of the SVD algorithm is to
compute a bidiagonal reduction, e.g., using the
`dgebrd' routine in the LAPACK library:
\begin{equation}\label{eq:householder}
  U_0^T A V_0 = B_0.
\end{equation}
Here $U_0$ and $V_0$ are orthogonal matrices and $B_0$ is upper
bidiagonal. We will show that
\begin{equation}\label{eq:sgn:det}
  \sgn \det A = (\det U_0)(\det V_0)\sgn(\det B_0),
\end{equation}
where $\sgn(\det B_0)=\prod_{j=1}^n \sgn\big((B_0)_{jj}\big)$.
Here $(\det U_0)$ and $(\det V_0)$ are $\pm1$ with parity matching the
number of left and right Householder transformations performed in the
bidiagonal reduction, which are easy to count from the output of
`dgebrd'. The left-hand side of \eqref{eq:sgn:det} is still defined as
$(\det U)(\det V)$, but we wish to use \eqref{eq:sgn:det} as a cheaper
alternative.

Our task is now to analyze what would happen if we were to continue
with the standard algorithm to compute $U$ and $V$ along with
$\Sigma$. The next step of this standard algorithm is to call `dbdsqr'
to compute a sequence of upper bidiagonal matrices
\begin{equation}\label{eq:Bk}
  B_k=U_k^TB_{k-1}V_k, \qquad (k=1,2,3,\dots)
\end{equation}
that converge rapidly to a diagonal matrix $\td\Sigma$. In the initial
iterations, while searching for the smallest singular values, `dbdsqr'
employs a zero-shift in the implicit QR algorithm \cite{demmel:kahan}
or the mathematically equivalent `dqds' algorithm
\cite{fernando:parlett,demmel:book}. This leads to high relative
accuracy in all the computed singular values. Moreover, each iteration
of \eqref{eq:Bk} in floating-point arithmetic is equivalent to
introducing a small relative perturbation of each non-zero matrix
entry of $B_{k-1}$, performing a ``bulge-chasing'' sequence of Givens
rotations in exact arithmetic \cite{demmel:kahan, demmel:book}, and
then perturbing each non-zero entry of the result by a small relative
amount to obtain $B_k$. These perturbations of the diagonal and
superdiagonal entries of $B_{k-1}$ and $B_k$ do not affect the signs
of their determinants, and the Givens rotations all have unit
determinant, so $\sgn(\det B_k) = \sgn(\det B_{k-1})$.  On each
iteration, super-diagonal entries of $B_k$ that are sufficiently small
relative to their neighboring diagonal entries are zeroed out, which
does not affect $\sgn(\det B_k)$. The algorithm terminates and
$\td\Sigma$ is set to $B_k$ when the last super-diagonal entry is
zeroed out. Thus, $\sgn(\det\td\Sigma)=\sgn(\det B_0)$. At this point
we have
\begin{equation}\label{eq:A:td:UV}
  A = \td U\td\Sigma\td V^T, \qquad
  \td U = U_0U_1\cdots U_K, \qquad  \td V^T = V_K^T\cdots V_1^TV_0^T,
\end{equation}
where $K$ is the number of iterations required for convergence.  We
finally obtain
\begin{equation}\label{eq:A:USVT}
  A = U\Sigma V^T = (\td UP^T)(P\td\Sigma DP^T)(PD\td V^T), \qquad
  D = \diag\big(\sgn(\td\sigma_1),\dots,\sgn(\td\sigma_n)\big),
\end{equation}
where multiplying $\td\Sigma$ by $D$ takes the absolute values of
the diagonal entries and $P$ is a permutation matrix such that
$\Sigma=P\td\Sigma DP^T$ contains the singular values on the
diagonal in non-increasing order. Note that we have transferred the
signs on the diagonal of $\td\Sigma$ to the rows of $\td V^T$
via $D$. We conclude that
\begin{equation}
  \begin{aligned}
    (\det U_0)(\det V_0)\sgn(\det B_0) &=
    (\det U_0)(\det V_0)\sgn(\det\td\Sigma) = 
  (\det \td U)(\det\td V)\det(D) \\
  &=
  \det(\td UP^T)\det(PD\td V^T) =
      (\det U)(\det V),
  \end{aligned}
\end{equation}
where we used $\det(\td U)=\det(U_0)$ and $\det(\td V)=\det(V_0)$
since the matrices $U_1,\dots,U_K$ and $V_K^T,\dots,V_1^T$ are
comprised of Givens rotations of unit determinant.  This shows that we
can stop at \eqref{eq:sgn:det} and get the same result as continuing
to the completed computation of $(\det U)(\det V)$. Again, this is due
to the `dbdsqr' and `dqds' algorithms maintaining high relative
accuracy on the entries of successive bidiagonal matrices $B_k$.  The
smallest entries on the diagonal cannot jump across zero as this would
entail a large relative change.

Since $\sgn\det A$ is known already after the initial bidiagonal
reduction, it is not necessary to accumulate the Givens rotations to
form $\td U$ and $\td V^T$ in \eqref{eq:A:td:UV}, apply the
permutations to form $U$ and $V^T$ in \eqref{eq:A:USVT}, or compute
the determinants of $U$ and $V$ explicitly. The initial bidiagonal
reduction involves $(8/3)n^3+O(n^2)$ flops \cite{demmel:book} while
forming $U$ and $V$ and computing their determinants involves another
$(16/3)n^3+O(n^2)$ flops, which triples the running time.  Once a zero
of $\chi(s)$ has been found, we can form $V$ to find the null vector
of $A$ in $2n^3+O(n^2)$ flops, without also forming $U$ or computing
the determinants of $U$ and $V$.

\section{Numerical results}
\label{sec:num}
We now present two examples of quasi-periodic traveling waves that
bifurcate from finite-amplitude periodic traveling waves:
quasi-periodic gravity waves and overturning quasi-periodic
gravity-capillary waves.

\subsection{Quasi-periodic gravity waves}
\label{sec:num:grav}

As noted in the introduction, typical wave numbers for capillary waves
in the ocean are $10^7$ times greater than those of gravity waves, and
one does not expect to observe interesting nonlinear interaction
between component waves of such different length scales.  For ocean
waves, it is appropriate to set the surface tension coefficient to
zero, which removes the capillary wave branch \cite{djordjevic77} from
the dispersion equation \eqref{eq:disp:rel}. We are interested in
quasi-periodic waves in which the two wavelengths are comparable, so
we use $k_1=1$ and $k_2=k=1/\sqrt2$ for comparison with several of the
examples in \cite{quasi:trav}. Whereas the quasi-periodic solutions
computed in \cite{quasi:trav} persist to zero amplitude, the pure
gravity wave problem does not support genuinely quasi-periodic solutions in the
linearization about the zero solution since there is only one wave
number $k$ for a given wave speed $c$ in the dispersion relation
\eqref{eq:disp:rel}. Thus, we must search for secondary bifurcations.

Since surface tension is held fixed at $\tau=0$, this is a
one-parameter bifurcation problem. We use the wave height,
$h=\td\eta(0)-\td\eta(\pi)$, as the amplitude parameter as it
increases monotonically from the zero solution to the sharply crested
$120^\circ$ corner wave \cite{lhf:78}.  The blue curves in
Figure~\ref{fig:grav:main} show wave speed $c$ versus wave height $h$
(left panel and inset) and versus $\hat\eta_{1,0}$ (right panel and
  inset).  Here we follow the convention of Remark~\ref{rmk:eta10} and
write $\hat\eta_{1,0}$ rather than $\hat\eta_1$, even though
$\eta(\alpha_1,\alpha_2)=\td\eta(\alpha_1)$ is independent of
$\alpha_2$.  Each blue marker corresponds to a computed periodic
traveling wave. Note that both the wave speed $c$ and the first
Fourier mode $\hat\eta_{1,0}$ possess turning points beyond which they
no longer increase monotonically as one progresses further along the
primary bifurcation branch. The black markers correspond to
quasi-periodic solutions, and will be discussed below.  We increase
the Fourier cutoff $N_1$ in \eqref{eq:p:trav:1d} as needed to maintain
spectral accuracy in the computed traveling waves.

\begin{figure}
  \includegraphics[width=.92\textwidth]{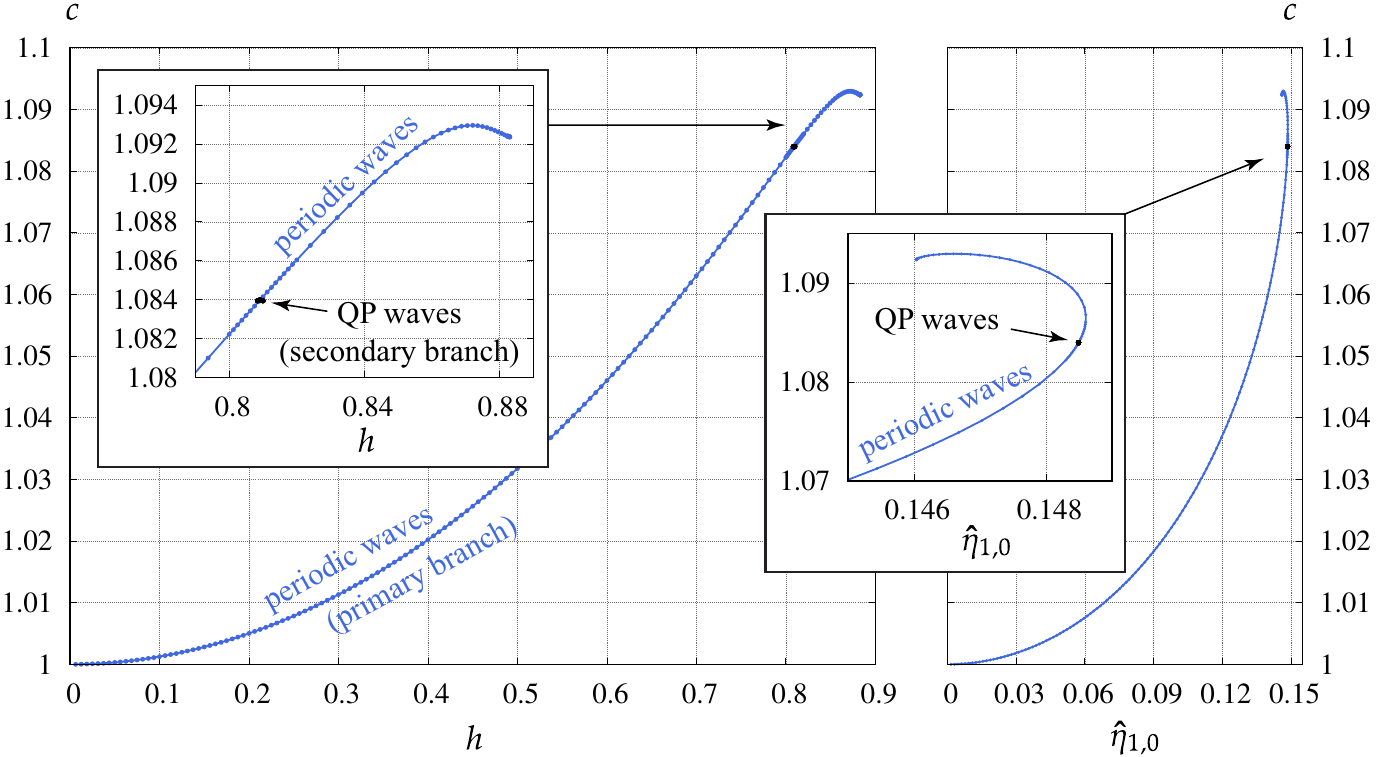}
  \caption{\label{fig:grav:main} Bifurcation diagrams of the primary
    branch (blue) and a secondary branch (black) for traveling water
    waves with zero surface tension. The wave height
    $h=\td\eta(0)-\td\eta(\pi)$ is preferable as an amplitude
    parameter to the wave speed $c$ or the first Fourier mode
    $\hat\eta_{1,0}$ on the primary branch as it increases
    monotonically all the way to the extreme $120^\circ$ corner
    wave. }
\end{figure}

Table~\ref{tbl:h:N1} gives the sequence of Fourier cutoff values $N_1$
used in the data of Figure~\ref{fig:grav:main} as well as the largest
wave height $h$ for which $N_1$ was used. We used $M_1=3N_1$
gridpoints in the pseudo-spectral computation of the products and
quotients in \eqref{eq:govern} and \eqref{eq:Jij:levmar} and for the
rows of the residual function $r_m$ in \eqref{eq:min:F:1d}. Also shown
in the table are the Fourier cutoff values $N$ used to truncate $\mc
J^\qua$ in \eqref{eq:J:qua} to $2N+1$ rows and columns. This
calculation also requires a grid on which to evaluate the products and
quotients in \eqref{eq:lin:l1l2}--\eqref{eq:uhat:l1l2}. For this we
used $M=3N$ gridpoints. In the last column of Table~\ref{tbl:h:N1},
with $N=N_1=32000$, we reduced $M_1$ to $65536$ and $M$ to $78732$ to
reduce the memory cost and running time, so aliasing errors may be
slightly higher in this final batch of results.

\begin{figure}
  \includegraphics[width=\textwidth]{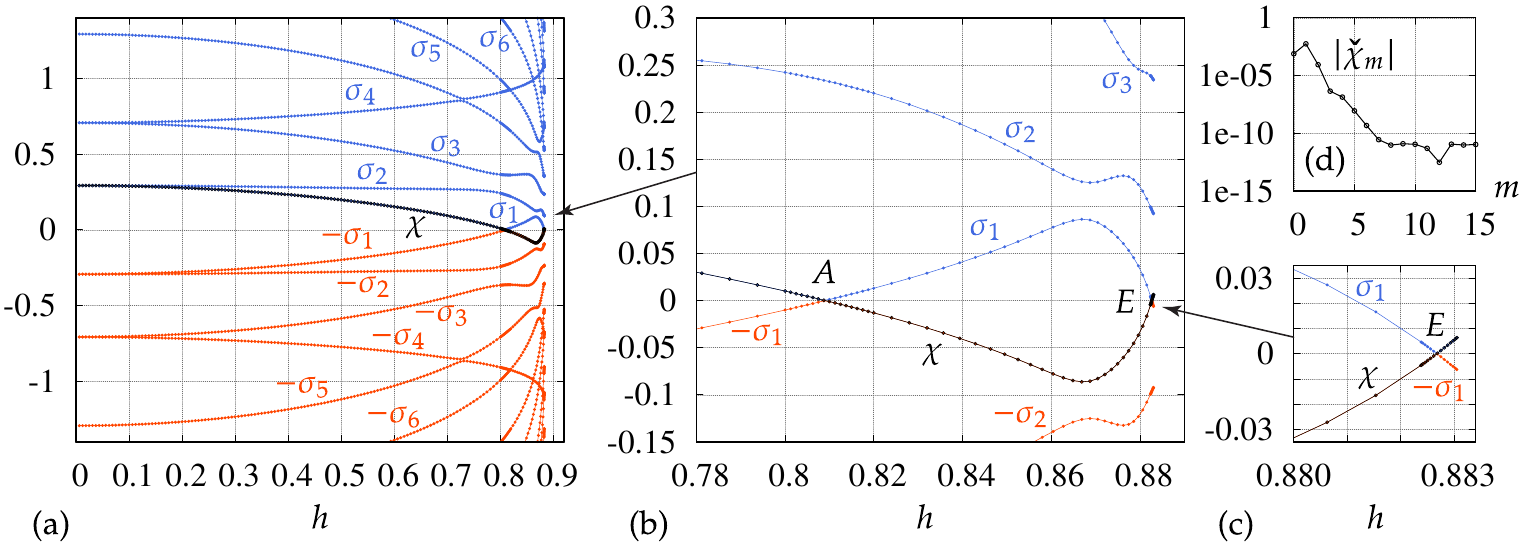}
  \caption{\label{fig:bif:svd} Plots of the smallest singular values
    of $\mc J^\qua$ versus $h$ (blue), their negations (red), the test
    function $\chi(h)$ (black), the bifurcation points $A$ and $E$,
    and the Chebyshev amplitudes $|\check\chi_m|$ obtained from the
    interpolation points shown in panel (c). Note that $\chi(h)$
    is a smooth function in spite of the many changes in Fourier
    cutoff $N_1$ and mesh size $M_1$ from Table~\ref{tbl:h:N1}
    represented in the graphs in panels (a)--(c).}
\end{figure}

\begin{table}[t]
\begin{equation*}
  \begin{array}{c|c|c|c|c|c|c|c|c|c|c}
  h & 0.075 & 0.2 & 0.3 & 0.4 & 0.575 & 0.65 & 0.725 & 0.8 & 0.82 &
  0.832 \\
  N_1 & 22 & 30 & 48 & 60 & 128 & 180 & 300 & 600 & 768 & 864 \\
  N  & 48 & 54 & 64 & 80 & 160 & 192 & 324 & 640 & 768 & 864 \\
  \hline
  h & 0.8394 & 0.8429 & 0.8525 & 0.8627 & 0.8687 & 0.8721 & 0.8763 &
  0.8786 & 0.8815 & 0.88305 \\
  N_1 & 1024 & 1350 & 1800 & 2400 & 4096 & 6144 & 8192 & 12288 &
  16384 & 32000
  \end{array}
\end{equation*}
\caption{\label{tbl:h:N1} Largest wave height $h$ for which the
  Fourier cutoff $N_1$ was used to compute the periodic waves on the
  primary branch in Figure~\ref{fig:grav:main}, and the cutoff $N$ that was
  used to truncate $\mc J^\per$ to $2N+1$ rows and columns. In the
  bottom row, $N=N_1$.}
\end{table}

Panels (a)--(c) of Figure~\ref{fig:bif:svd} show the first several
singular values $\sigma_j$ (blue) as well as $-\sigma_j$ (red),
keeping in mind the convention \eqref{eq:sigma:ord} that
$\sigma_{j+1}\ge\sigma_{j}\ge0$. These are plotted as functions of the
wave height $h$ for each of the periodic solutions corresponding to
the blue markers in Figure~\ref{fig:grav:main}. We find two
bifurcation points in the range $0\le h\le 0.88305$, which we label
\begin{equation}\label{eq:hA:hE}
  h_A = 0.8090707936918, \qquad\quad
  h_E = 0.882674234631.
\end{equation}
The corresponding wave speeds are
\begin{equation}\label{eq:cA:cE}
  c_A = 1.083977046908, \qquad\quad
  c_E = 1.09238325132.
\end{equation}
Here $E$ stands for ``extreme,'' as the wave profile of this
bifurcation is getting close to the limiting $120^\circ$ corner wave,
which has been computed accurately by Gandzha and Lukomsky
\cite{gandzha:07} and has a wave height of
$h_\text{max}=0.88632800992$.  The black curve in
Figure~\ref{fig:bif:svd} shows $\chi(h)$, which turns out to satisfy
\begin{equation}
  \chi(h) = \begin{cases} \phm\sigma_1(h), & h\le h_A \;\text{ or }\;
    h\in[h_E,0.88305], \\
    -\sigma_1(h), & h_A \le h\le h_E.
  \end{cases}
\end{equation}
There may be additional zero-crossings with $h>0.88305$, but we ran
out of computational resources to search for them.

We compute $h_A$ using Brent's method \cite{brent:73} starting with
the bracket $\chi(0.8)=0.0100259>0$ and
$\chi(0.82)=-0.0130748<0$. Brent's method uses a combination of linear
interpolation, inverse quadratic interpolation, and bisection to
rapidly shrink the bracket to a zero of the function without
derivative evaluations.  In this example, only 7 additional function
evaluations were needed to converge, with $\chi(h)$ taking on the
values $4.5\times10^{-4}$,
$-4.7\times10^{-9}$, $7.2\times10^{-12}$, $1.5\times10^{-13}$,
$-2.3\times10^{-13}$, $1.1\times10^{-14}$, $-2.9\times10^{-15}$. The
last value corresponds to $h=h_A$ reported in \eqref{eq:hA:hE}. We
used $N_1=768$ and $M_1=2304$ in the traveling wave calculation and
$2N+1=1537$ for the dimension of $\mc J^\qua$ in the SVD calculation
when computing $\chi(h)$ inside Brent's method. The total running time
of the 9 function evaluations was 21.7 seconds on a workstation with
two 12-core 3.0 GHz Intel Xeon Gold 6136 processors.

To demonstrate an alternative approach, we compute $h_E$ using
polynomial interpolation. First, we identify a bracket with
$\chi(0.88238)=-0.00467<0$ and $\chi(0.88305)=0.00638>0$. We then
evaluate $\chi$ on a 16-point Chebyshev-Gauss grid over the interval
$[0.88238,0.88305]$. Unlike Brent's method, this can be done in
parallel, though we did not have the computational resources to
do this. From these values, we obtain the expansion
\begin{equation}\label{eq:cheb:expand}
  \chi(h) \approx \acute\chi(h) = \sum_{m=0}^{15} \check\chi_m T_m\Big(
    2\frac{h-0.88238}{0.88305-0.88238} - 1 \Big).
\end{equation}
Panel (d) of Figure~\ref{fig:bif:svd} shows the Chebyshev mode
amplitudes $|\check\chi_m|$. It appears that the modes decay rapidly
up to $m=7$, and then start to be corrupted by floating-point
arithmetic errors. So we truncate the series \eqref{eq:cheb:expand} by
reducing the upper limit from 15 to 7 and then use Newton's method on
$\acute\chi(h)$ to obtain $h_E$ in \eqref{eq:hA:hE}. A final
evaluation of the original function $\chi$ (not its polynomial
  approximation) yields $\chi(h_E)=2.1\times10^{-12}$. The relatively
large floating-point errors visible in the high-frequency Chebyshev
modes and the larger value of $|\chi(h_E)|$ relative to $|\chi(h_A)|$
are due to the increase in problem size.  For the $h_E$ calculation,
we used $N_1=N=32000$, $M_1=65536$ and $M=78732$. We discuss
floating-point errors in the smallest singular value (and hence in
  $\chi$) in Appendix~\ref{sec:fp:err}.

\begin{remark}
  It would have been better (though not worth redoing) to use a nested
  set of Chebyshev-Lobatto grids with $2^n+1$ points. We could have
  stopped at $n=3$ rather than guessing that 16 points would be enough
  to resolve $\chi(h)$ with spectral accuracy, which turned out to be
  overkill. Each $\chi(h)$ evaluation involves computing a traveling
  wave and then computing the SVD of $\mc J^\qua$. At this problem
  size, each traveling wave calculation takes 45 minutes on one large
  memory node of the Lawrencium cluster at Lawrence Berkeley National
  Laboratory (LBNL) while the SVD takes 50 minutes on 15 standard
  memory nodes, using ScaLapack. Each node has 32 cores (2.3 GHz) and
  either 96 GB or 1.6 TB of memory.
\end{remark}

\begin{remark}\label{rmk:large:chebyshev}
  If the problem is so large that even the bidiagonalization phase of
  the SVD is prohibitively expensive, one can compute only the
  smallest singular values using bisection and inverse iteration
  \cite{demmel:book,schwetlick:03}. In this case, one can still use
  Chebyshev polynomials to represent $\chi(s)$ and locate bifurcation
  points, but a sign has to be added by hand to some of the values of
  $\sigma_\text{min}(s_i)$ to convert them into $\chi(s_i)$, where
  $s_i$ are the Chebyshev-Lobatto gridpoints. This is usually easy by
  plotting both $\sigma_\text{min}(s_i)$ and $-\sigma_\text{min}(s_i)$
  on the same plot, as in panel (c) of Figure~\ref{fig:bif:svd}, and
  looking at the graphs to determine where they cross zero. Usually at
  most one point has an ambiguous sign, and it is easy to tell which
  sign is correct since the Chebyshev modes will only decay rapidly
  when $\chi(s)$ has been sampled with the correct signs.  One could
  potentially automate this process using continuous invariant
  subspace (CIS) methods \cite{demmel01,dieci01,friedman01,beyn01,
    bindel08} to avoid sign flips in the corresponding singular
  vectors, or by introducing signs from one end of the list of
  $\chi(s_i)$ values until the Chebyshev modes suddenly decay rapidly.
\end{remark}

\begin{figure}
  \begin{center}
    \includegraphics[width=.865\textwidth]{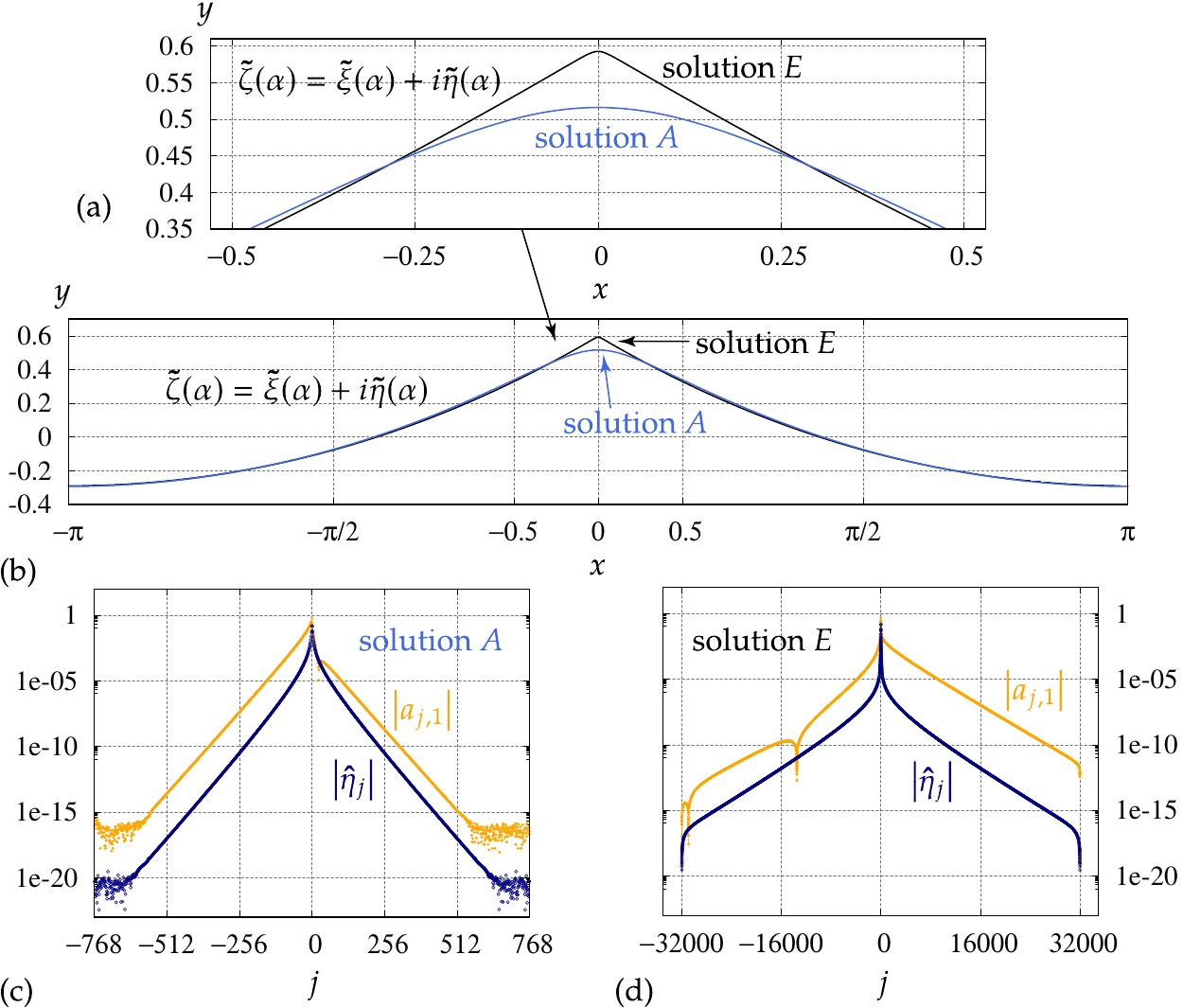}
  \end{center}
  \caption{\label{fig:grav:bif:soln} Wave profiles of solutions $A$
    and $E$, where quasi-periodic bifurcations with $k=1/\sqrt2$ were
    detected in Figure~\ref{fig:bif:svd}, along with their Fourier
    mode amplitudes $\big|\hat\eta_{j}\big|$ and the amplitudes of
    the components $a_{j,1}$ of the null vector $\dot q^\qua$ in
    \eqref{eq:q:qua:from:v}.}
\end{figure}

Panels (a) and (b) of Figure~\ref{fig:grav:bif:soln} show the wave
profiles $\td\zeta(\alpha)=\td\xi(\alpha)+i\td\eta(\alpha)$ of the
periodic traveling waves with wave heights $h_A$ and $h_E$ in
\eqref{eq:hA:hE}. Panel (b) shows a full period of the wave while
panel (a) shows a closer view at the crest tip. The aspect ratio of
both plots is~$1$ to demonstrate how close solution $E$ is to the
$120^\circ$ corner wave. The mean fluid height of both waves is 0,
with solution $E$ higher close to the crest tip and solution $A$
higher beyond the crossing points visible in panel (a).  Panels (c)
and (d) show the Fourier mode amplitudes $\big|\hat\eta_{j}\big|$
of these two solutions as well as the components $a_{j,1}$ of the null
vector $\dot q^\qua$ expressed in the $\varphi_{j,1}$ basis, as in
\eqref{eq:q:qua:from:v}. Solution $A$ is resolved to double-precision
accuracy using $N_1=N=768$. Indeed, the mode amplitudes in panel (c)
decay through at least 15 orders of magnitude with floating-point
errors evident in the highest-frequency modes. Solution $E$ is nearly
fully resolved with $N_1=N=32000$, but the mode amplitudes in panel
(d) have not decayed all the way to the point that roundoff effects
become visible. We did not attempt to increase $N_1$ and $N$ further
due to the computational expense.

We next search for solutions on the quasi-periodic branch
$\mfq^\qua(\theta)$ that intersects $\mfq^\per(s)$ at solution $A$. We
have been using $s=h$, the wave height, as the amplitude parameter on
the primary branch and will switch to $\theta=\hat\eta_{0,1}$ on the
secondary branch. Note that $\theta=0$ for all the solutions on the
primary branch. Another mode $\hat\eta_{j,1}$ with $j\ne0$ could have
been used instead, but $j=0$ turns out to maximize $|a_{j,1}|$ in the
expansion of $\dot q^\qua$ in \eqref{eq:q:qua:from:v}. The lowest-frequency
coefficients, normalized so that $\sum_{|j|\le N} |a_{j,1}|^2=1$, are
\begin{equation}\label{eq:aj1:A}
  \begin{array}{c||c|c|c|c|c|c|c|c|c}
    j & -4 & -3 & -2 & -1 & 0 & 1 & 2 & 3 & 4  \\ \hline
    a_{j,1} & 0.266 & 0.306 & 0.340 & -0.239  & -0.565  & -0.193  &
    -0.102 & -0.062 & -0.042
    \end{array}.
\end{equation}
Note that $a_{j,1}$ is not symmetric about $j=0$, which is also
evident in panel (c) of Figure~\ref{fig:grav:bif:soln}. As explained
at the end of Section~\ref{sec:num:alg}, we use
\begin{equation}\label{eq:q:guess}
  \mfq^\qua(\theta) \approx q^\guess = q^\bif + \veps \dot q^\qua, \qquad\quad
  \veps = \frac\theta{a_{0,1}}, \qquad \theta=\hat\eta_{0,1},
\end{equation}
as an initial guess for the first point on the bifurcation path, where
$q^\bif=(\eta_A,\tau_A,b_A)$ with $\tau_A=0$ and $b_A=c_A^2$ from
\eqref{eq:cA:cE}, and $\dot q^\qua=(\dot\eta^\qua,0,0)$. The leading
2D Fourier modes of $\eta^\guess=\eta_A+\veps\dot\eta^\qua$ are given
by
\begin{equation}\label{eq:eta:hat:guess}
  \Big( \hat\eta^\guess_{j_1,j_2} \Big) = 
  \begin{array}{r||c|c|c|c|c|c}
    j_1   & 0 & 1 & 2 & 3 & 4 & \dots \\ \hline
    j_2\ge2 & 0 & 0 & 0 & 0 & 0 & \dots \\
    \raisebox{-1pt}{$j_2=1$} & \veps a_{0,1}  & \veps a_{1,1}  & \veps a_{2,1}  &
    \veps a_{3,1}  & \veps a_{4,1} & \dots \\[3pt]
    \raisebox{1pt}{$j_2=0$} & -0.0631 & 0.1485  & 0.0496 & 0.0253 & 0.0155 & \dots \\[-1pt]
    j_2=-1 & \veps a_{0,1}  & \veps a_{-1,1}  & \veps a_{-2,1}  & \veps a_{-3,1}  &
    \veps a_{-4,1} & \dots \\
    j_2\le-2 & 0 & 0 & 0 & 0 & 0 & \dots
  \end{array}
\end{equation}
where we make use of \eqref{eq:assump:eta} to obtain
$\hat\eta^\guess_{j,-1}=\hat\eta^\guess_{-j,1}=\veps a_{-j,1}$ for
$j\ge0$. We use the `r2c' and `c2r' routines of FFTW, which take
advantage of \eqref{eq:assump:eta} to avoid storing Fourier modes with
$j_1<0$.

Since the kernel of $D_q\mc R\big[q^\bif\big]$ in
\eqref{eq:span:dot:q} is two-dimensional, it might have been necessary
to include a term $\veps_2\dot q^\per$ in \eqref{eq:q:guess} with
$\veps_2$ depending linearly on $\theta$ and $\veps$.  But solutions
$\mfq^\qua(\theta)=\big(\eta^\qua(\theta),0, b^\qua(\theta)\big)$ on
this path turn out to have the symmetry property
\begin{equation}\label{eq:symmetry:alpha2}
  \eta^\qua(-\theta)(\alpha_1,\alpha_2) =
  \eta^\qua(\theta)(\alpha_1,\alpha_2+\pi), \qquad
  b^\qua(-\theta) = b^\qua(\theta).
\end{equation}
Since all the nonzero Fourier modes
$\dot{\hat\eta}^\qua_{j_1,j_2}$ have $j_2\in\{1,-1\}$, we see that
\begin{equation}
  (-\veps)\dot \eta^\qua(\alpha_1,\alpha_2) =
  \veps \dot \eta^\qua(\alpha_1,\alpha_2+\pi),
\end{equation}
so changing the sign of the perturbation in \eqref{eq:q:guess} is
equivalent to shifting by $\pi$ in the $\alpha_2$ direction. Including
$\veps_2\dot q^\per$ in the formula for $q^\guess$ would break this
symmetry if $\veps_2$ were a non-zero multiple of $\veps$.

In Figure~\ref{fig:grav:bif:branch}, we plot 3 bifurcation curves
showing different aspects of how the secondary branch, plotted in
black, splits from the primary branch, plotted in blue. The black
markers between the bifurcation point $A$ and the point labeled $B$
correspond to solutions computed by minimizing the objective function
\eqref{eq:min:F} holding $\theta=\hat\eta_{0,1}$ at fixed values. The
markers on the other side, between $A$ and $B'$, were obtained from
these solutions using the symmetry \eqref{eq:symmetry:alpha2} rather
than carrying out the minimization again. The decision to explore
negative values of $\theta$ first was arbitrary. Exactly the same
results would have been obtained in the other direction (aside from
  swapping the labels $B$ and $B'$). We define the wave height $h$
plotted in panels (a) and (b) as $\eta(0,0)-\eta(\pi,0)$. This choice
will be discussed and justified in Remark~\ref{rmk:hQP} below.

\begin{figure}
  \includegraphics[width=\textwidth]{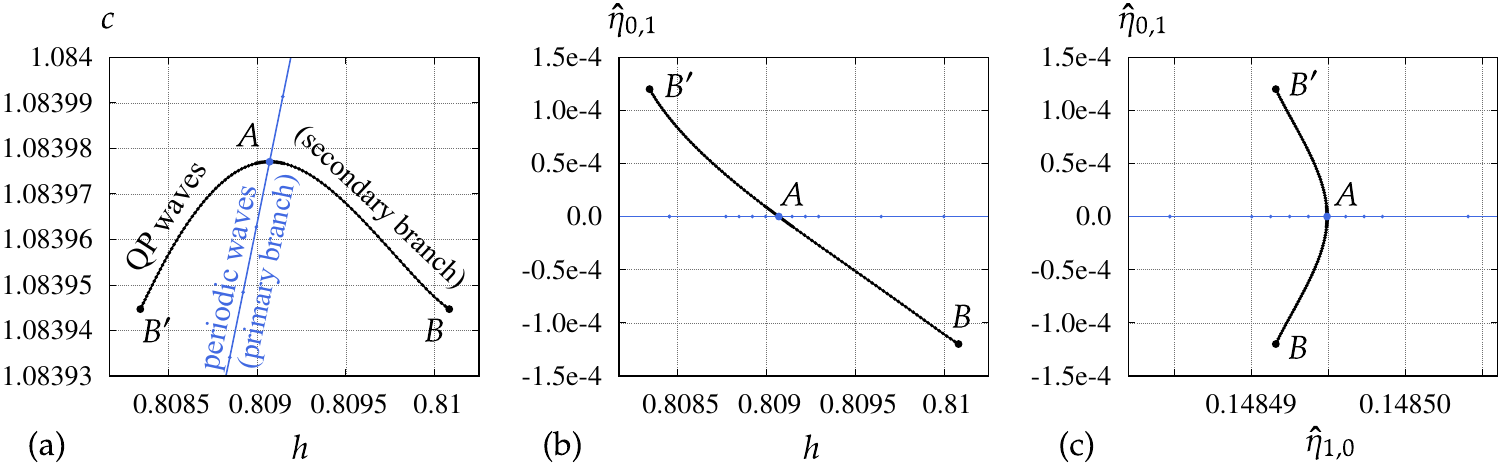}
  \caption{\label{fig:grav:bif:branch} The secondary branch of
    quasi-periodic traveling waves (black) bifurcates from the primary
    branch of periodic waves (blue) at solution $A$. The solutions
    between $A$ and $B'$ are related to those between $A$ and $B$ by
    the symmetry \eqref{eq:symmetry:alpha2}. Panel (a) shows a highly
    magnified view of the plots in the left panel of
    Figure~\ref{fig:grav:main}.}
\end{figure}

For the first point on the secondary branch, we set
$\veps=1.768\times10^{-7}$, which corresponds to
$\theta=\hat\eta_{0,1}=-1.0\times10^{-7}$. We then minimize the
objective function \eqref{eq:min:F} holding $\hat\eta_{0,1}$ fixed.
The initial guess for the second point can be computed in the same
way, by doubling $\veps$, or from linear extrapolation from the zeroth
point ($q^\bif$) and the first point. Both methods work well.  The
starting guess for each additional solution on the path is computed
via linear extrapolation from the previous two solutions. As we
progress along the path, we monitor the amplitudes of the 2D Fourier
modes and increase $N_2$ as necessary to maintain spectral
accuracy. The first several solutions are still very close to the
periodic traveling wave, so $N_2=3$ and $M_2=8$ are sufficient.  The
traveling wave requires $N_1=550$, $M_1=1200$ to achieve spectral
accuracy, so we held these fixed in the quasi-periodic
calculation. Presumably $N_1$ and $M_1$ will need to be increased if
one proceeds far enough along the path, but we always ran out of
resolution in $N_2$ and $M_2$ first.  The mesh parameters used for
different values of $\theta=\hat\eta_{0,1}$ in our calculation are as
follows:
\begin{equation*}
  \begin{array}{r|cccccccc|c}
    N_2 & 3 & 6 & 6 & 8 & 12 & 18 & 20 & 48 & \\
    M_2 & 8 & 16 & 16 & 24 & 36 & 48 & 48 & 108 & \\ \hline
    \theta_0 & -0.1 & -1 & -10 & -20 & -40 & -70 & -100 & -120 &  \times 10^{-6} \\
    \theta_1 & -1 & -10 & -20 & -40 & -70 & -100 & -120 & -120 & \times 10^{-6} \\
    \Delta\theta & -0.1 & -1 & -2.5 & -2.5 & -2.5 & -2.5 & -2.5 & \text{---} &  \times 10^{-6}
  \end{array}
\end{equation*}
Each column corresponds to a batch of solutions computed by numerical
continuation by taking equal steps of size $\Delta\theta$ from
$\theta=\theta_0$ to $\theta=\theta_1$. The factors of $10^{-6}$ apply
to all the entries in the bottom 3 rows. For each new column,
$\theta_0$ is the same as $\theta_1$ from the previous column, which
means we recompute the solution with larger values of $N_2$ and $M_2$
using the previous solution as a starting guess.  In the last column
(with $N_2=48$), we simply refine the solution at
$\theta=-1.2\times10^{-4}$ without progressing further along the path.

In Figure~\ref{fig:fourier:big}, we plot the leading Fourier mode
amplitudes $\big|\hat\eta_{j_1,j_2}\big|$ for solution $B$. The same
data is plotted from three viewpoints in panels (a), (b) and (c).  The
grid is $M_1\times M_2=1200\times108$, and the solver searches for
modes $\hat\eta_{j_1,j_2}$ with $0\le j_1\le N_1=550$ and $|j_2|\le
N_2=48$. The remaining modes with $551\le j_1\le600$ and
$49\le|j_2|\le 54$ are set to zero in the nonlinear least squares
solver but retained in the FFT calculations to reduce aliasing errors.
The solver does not control $\hat\eta_{0,1}$, which is held fixed at
$-1.2\times10^{-4}$. The modes with $j_1<0$ are assumed to satisfy
\eqref{eq:assump:eta}. In the figure, we truncate $j_1$ at $50$, but
the modes are non-zero out to $j_1=550$, and continue to decay along
slices of constant $j_2$ in a similar way to the modes of the periodic
wave $A$, shown in Figure~\ref{fig:grav:bif:soln}. The objective
function $F$ in \eqref{eq:min:F} has been minimized to
$1.9\times10^{-25}$ for this solution. We did not try to compute more
than one solution at this resolution as each calculation is quite
expensive. The Jacobian $\mc J$ in \eqref{eq:Jij:levmar} has 129600
rows and 53398 columns and has to be factored at least once via the
SVD in our implementation of the Levenberg-Marquardt method. We used
ScaLapack for this purpose, and ran the code on 18 nodes (432 cores)
of the Savio cluster at UC Berkeley. The running time was 70 minutes.
We also did not attempt to follow the bifurcation branches at
solution $E$ as the periodic problem already involves 32000 active
Fourier modes.

\begin{figure}
  \includegraphics[width=\textwidth,trim=12pt 4pt 17pt 0pt,clip]{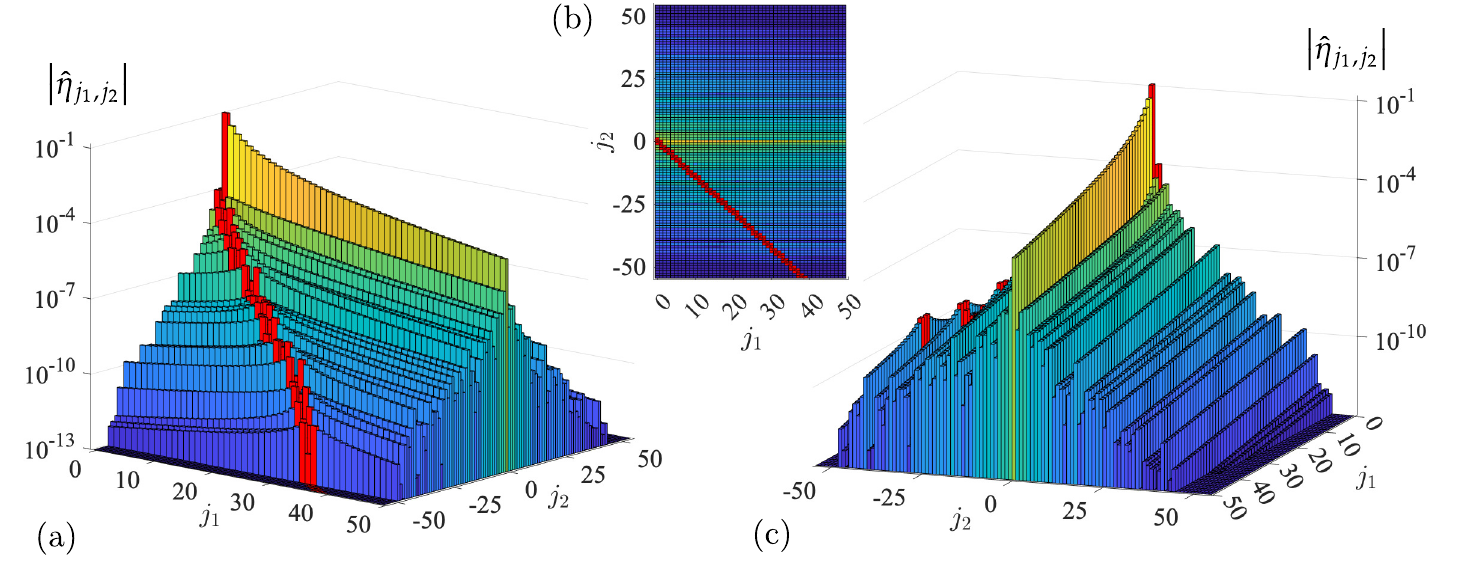}
  \caption{\label{fig:fourier:big} Fourier mode amplitudes
    $\big|\hat\eta_{j_1,j_2}\big|$ for solution $B$ of
    Figure~\ref{fig:grav:bif:branch}. The three plots show the same
    data from different viewpoints. The modes were truncated at
    $j_1=50$ in the plots, but extend to $j_1=550$ in the
    calculation. The resonant line $j_1+kj_2=0$ is plotted in
    dark red in panel (b). The closest lattice points to
    this line are plotted in red in all three panels. }
\end{figure}

Panel (b) of Figure~\ref{fig:fourier:big} shows the Fourier amplitude
data in a two-dimensional view using the same colormap as panels (a)
and (c). Also plotted, in dark red, is the resonant line
$j_1+j_2k=0$. Since $k=1/\sqrt2$ is irrational, the only lattice point
lying precisely on this line is $(j_1,j_2)=(0,0)$. As discussed in
\cite{quasi:trav}, lattice points $(j_1,j_2)$ close to this line
correspond to plane waves $e^{i(j_1+j_2k)\alpha}$ of long
wavelength. The residual function $\mc R$ in \eqref{eq:Rq} is
unchanged if $\eta$ is perturbed by a constant function, and
changes little for long wavelength perturbations. Thus, the above
plane waves are approximate
null vectors of the Jacobian \eqref{eq:Jij:levmar}. Such ``small
divisors'' have been overcome in similar problems
\cite{berti2016quasi,baldi2018time,berti2020traveling,feola2020trav}
building on Nash-Moser theory \cite{plotnikov01,iooss05}, though so
far always at small-amplitude, near the zero solution. Adapting these
rigorous techniques to the spatially quasi-periodic setting is a
challenging open problem, especially in the present case of
bifurcations from finite-amplitude periodic waves.

We avoid running into small divisors in our search for bifurcations by
restricting $\ker D_q\mc R[q^\bif]$ to $\mc D_\sigma^\e{1}$ in
\eqref{eq:J:qua}.  But in fully nonlinear calculations such as
solution $B$ of Figures~\ref{fig:grav:bif:branch}
and~\ref{fig:fourier:big}, which has many active modes in both the
$j_1$ and $j_2$ directions, one can see some effects of the small
divisors on the Fourier modes corresponding to lattice points near the
resonant line, which we plotted in red in all three panels of
Figure~\ref{fig:fourier:big}. In panel (a) we see that the modes
$\hat\eta_{j_1,j_2}$ with $j_2<0$ grow in amplitude as $j_1$ increases
to the resonant line (holding $j_2$ fixed), and then decay
afterward. By contrast, we see in panel (c) that on the ``back side''
(with $j_2>0$), the modes appear to generally decay monotonically
right away as $j_1$ increases.  If we instead decrease $j_1$ through
negative values with $j_2$ held fixed, the modes with $j_2>0$ are the
ones that increase in magnitude until $j_1$ crosses the resonant line
while the modes with $j_2<0$ will generally decay right away as
$|j_1|$ increases.  It is not necessary to plot this as it follows
from the data shown in Figure~\ref{fig:fourier:big} and the symmetry
\eqref{eq:assump:eta}, namely
$\hat\eta_{-j_1,-j_2}=\hat\eta_{j_1,j_2}$.

\begin{figure}
  \includegraphics[width=\textwidth]{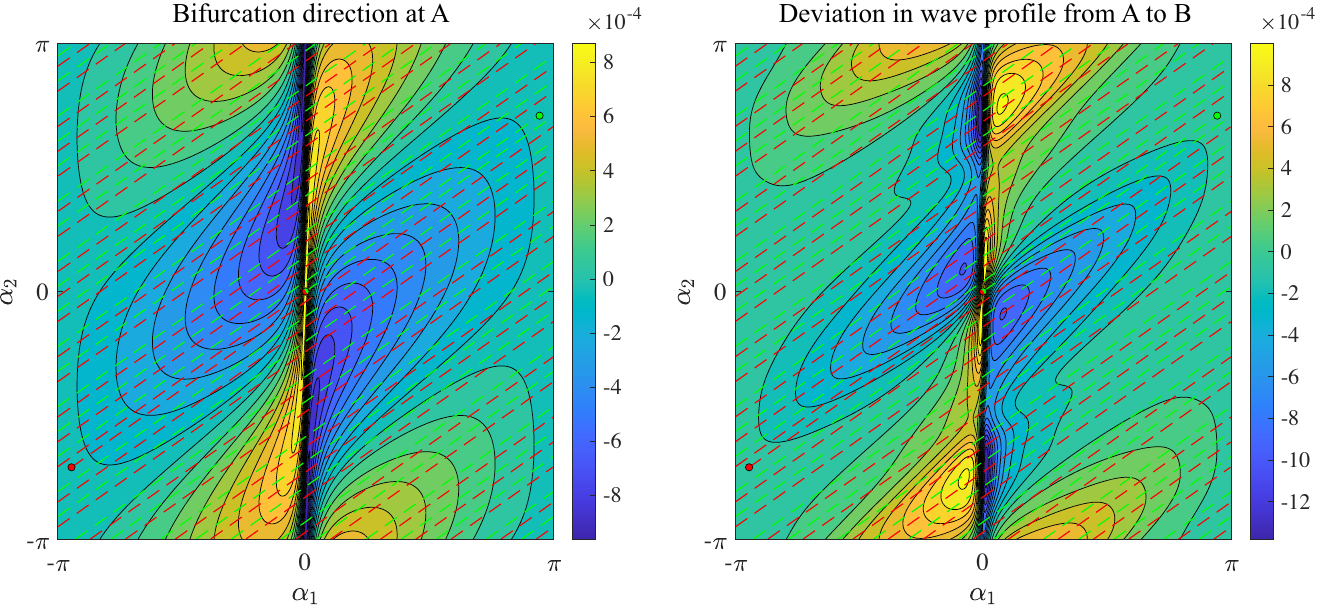}
  \caption{\label{fig:torus:big} Contour plots of the bifurcation
    direction $C\dot\eta^\qua(\alpha_1,\alpha_2)$ and the deviation
    $\eta_B^\text{dev}(\alpha_1,\alpha_2)$ from solution $A$ to
    solution $B$. The dashed lines show how the characteristic line
    $(\alpha,k\alpha)$ wraps around the periodic torus for
    $\alpha\ge0$ (green) and $\alpha\le0$ (red).  Evaluation along
    this characteristic line yields the one-dimensional quasi-periodic
    functions shown in panels (c) and (d) of
    Figure~\ref{fig:quasi:wave}.  }
\end{figure}

Figure~\ref{fig:torus:big} shows torus views of the bifurcation
direction at solution $A$, namely $\dot\eta^\qua(\alpha_1,\alpha_2)$
from \eqref{eq:q:qua:from:v}, and the deviation in the wave profile
from solution $A$ to solution $B$,
\begin{equation}
  \eta_B^\text{dev}(\alpha_1,\alpha_2) =
  \eta_B(\alpha_1,\alpha_2)-\eta_A(\alpha_1,\alpha_2),
\end{equation}
where $\eta_A(\alpha_1,\alpha_2)$ is independent of $\alpha_2$.  In
the left panel, we actually plot $C\dot\eta^\qua$, where the
normalization constant $C$ is chosen to minimize the distance from
$C\dot\eta^\qua$ to $\eta_B^\text{dev}$ in $L^2$ on the torus, which
turns out to be
\begin{equation}\label{eq:C:dir}
  C=\frac{\la\dot \eta^\qua,\eta_B^\text{dev}\ra_{L^2(\mbb
    T^2)}}{\la\dot \eta^\qua,\dot \eta^\qua\ra_{L^2(\mbb T^2)}} =
  2.121\times10^{-4}.
\end{equation}
While the bifurcation direction predicts the large-scale features of
$\eta_B^\text{dev}$, there are clear differences in the two contour
plots. In particular, the symmetry in the left panel in which
$\dot\eta^\qua(\alpha_1,\alpha_2)$ changes sign on shifting $\alpha_2$
by $\pi$, which occurs due to $\dot{\hat\eta}_{j_1,j_2}=0$ for
$j_2\not\in\{-1,1\}$, is broken in the right panel. Indeed,
$\eta_B^\text{dev}$ has a richer Fourier structure consisting of the
modes plotted in Figure~\ref{fig:fourier:big} minus the modes of
solution $A$, $\hat\eta_{j_1,j_2}=\hat\eta_{j}\delta_{j_2,0}$, where
$\hat\eta_j$ is plotted in Figure~\ref{fig:grav:bif:soln}. Replacing
$\alpha_2$ by $\alpha_2+\pi$ in $\eta_B^\text{dev}(\alpha_1,\alpha_2)$
yields $\eta_{B'}^\text{dev}(\alpha_1,\alpha_2)$, where $B'$ is the
solution at the other end of the secondary bifurcation branch in
Figure~\ref{fig:grav:bif:branch}.

\begin{remark}\label{rmk:hQP}
  For all the solutions on the path from $A$ to $B$ in
  Figure~\ref{fig:grav:bif:branch}, which correspond to negative
  values of $\theta=\hat\eta_{0,1}$, we find (by studying the
    numerical results) that the maximum and minimum values of
  $\eta(\theta)(\alpha_1,\alpha_2)$ occur at
  \begin{equation}
    \operatorname{argmax}\eta(\theta)(\alpha_1,\alpha_2) = (0,0), \qquad
    \operatorname{argmin}\eta(\theta)(\alpha_1,\alpha_2) = (\pi,0), \qquad
    (\theta<0).
  \end{equation}
  Thus, we define the wave height $h$ on this quasi-periodic
  branch as
  \begin{equation}\label{eq:h:QP:def}
    h(\theta) = \eta(\theta)(0,0) - \eta(\theta)(\pi,0).
  \end{equation}
  However, on the path from $A$ to $B'$, where $\theta>0$, the maximum
  and minimum values occur at $(0,\pi)$ and $(\pi,\pi)$, respectively,
  due to \eqref{eq:symmetry:alpha2}. So if we define
  $h_1(\theta)=\eta(\theta)(0,\pi)-\eta(\theta)(\pi,\pi)$, then the
  physical wave height is $h_1(\theta)$ for $\theta>0$ and $h(\theta)$
  for $\theta<0$. Both $h_1(\theta)$ and $h(\theta)$ can be read off
  of panels (a) and (b) of Figure~\ref{fig:grav:bif:branch}, since
  $h_1(\theta)=h(-\theta)$. So it is preferable to plot $h(\theta)$
  for positive and negative values of $\theta$, as we have done,
  rather than plotting the physical wave height, which would introduce
  a slope discontinuity at the bifurcation point in panels (a) and (b)
  of Figure~\ref{fig:grav:bif:branch}. Moreover, it would discard the
  information about $h(\theta)$ with $\theta>0$, replacing it by
  $h_1(\theta)$, which is already known from $h(\theta)$ with
  $\theta<0$.  We will simply refer to $h(\theta)$ in
  \eqref{eq:h:QP:def} as ``the wave height.''
\end{remark}

Our next goal is to plot the 1D quasi-periodic functions obtained by
evaluating the torus functions of Figure~\ref{fig:torus:big} along the
dashed red and green lines, namely
\begin{equation}\label{eq:eta:dev:1d}
  C\dot{\td\eta}(\alpha) = C\dot\eta^\qua(\alpha,k\alpha), \qquad
  \td\eta_B^\text{dev}(\alpha) = \eta_B^\text{dev}(\alpha,k\alpha).
\end{equation}
We will plot them as functions of $\alpha$ rather than in the
parametric form used in Figure~\ref{fig:grav:bif:soln}. This allows
for a simpler correspondence with the torus functions of the
  conformal mapping formulation of \eqref{eq:govern} and avoids the
complication of solutions $A$ and $B$ having slightly different
parameterizations $\td\xi_A(\alpha)$ and $\td\xi_B(\alpha)$ in the
$x$-direction. One would have to transform to a graph-based
formulation of the problem to define analogues of $C\dot{\td\eta}$
and $\td\eta_B^\text{dev}$ in \eqref{eq:eta:dev:1d} that are functions
of $x$ rather than $\alpha$. Panel (a) of
Figure~\ref{fig:alphaXi:rat} shows that $\td\eta(\alpha)$
is more sharply peaked than the physical wave profile obtained by
plotting $\td\zeta(\alpha)=\td\xi(\alpha)+i\td\eta(\alpha)$
parametrically.

We wish to plot the functions in \eqref{eq:eta:dev:1d} over many
cycles of the underlying periodic wave without losing resolution due
to excessive horizontal compression of the plot. We do this by
plotting the results on a periodic domain with a period that differs
from $2\pi$, the period of the underlying Stokes wave. To select a
useful period for the plot, we consider best rational approximations
of $k$.  Panel (b) of Figure~\ref{fig:alphaXi:rat} shows the
fractional part
\begin{equation}\label{eq:frac:part}
  \{Qk\} = Qk-\lfloor Qk\rfloor, \qquad Q=1,2,3,\dots,50
\end{equation}
of the first 50 integer multiples of the second wave number $k$, which
is $k=1/\sqrt2$ in the present calculation. Here $\lfloor
\cdot\rfloor$ is the floor function. We also write
$\lfloor\cdot\rceil$ for the function that rounds its argument up or
down to the nearest integer. Given a positive integer $Q$, the closest
rational number of the form $P/Q$ to $k$ is found by setting
$P=\lfloor Qk\rceil$. With this choice of $P$, let
\begin{equation}
  e_1(Q) = \big|P - Qk\big| = \min(\{Qk\},1-\{Qk\}), \qquad
  e_2(Q) = \left|\frac PQ -k\right| = \frac{e_1(Q)}{Q}.
\end{equation}
The black markers in the right panel of Figure~\ref{fig:alphaXi:rat}
minimize $e_1(Q)$ over previously seen values, i.e., $e_1(Q)\le e_1(q)$ for
$1\le q\le Q$. They correspond to the rational approximations
\begin{equation}
  k \approx \frac11, \qquad
  k \approx \frac23, \qquad
  k \approx \frac57, \qquad
  k \approx \frac{12}{17}, \qquad
  k \approx \frac{29}{41},
\end{equation}
where $k=1/\sqrt2$. If we had minimized $e_2(Q)$ instead, as is
usually done in defining best rational approximations, $1/2$ and
$7/10$ would be added to the list.

\begin{figure}
  \includegraphics[width=.875\textwidth]{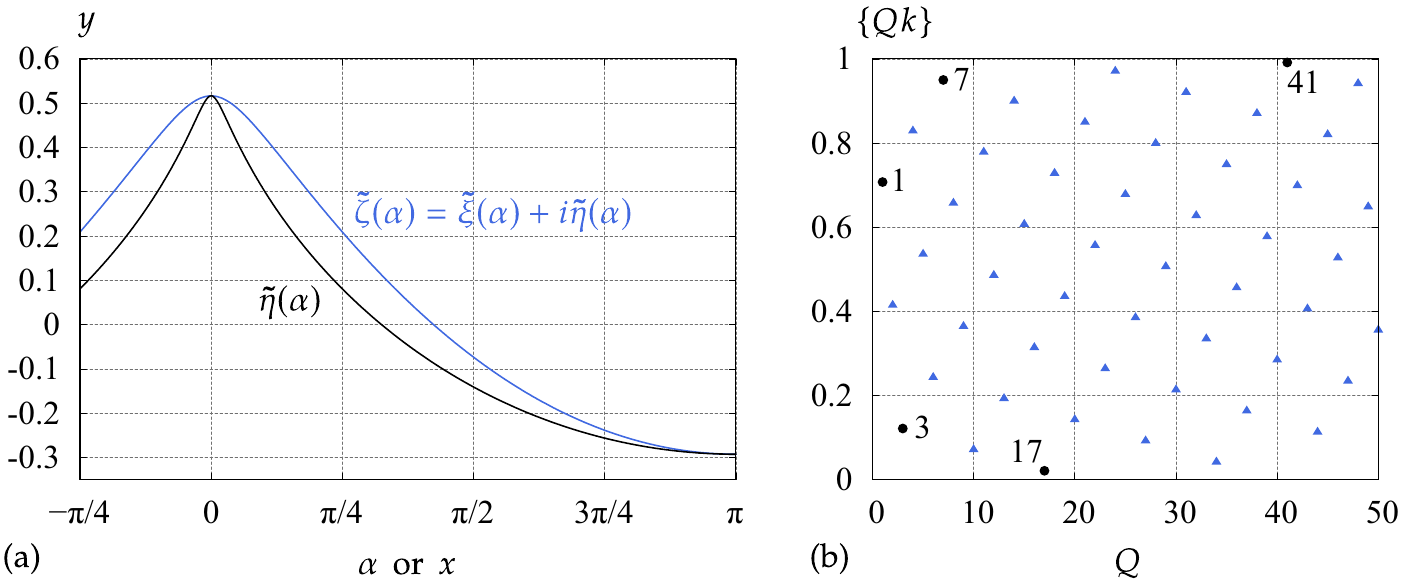}
  \caption{\label{fig:alphaXi:rat} Comparison of $\td\eta(\alpha)$ and
    the parametric plot
    $\td\zeta(\alpha)=\td\xi(\alpha)+i\td\eta(\alpha)$ for solution
    $A$, and fractional parts of the first 50 integer multiples of
    $k=1/\sqrt2$.}
\end{figure}

\begin{remark}\label{rmk:eQ}
  In the present problem, $e_1(Q)$ gives the vertical shift (divided by
    $2\pi$) of the characteristic line passing through $(0,0)$ in the
  $(1,k)$ direction after wrapping around the torus $Q$ times in the
  $\alpha_1$-direction and approximately $P$ times in the $\alpha_2$
  direction. So a small value of $e_1(Q)$ means $C\dot{\td\eta}(\alpha)$
  and $\td\eta_B^\text{dev}(\alpha)$ will nearly recur on shifting
  $\alpha$ by $2\pi Q$. We will focus on the $Q=17$, $P=12$ case.
\end{remark}

Panel (a) of Figure~\ref{fig:quasi:wave} shows $17\frac{17}{36}$
cycles of solution $A$ to the left and right of the origin. In panel
(b), we wrap this solution around a torus of period $17\pi/9$.
This is done by plotting $\td\eta_A(\alpha)$ parametrically versus
\begin{equation}\label{eq:alpha:bar}
  \bar\alpha = \opn{rem}\left( \alpha + \frac{17}{18}\pi\,,\,
    \frac{17}9\pi \right) - \frac{17}{18}\pi,
  \qquad \opn{rem}(a,b) = \Big(a/b-\lfloor a/b\rfloor\Big)b.
\end{equation}
The peak at the origin in panel (a) remains at the origin in panel
(b), but successive peaks of the $2\pi$-periodic Stokes wave are
shifted by $\Delta\alpha=2\pi-\frac{17}9\pi=\frac19\pi$
in panel (b) due to the mismatch of the period of the wave and that of
the plot domain. The labels above the peaks indicate how far one must
advance to the right in panel (a) to obtain the corresponding peak in
panel (b). For example, progressing through 8 periods of the Stokes
wave (to $\alpha=16\pi$) yields the right-most peak in panel (b). The
next peak wraps around the plot domain, so $\alpha=18\pi$ in panel (a)
gives the left-most peak of panel (b). The 17th peak in panel (a) (at
  $\alpha=34\pi$) sweeps out the same curve in panel (b) as the 0th
peak in panel (a). The labels below the peaks in panel (b) work the
same as those above the peaks, but moving left instead of right.

\begin{figure}[p]
  \begin{center}
  \includegraphics[width=.86\textwidth]{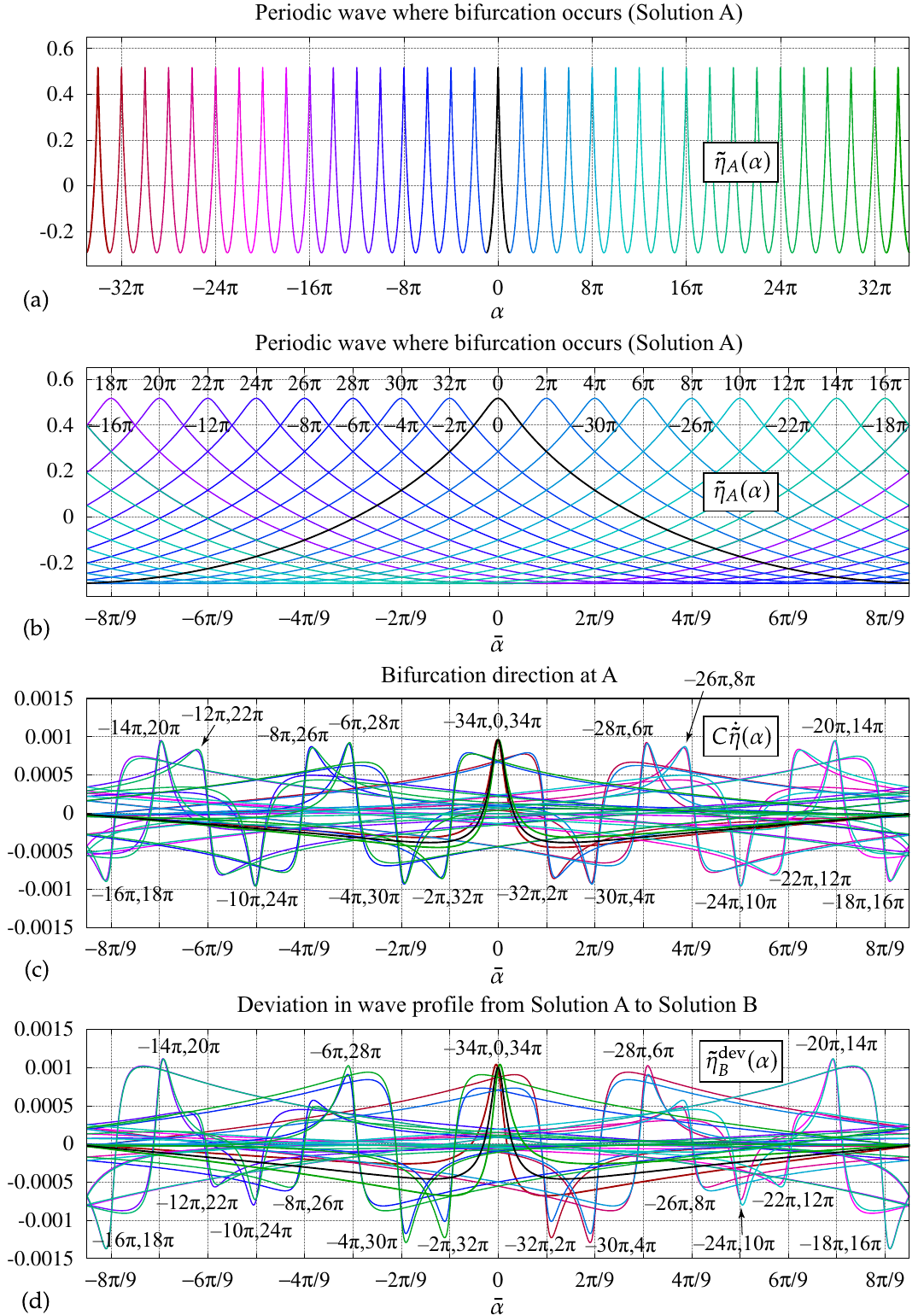}
  \end{center}
  \caption{\label{fig:quasi:wave} Plots of $\td\eta_A(\alpha)$ and the
    perturbations $C\dot{\td\eta}(\alpha)$ and
    $\td\eta_B^\text{dev}(\alpha)$ predicted by linearization about
    solution $A$ and actually occurring by following the
    secondary bifurcation branch from $A$ to $B$. In panels (b)--(d),
    we plot the functions on a domain of period $(17/9)\pi$ to stagger
    both the peaks of solution $A$ and the dominant features of the
    perturbations. The vertical dashed lines in panels (b)--(d) are
    centered on the peaks of the Stokes waves. The perturbations in
    panels (c) and (d) are color coded to match the corresponding
    peaks in panel (a) that have been perturbed.}
\end{figure}

Panels (c) and (d) of Figure~\ref{fig:quasi:wave} show the extracted
functions \eqref{eq:eta:dev:1d} for $|\alpha|\le (2\pi)\big(
  17\frac{17}{36}\big)$, wrapped around the torus of panel (b) via
\eqref{eq:alpha:bar}. We do this to better view the quasi-periodic
behavior of the traveling wave. In the same way that panel (b) shows
more detail than panel (a) about the shape of the peaks of the Stokes
wave, panels (c) and (d) show more detail than would be visible if
they were compressed horizontally to match the style of panel (a).  By
offsetting the peaks of successive cycles of the Stokes wave, the
dominant features of the linear perturbation
$C\dot{\tilde\eta}(\alpha)$ and the nonlinear perturbation
$\tilde\eta_B^\text{dev}(\alpha)$ are similarly offset.  Indeed, we
find that the perturbations change most rapidly near the peaks of the
Stokes wave, and these rapid changes are what we identify as their
dominant features. We label these features in panels (c) and (d) with
the value of $\alpha$ of the nearest peak of the Stokes wave. The
labels come in pairs that differ by $34\pi$. This is because two
points $\alpha$ separated by $34\pi$ are mapped to the same point
$\bar\alpha$ and will cross a peak of the Stokes wave together.  So
their dominant features will occur near each other when plotted versus
$\bar\alpha$. The peak at the origin has 3 labels since the 0th, 17th
and $-17$th peaks of panel (a) are mapped to the 0th peak in panel
(b).  The final $\frac{17}{36}$ of a cycle is included so that these
curves complete their cycles through the plot window rather than
stopping abruptly at $\bar\alpha=0$.

In addition to aligning their dominant features, mapping points
$\alpha$ that differ by 17 cycles to the same point $\bar\alpha$
causes these curves to be close to each other in panels (c) and
(d). By Remark~\ref{rmk:eQ}, advancing $\alpha$ through 17 cycles will
cause the torus function to be evaluated at the same value of
$\alpha_1$ and at a nearby value of $\alpha_2$, shifted up or down by
$2\pi e_1(17)=0.131$. In Figure~\ref{fig:torus:big}, the
  dashed green and red lines correspond to $\alpha\ge0$ and
  $\alpha\le0$, respectively.  Over the range $0\le\alpha\le 34\pi$,
each dashed green line is offset vertically by $0.131$ from a nearby
dashed red line. This vertical offset is equivalent to advancing
$\alpha$ by $34\pi$, i.e., by displacing $(\alpha_1,\alpha_2)$ by
$(34\pi,34\pi k)$ and mapping back to $\mbb T^2$ by periodicity,
starting at a point on the red line. The final fractional cycles of
length $(17/36)(2\pi)$ in each direction are offset vertically
  by $0.131$ from lines of the same color and terminate at the
circular green and red markers in Figure~\ref{fig:torus:big}. As noted
above, these fractional cycles are included to extend the plots in
panels (c) and (d) of Figure~\ref{fig:quasi:wave} to the end
of the plot window so they don't end abruptly at $\bar\alpha=0$.

Studying panels (c) and (d), we observe that the perturbation at the
origin, plotted in black, sharpens the Stokes wave symmetrically,
where we view following the secondary bifurcation branch as a
perturbation of solution $A$.  The other wave crests are perturbed
asymmetrically and can be sharpened, flattened or shifted right or
left, with no two perturbed in exactly the same way. Even though we
could not follow this branch very far at the scale of the bifurcation
diagram shown in Figure~\ref{fig:grav:main} using wave height and wave
speed as parameters, solution $B$ has many small scale features not
present in solution $A$. Moreover, nonlinear effects cause the
deviation of $B$ from $A$ in panel (d) to differ visibly from that predicted
by linearization, shown in panel (c).  We also see that while the
perturbations do stay reasonably close to each other when $\alpha$
increases by 17 cycles, differences are clearly visible and there
would not be much agreement after another 17 or 34 cycles.  Closer
agreement could be achieved by switching to $Q=41$ as $2\pi
e_1(41)=0.054$, but this would increase the number of peaks in panel (b)
from 17 to 41, making it more difficult to distinguish the features
that arise in panels (c) and (d).

From these results, it is natural to conjecture that this
path of quasi-periodic solutions will continue until the wave profile
develops a singularity, presumably with a $120^\circ$ corner in
physical space at the origin \cite{stokes1880theory,lhf:78}. In the
periodic case, this limiting corner wave has been proved to exist by
Amick, Fraenkel and Toland \cite{amick:82} and studied numerically by
Gandzha and Lukomsky \cite{gandzha:07}. Chen and Saffman
\cite{chen1980numerical} found that wavelength-doubling and wavelength-tripling
bifurcations also lead to families of solutions that appear (in the
  numerical simulations) to terminate with the tallest crest
sharpening to $120^\circ$ while the other crests remain rounded. In
the case of genuinely quasi-periodic traveling waves studied here, the
analogous result would be for the torus function representing the
traveling solution to develop a singularity at $(0,0)$ when the wave
height reaches a critical value.
  
To investigate this limit, it is preferable to transform our torus
functions from a conformal mapping formulation to a graph-based
formulation. Recall from \eqref{eq:govern} that $\xi$, which
represents the quasi-periodic part of the horizontal position of the
wave in the sense of \eqref{eq:xi:def}, is related to the wave profile
$\eta$ via $\xi=H[\eta]$. Here $\xi$ and $\eta$ are torus functions.
In \cite{quasi:ivp}, it is shown that if
$\pa_\alpha\big\vert_{\alpha=0}\xi(\alpha_1+\alpha,\alpha_2+k\alpha)>-1$
for $(\alpha_1,\alpha_2)\in\mbb T^2$, then the extracted wave profile
\begin{equation}
  \td\zeta(\alpha)= \overbrace{\alpha+\xi(\alpha,k\alpha)}^{\td\xi(\alpha)} + i
  \overbrace{\eta(\alpha,k\alpha)}^{\td\eta(\alpha)}
\end{equation}
is a graph and there is a torus function $\eta^\phys(x_1,x_2)$ such
that
\begin{equation}\label{eq:td:eta:phys}
  \td\eta(\alpha)=\td\eta^\phys(\td\xi(\alpha)), \qquad
  \td\eta^\phys(x) = \eta^\phys(x,k x), \qquad (\alpha,x\in\mbb R).
\end{equation}
This torus function can be computed via
\begin{equation}\label{eq:eta:phys:from:eta}
  \eta^\phys(\vec x) = \eta\big(\vec x + \vec k \mc A(\vec x)\big), \qquad
  \big(\vec x\in\mbb T^2\big),
\end{equation}
where $\vec k=(1,k)$ and $\mc{A}(\vec x)$ is the unique solution
\cite{quasi:ivp} of
\begin{equation}\label{eq:mcA:def}
  \mc A(\vec x) + \xi\big( \vec x + \vec k \mc A(\vec x) \big) = 0, \qquad
  \big(\vec x\in\mbb T^2\big).
\end{equation}
Note that the wave number ratio, $k$, which is set to $1/\sqrt2$ in
the examples presented here, is the same in physical space as in
conformal space. It is shown in \cite{quasi:ivp} that the inverse of
the mapping $\vec x = \vec\alpha + \vec k\xi(\vec \alpha)$ on $\mbb
T^2$ is $\vec\alpha=\vec x+\vec k\mc A(\vec x)$.

\begin{figure}
  \includegraphics[width=\textwidth]{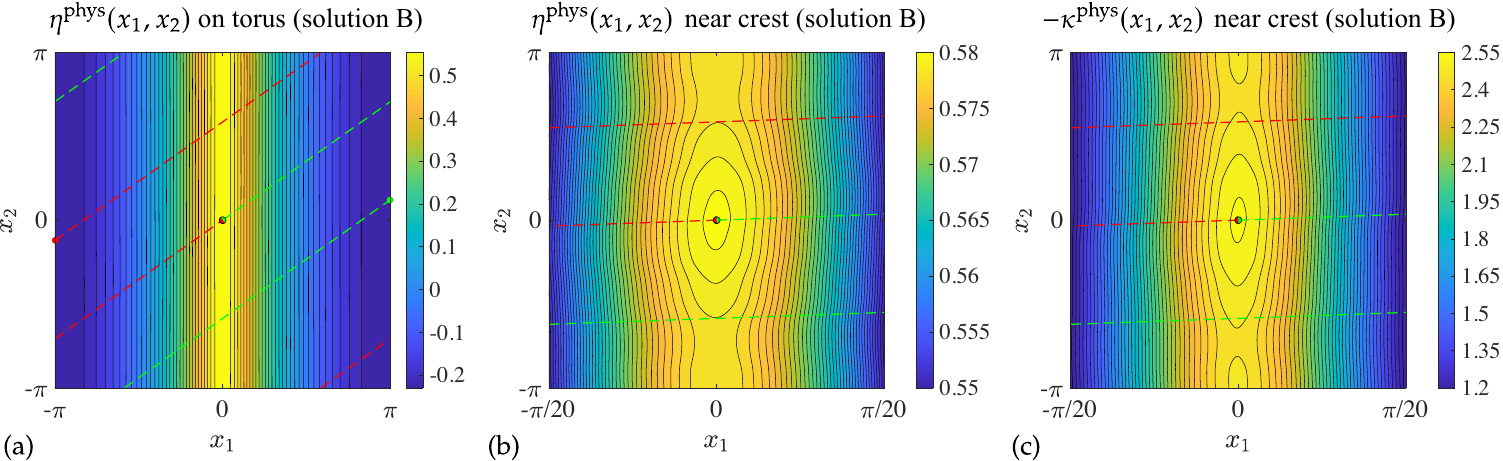}
  \caption{\label{fig:torus:kappa} Contour plots of the torus
    functions of the wave profile and the negative of the curvature
    for solution $B$ after transforming to a graph-based
    representation of the solution. The dashed lines in panels (b) and
    (c) are the same lines shown in panel (a), but we have zoomed in
    on the wave crest by restricting to $-\frac\pi{20}\le
    x_1\le\frac\pi{20}$.}
\end{figure}

Panel (a) of Figure~\ref{fig:torus:kappa} shows a contour plot of
$\eta^\phys(x_1,x_2)$ for solution $B$ of
Figures~\ref{fig:grav:bif:branch}, \ref{fig:fourier:big},
\ref{fig:torus:big} and \ref{fig:quasi:wave}. After computing the
torus function $\eta(\alpha_1,\alpha_2)$ for solution $B$ in conformal
space, we used Brent's method \cite{brent:73} to rapidly solve
\eqref{eq:mcA:def} for $\mc A(\vec x)$ on a $257\times257$ uniform
grid on $\mbb T^2$ with period cell $[-\pi,\pi]^2$ and then evaluated
$\eta^\phys(\vec x)$ via \eqref{eq:eta:phys:from:eta}. The dashed
green and red lines in panel (a) show where $\eta^\phys(x_1,x_2)$
would be evaluated to extract $\td\eta^\phys(x)$ in
\eqref{eq:td:eta:phys} over $-3\pi\le\alpha\le 3\pi$.  Because the
deviation of solution $B$ from the periodic wave $A$ is small, as seen
in Figures~\ref{fig:torus:big} and~\ref{fig:quasi:wave}, it is
difficult to see the variation of $\eta^\phys(x_1,x_2)$ with respect
to $x_2$ when plotted over the entire torus. In panel (b) of
Figure~\ref{fig:torus:kappa}, we repeat the calculation on a subset of
the torus, with $x_1\in [-\pi/20,\pi/20]$ and $x_2\in [-\pi,\pi]$.
The increased resolution achieved by zooming in on this region reveals
that the torus function $\eta^\phys(x_1,x_2)$ has a maximum at
$(0,0)$. This means the extracted wave $\td\eta^\phys(x)$ is largest
at $x=0$, where the characteristic line $(x,kx)$ passes through
$(0,0)$. In panel (c), we plot the (negative of the) curvature,
$\kappa^\phys(\vec x)=\kappa\big(\vec x + \vec k\mc A(\vec x)\big)$,
where the formula for $\kappa(\alpha_1,\alpha_2)$ is given in
\eqref{eq:govern}. These formulas imply that
$\td\kappa(x) = \kappa^\phys(x,kx)=
(\td\eta^\phys)''(x)/\big[1+\big((\td\eta^\phys)'(x)\big)^2\big]^{3/2}$.
We see that $-\kappa^\phys(x_1,x_2)$ has a maximum at $(0,0)$,
confirming that the highest peak of $\td\eta^\phys(x)$, which occurs
at $x=0$, coincides with the sharpest peak, where the curvature
is most negative.

These results are consistent with the conjecture that the maximum of
$\eta^\phys(\vec x)$ at $\vec x=(0,0)$ will continue to grow and
sharpen to form a singularity at the origin in such a way that the
extracted wave $\td\eta^\phys(x)$ forms a $120^\circ$ corner at
$x=0$. All the other peaks would remain rounded in this limit (assuming the torus
  function remains smooth except at the origin, where it is continuous
  but has a discontinuous gradient), though there
would be peaks of arbitrarily high curvature as the characteristic
line $(x,kx)$ will pass arbitrarily closely to $(0,0)$ modulo
$2\pi\mbb Z^2$ as $x\to\pm\infty$. This conjecture is highly
speculative as there is still a long way to go from solution $B$ to a
solution with a sharp corner. The curvature $\td\kappa^\phys(0)$ has
only decreased from $-2.513$ for solution $A$ to $-2.554$ for solution
$B$ and would have to approach $-\infty$ in order to form a corner
wave.  Bifurcation from solution $E$ is perhaps more promising for
reaching a genuinely quasi-periodic wave in which the crest at the
origin has a sharp corner since solution $E$ itself is closer to the
limiting $120^\circ$ periodic wave than solution~$A$. The curvature at
the crest of solution $E$ is $-52.015$.

Exploring this conjecture further numerically is currently out of
reach due to the computational cost of tracking quasi-periodic solutions on the
branch from solution $A$ past solution $B$ or computing any solution
on the path bifurcating from solution $E$. This is partly because the
conformal mapping approach is not well suited to representing nearly
singular wave profiles. The grid spacing tends to spread out precisely
where one needs mesh refinement. This is evident already for solution
$A$ in Figure~\ref{fig:alphaXi:rat} by comparing the plot of
$\td\eta(\alpha)$ versus $\alpha$, where the grid is uniformly spaced,
to that of $\td\eta(\alpha)$ versus $\td\xi(\alpha)$, which describes
the curve in physical space. The effect is much worse for solution
$E$, and in the limit that $\td\zeta(\alpha)$ forms a corner,
$\td\eta(\alpha)$ will form a cusp. Boundary integral methods
\cite{wilkening2012overdetermined} and auxiliary conformal maps
\cite{tanaka:83,lushnikov:17} are more flexible for
controlling the grid spacing but have not yet been adapted to the
quasi-periodic setting.

For temporally periodic standing water waves, Penney and
Price \cite{penney:52} conjectured that the largest-amplitude standing
wave will form a $90^\circ$ corner each time it comes to rest.  Taylor
performed wave tank experiments corroborating this conjecture but
doubted Penney and Price's analysis \cite{taylor:53}. Careful
numerical studies suggested that the limiting wave may have a corner
as sharp as $60^\circ$ \cite{mercer:92}, or even a cusp
\cite{schultz}.  Wilkening \cite{water1} increased the resolution near
the crest tip by a factor of 200 over these previous studies and
showed that the Penney and Price conjecture is false due to a
breakdown of self-similarity. Increasing the amplitude leads to
increasingly complex behavior at small scales that prevents the
emergence of a limiting standing wave
\cite{wilkening2012overdetermined}. It is an interesting open question
whether the $120^\circ$ corner wave conjecture will turn out to be
true for spatially quasi-periodic traveling waves, or whether the
torus functions will become rough on small scales, diverging in some
Sobolev norm before the crest can sharpen to a corner at the origin.

\subsection{Overturning quasi-periodic gravity-capillary waves}
\label{sec:num:cap}

In this section, we compute the two-parameter family of ``type 1''
gravity-capillary waves studied numerically by Schwartz and
Vanden-Broeck \cite{schwartz:79} and more recently by Akers, Ambrose
and Wright \cite{akers2014gravity} and search for bifurcations to
quasi-periodic traveling waves corresponding to $k = 1/\sqrt{2}$. In
particular, we obtain overturning quasi-periodic waves that bifurcate
from periodic overturning waves.

At large amplitude, the type 1 waves possess a symmetric air pocket at
$x=\pi$ that drops down into the fluid, possibly surrounded by
overhanging regions of the free surface. By contrast, type 2
  waves have two air pockets that drop down into the fluid on either
  side of $x=\pi$; see Figure~\ref{fig:type1and2}. In our
dimensionless units with the wavelength normalized to $2\pi$ and the
acceleration of gravity normalized to $g=1$, our surface tension
parameter $\tau$ agrees with the parameter $\kappa$ used by Wilton
\cite{wilton1915} and by Schwartz and Vanden-Broeck
\cite{schwartz:79}. The waves we seek, i.e., the primary branch of
type 1 periodic traveling waves, bifurcate from zero amplitude for
$\tau>1/2$. The case $\tau=1/2$ corresponds to a Wilton ripple
\cite{wilton1915,vandenBroeck:book}, where there are two solutions of
the Stokes expansion, one of which matches up with this family of type
1 waves. This family can be numerically continued to smaller values of
$\tau$, which we do, but the result is different than bifurcation from
zero amplitude at these smaller values of $\tau$. Indeed, bifurcating
from zero amplitude with $\tau\in(1/3,1/2)$ leads to ``type 2'' waves
\cite{schwartz:79} that match up with the other solution of the Wilton
ripple expansion at $\tau=1/2$. Figure~\ref{fig:type1and2} shows
  a type 1 wave and a type 2 wave with $\tau=1/2$, with amplitudes
  chosen just below the point that self-intersection occurs. We use
  $\hat\eta_{1,0}=\hat\eta_1$ for the amplitude parameter.

\begin{figure}
\begin{center}
\includegraphics[width=.75\textwidth]{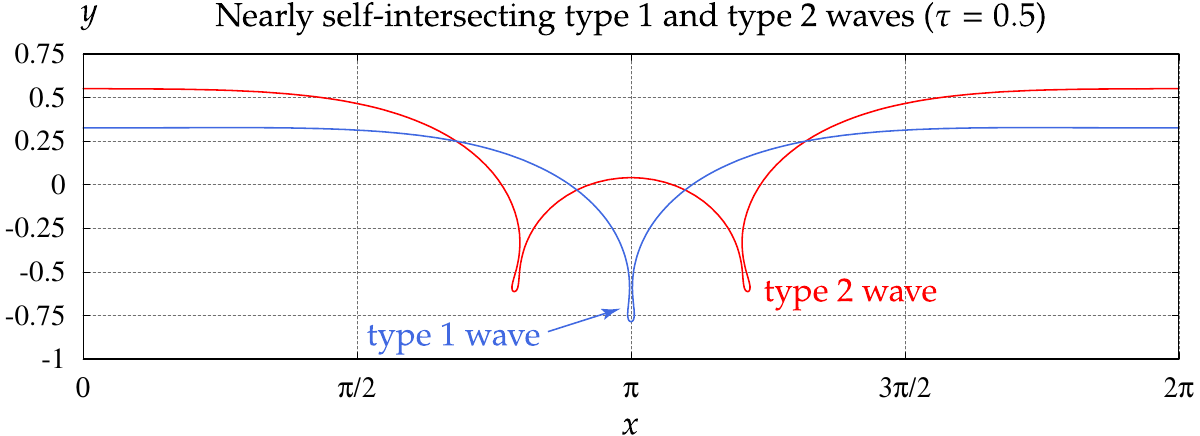}
\end{center}
\caption{\label{fig:type1and2} Nearly self-intersecting type 1 and
  type 2 periodic traveling gravity-capillary waves with surface tension
  parameter $\tau=1/2$. The amplitude parameter is $\hat\eta_1=0.26417$
  for the type 1 wave shown, and $\hat\eta_1=0.28947$ for the type 2 wave
  shown.}
\end{figure}

Panel (a) of Figure~\ref{fig:bifur_overturn} shows a contour plot of
$\chi$ evaluated on the family of type 1 waves we computed. We split
the domain of the contour plot into 3 regions
\begin{equation}\label{eq:3:regions}
  \text{region 1: } 1\le\tau\le 6.5, \quad
  \text{region 2: } 0.52\le\tau\le 0.98, \quad
  \text{region 3: } 0.05\le\tau\le 0.5.
\end{equation}
In each region, we sample $\tau$ at equal intervals of size
$\Delta\tau= 0.05$, $0.02$ and $0.01$, respectively. At each $\tau$
value of regions 1 and 2, we use $s=\hat\eta_{1,0}$ as an amplitude
parameter and sweep forward with steps of size $\Delta s=0.01$ until
the wave self-intersects to form a trapped bubble.  The conformal
mapping method can compute non-physical waves in which the free
surface crosses through itself to form an overlapping fluid region. We
use this feature to root-bracket the amplitude at which the walls of
the air bubble first meet. In more detail, once the amplitude is large
enough that the wave contains an air pocket with overhanging walls, we
compute the first zero, $\alpha_0(s,\tau)$, of $\td\xi'(\alpha)$ using
Newton's method. We then evaluate $\td\xi(\alpha_0(s,\tau))-\pi$ as we
continue to increase $s$ by $\Delta s=0.01$. Once this function is
positive, the wave has self-intersected and we have found a bracket to
use in Brent's method to find $s(\tau)$ such that
$\td\xi(\alpha_0(s(\tau),\tau))=\pi$ to double-precision accuracy.  We
then compute $\chi(s,\tau)$ at 81 values of $s$, uniformly spaced
between 0 and $s(\tau)$.  As a result, the right boundary of the
contour plot corresponds to the maximum amplitude for each $\tau$
where the air pocket closes to form a bubble.

In the third region of \eqref{eq:3:regions}, an additional step is
taken in which two numerical continuation paths are computed with
$s=\hat\eta_{1,0}$ held fixed and $\tau$ decreasing. The specific
choices of $s$ are $0.005$ and $0.006$, with starting points at
$\tau\in\{0.52,0.54\}$, computed as part of region 2 in
\eqref{eq:3:regions}. Once solutions are known with $(\tau,s)$ in the
range $0.05\le\tau\le0.5$ and $s\in\{0.005,0.006\}$, we proceed as
above to find the boundary to the right, but with $\Delta s$ decreased
to $0.001$ for the search for the initial bracket for Brent's
method. When the right boundary is found (where the air pocket pinches
  off into an air bubble), we compute $\chi$ at 81 equally spaced
points between $s=0.0017$ and the pinch-off amplitude. We re-iterate
that only for $\tau>1/2$ does the wave approach zero-amplitude when
the parameter $s$ decreases to zero. We stop at $s=0.0017$ for
$\tau\le1/2$ since a different bifurcation parameter than
$\hat\eta_{1,0}$ is needed to properly explore the limit as $s\to0^+$
when $\tau\le1/2$ on this sheet of type 1 solutions, and $s=0.0017$ is
small enough for the purpose of plotting $\chi(s,\tau)$.

\begin{figure}
\begin{center}
\includegraphics[width=.9\textwidth]{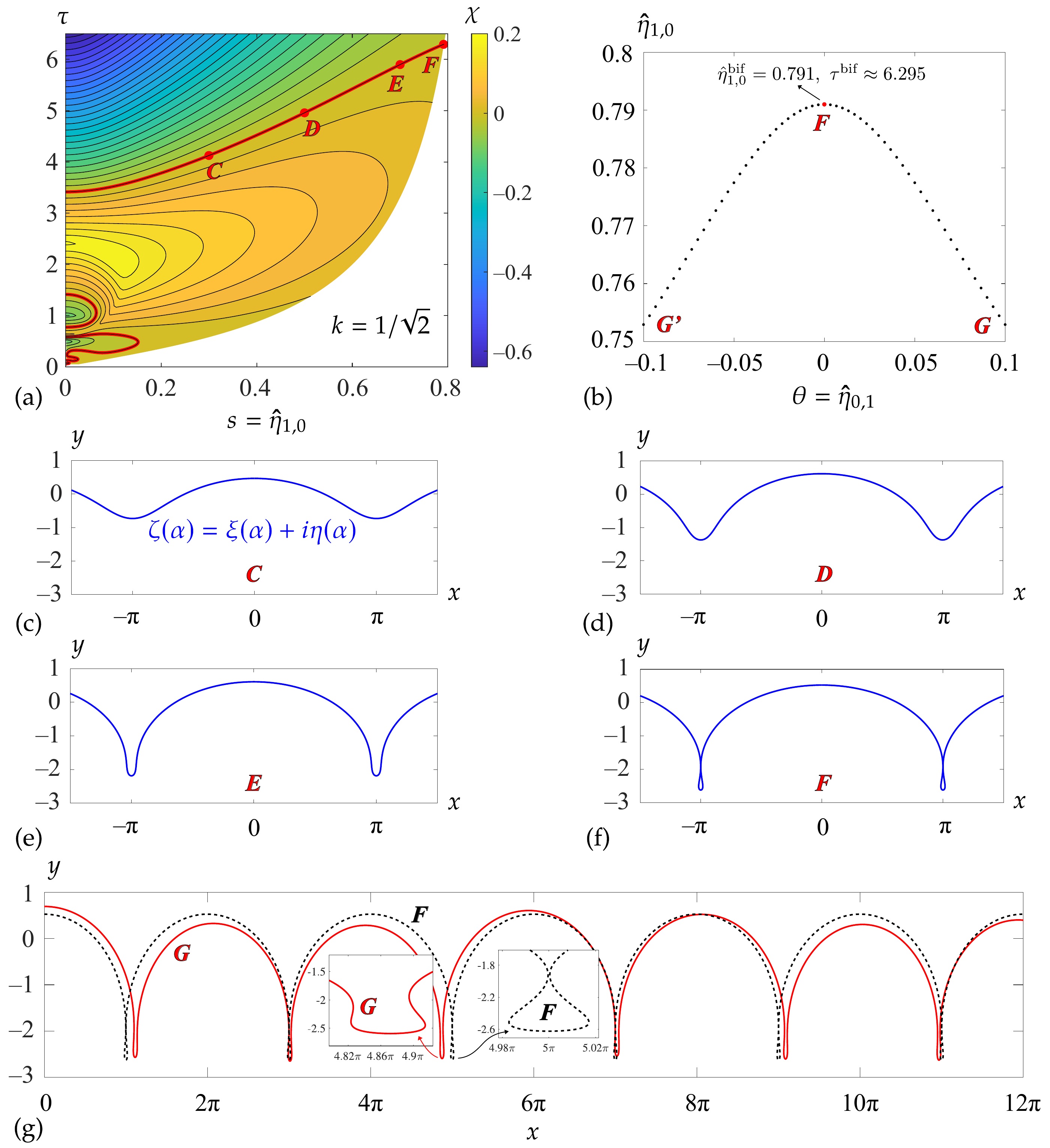}
\end{center}
\caption{\label{fig:bifur_overturn} Contour plot of $\chi(s,\tau)$ for
  type 1 gravity-capillary waves along with representative examples of
  bifurcation points on the zero contour of $\chi$ and a secondary
  branch of quasi-periodic, overhanging traveling waves bifurcating
  from solution $F$.}
\end{figure}

Because it is a two-parameter family, many solutions had to be
computed to generate the contour plot of $\chi(s,\tau)$ in panel (a)
of Figure~\ref{fig:bifur_overturn}. There are 14661 solutions
represented in the plot. The Fourier cutoffs $N_1$ and $N$ used to
compute the traveling waves and $\chi$ are as follows:
\begin{equation}
  \begin{array}{c|cccccc}
    \text{$\tau$ range} &
    [0.05,0.09] & [0.1,0.19] & [0.2,0.5] &
    [0.52,0.98]    & [1.0,2.45]  & [2.5,6.5] \\ \hline
    N_1 & 320 & 256 & 128 & 108 & 72 & 54\\
    N   & 384 & 300 & 160 & 128 & 90 & 64
  \end{array}
\end{equation}
We set $M_1=3N_1$ and $M=3N$ in all cases. We highlight the zero
contours of $\chi(s,\tau)$ in red. There are 4 values of $\tau>0.5$ at
which a red contour line reaches the $\tau$ axis, i.e., where
$\chi(0,\tau)=0$. In the small-amplitude limit, the candidate
quasi-periodic bifurcations are plane waves, enumerated by Fourier
lattice points $\vec l\in\pm\Lambda$ from \eqref{eq:varphi:def}. Plane
waves with wave numbers $k_1=1$ and $k_2=l_1+kl_2$ can co-exist as
traveling waves if they are both roots of the dispersion relation
\begin{equation}
  c^2=\frac g{k_j} + \tau k_j, \qquad (j=1,2).
\end{equation}
Eliminating $c^2$ and setting $g=1$, we have $1+\tau=\frac1{k_2}+\tau
k_2$, or $\tau=1/k_2$. In our calculation, we use Bloch theory to
restrict attention to the case $l_2\in\{\pm1\}$ when constructing $\mc
J^\qua$ in \eqref{eq:J:qua}. The possible values of $\tau$ that exceed
1/2, the Wilton ripple cutoff for type 1 waves, are then
\begin{equation}\label{eq:tau:crit}
  \tau = \frac{1}{1-k}=3.414, \quad
  \tau = \frac{1}{0+k}=1.414, \quad
  \tau = \frac{1}{2-k}=0.773, \quad
  \tau = \frac{1}{1+k}=0.586,
\end{equation}
where $k=1/\sqrt2$.  These are precisely the locations where the red
contour lines in panel (a) meet the $\tau$ axis. We find that the
contours at $1.414$ and $0.773$ form a closed loop, and the
contour at 0.586 drops down below $\tau=0.5$ into region 3 from
\eqref{eq:3:regions}, where we used numerical continuation to extend
the range of type 1 waves below the Wilton ripple cutoff.

We now focus on the remaining contour, which begins at $\tau=3.414$.
We see in Figure~\ref{fig:bifur_overturn} that this contour extends
all the way to the right boundary of the contour plot, where the type
1 waves self-intersect at a point, trapping an air bubble below. The
solutions labeled $C$, $D$, $E$ and $F$ have values of $s$ equal to
$0.3$, $0.5$, $0.7$ and $0.791$, respectively, and are plotted in
panels (c)--(f). As $s$ increases along the contour, the air pocket
at $x=\pi$ (or one of its $2\pi$-periodic translations) deepens, eventually
forming a trapped bubble. Solution $E$ has formed a deep air pocket,
but the free surface is still a graph. Solution $F$ has overhanging
walls that nearly touch. Using polynomial interpolation from the
points on the right boundary with $6.1\le\tau\le 6.5$, we find
that the red contour meets the right boundary at
\begin{equation}
  (s,\tau) = (\,0.7916457855\;,\;6.297256422\, ).
\end{equation}
Solution $F$ is very close to this, with
$(s,\tau)=(0.791,6.294747714)$. Solutions $C$, $D$ and $E$
correspond to $(s,\tau)=(0.3,4.1248233)$, $(0.5,4.9579724)$ and
$(0.7,5.8995675)$, respectively. Each of the waves $C$--$F$ bifurcates
to a quasi-periodic family with basic frequencies $k_1=1$ and
$k_2=k=1/\sqrt2$. Indeed, the entire red curve corresponding to
$\chi=0$ gives the intersection of the two-parameter family of
periodic type 1 waves and a two-parameter sheet of such QP
waves. Fixing $\tau$ or $s$ reduces the problem to a standard
one-parameter bifurcation problem of the type studied in
Section~\ref{sec:num:grav} above.  See \cite{antman:81,antman:book}
for background on the general theory of multi-parameter bifurcation
theory.

To explore the existence of spatially quasi-periodic, overhanging
traveling waves, we follow the bifurcation branch from solution $F$
using our numerical continuation algorithm. We choose $N_1 = N_2 =
64$, $M_1 = M_2 = 150$ and use $q^\bif \pm 10^{-5} \big(0,0, \dot
  {\eta}^\qua\big)$ to jump from the periodic branch to the
quasi-periodic branch. The largest Fourier coefficient of $\dot
{\eta}^\qua$ is $\dot{\hat{\eta}}^\qua_{0,1} \approx 0.6901$.  Thus,
$\theta=\hat{\eta}_{0,1}$ is a natural choice for the continuation
parameter on the QP branch. We hold $\tau$ fixed at $\tau^\bif$ in
this search.  We are able to compute the quasi-periodic continuation
path until $\theta$ reaches $\pm 0.1$.  The corresponding solution
with positive $\theta$, labeled $G$ in panel (b), is plotted in panel
(g) of Figure \ref{fig:bifur_overturn}. For this solution,
$\hat{\eta}_{1,0} \approx 0.7529$ and the objective function $f$ is
minimized to $3.7\times 10^{-25}$.  In panel~(b) of Figure
\ref{fig:bifur_overturn}, we plot the Fourier coefficients
$\hat{\eta}_{1,0}$ and $\hat{\eta}_{0,1}$ of the bifurcated
quasi-periodic solutions.  Along the quasi-periodic branch, as
$|\theta|$ increases, $|\hat{\eta}_{1,0}|$ decreases.  We also observe
that the plot is symmetric with respect to the vertical line
$\hat{\eta}_{0,1} = 0$; this is because the quasi-periodic solutions
with negative $\theta$ can be obtained from ones with positive
$\theta$ through a spatial shift in $\alpha_2$: $(\alpha_1, \alpha_2)
\mapsto (\alpha_1, \alpha_2 + \pi)$. Solutions $G$ and $G'$ are
related in this way.

In panel (g) of Figure \ref{fig:bifur_overturn}, we compare the wave
profile $\zeta=\xi+i\eta$ of solution $F$, which is periodic with
amplitude $s=0.791$, and solution $G$, which is quasi-periodic with
$\theta=0.1$.  We observe that the peaks and troughs of the QP
solution appear in a non-periodic pattern.  The peaks of solution $G$
are above those of solution $F$ near $\xi = 0, 6\pi$ and below near
$\xi = 2\pi, 4\pi, 10\pi$; the troughs of $G$ are on the left of those
of $F$ near $\xi = 5\pi, 11\pi$ and on the right near $\xi = 3\pi,
7\pi, 9\pi$. Beyond the plot window shown, the deviation of solution
$G$ from $F$ will continue to differ from one peak and trough to
the next, never exactly repeating over the real line.  We zoom in on
the troughs of the two solutions near $\xi = 5\pi$ and observe that
the trough of solution $G$ is asymmetrical and wider than that of the
periodic solution $F$.  Neither solution is self-intersecting.
Moreover, solution $G$ is further from self-intersecting than solution
$F$. This may be related to the result in panel (b) that increasing
$|\theta|$ causes $\hat\eta_{1,0}$ to decrease, and decreasing
$\hat\eta_{1,0}$ on the periodic branch increases the gap between the
overhanging walls. Specifically, solution $G$ has
$\hat{\eta}_{1,0}\approx0.7529$ while solution $F$ has
$\hat{\eta}_{1,0}\approx0.7916$. Nevertheless, this quasi-periodic
solution does exhibit overhanging regions, as shown in the inset plot.

\section{Conclusion}
\label{sec:conclusion}

We have shown that a signed version, $\chi$, of the smallest singular
value, $\sigma_\text{min}$, of the Jacobian serves as an excellent
test function to locate branch points in equilibrium problems. While
$\sigma_\text{min}$ has a slope discontinuity at each of its zeros,
the existence of an analytic or smooth SVD ensures that the function
becomes smooth when a sign is included on one side of the bifurcation
point. We show that this sign factor can be defined as the product of
the determinants of the orthogonal matrices containing the left and
right singular vectors in the standard SVD, which alleviates the need
to actually compute an analytic SVD. We also show how to compute this
sign efficiently from the bidiagonal matrix obtained in the first
phase of the standard SVD algorithm.  It is not necessary to form the
matrices of singular vectors, compute their determinants, or compute
the determinant of the Jacobian itself.

One benefit of using $\chi$ as a test function is that root bracketing
algorithms such as Brent's method can then be used to locate
bifurcation points. This is simpler than the Newton-type method
proposed by Shen \cite{shen:bif:97} to locate zeros of
$\sigma_\text{min}$ or by various authors \cite{griewank:84,
  allgower:97, beyn01, dhooge:03, bindel14} to solve minimally
extended systems. Within the constraints and philosophy of
Remark~\ref{rmk:dyn:sys}, our method is as efficient as these
alternative approaches. We also proposed a polynomial interpolation
approach using Chebyshev polynomials, which relies on the smoothness
of $\chi$ to achieve spectral accuracy.  In multi-parameter
bifurcation problems such as the gravity-capillary problem of
Section~\ref{sec:num:cap}, the zero level set of $\chi$ can be used to
visualize and compute the intersection of the primary family of
solutions with the secondary family of solutions.  This would not work
well using $\sigma_\text{min}$ or $\psi$ from \eqref{eq:J:ee} instead
of $\chi$.  We also showed that $\chi$ stops changing on further mesh
refinement once the singular vectors corresponding to the smallest
singular values are resolved.  As a result, the plots in
Figure~\ref{fig:bif:svd} and the contour plot of
Figure~\ref{fig:bifur_overturn} are smooth even though many different
mesh sizes are used in the underlying calculations plotted.  Other
test functions such as the determinant of the Jacobian or $\psi(s)$ in
\eqref{eq:J:ee} do not have this mesh independence property, and some
behave poorly on large-scale discretizations of infinite dimensional
problems.

We use this method to compute, for the first time, quasi-periodic
traveling gravity waves with zero surface tension and overhanging
traveling gravity-capillary waves. The former example yields traveling
waves that still make sense at the scale of the ocean, where the
length scale of capillary waves is so much smaller than that of
gravity waves (by a factor of $10^{-7}$) that one can set
$\tau=0$. Genuinely quasi-periodic pure gravity waves do not persist to zero
amplitude, which motivated us to search for quasi-periodic
bifurcations from large-amplitude periodic waves.  The latter example
showcases the use of $\chi(s,\tau)$ to study two-parameter bifurcation
problems in which the primary sheet is parameterized over a region
with one side bounded by singular solutions. In our case, the right
boundary of the contour plot in Figure~\ref{fig:bifur_overturn}
corresponds to type 1 waves that self-intersect to form an air
pocket. This contour plot makes it easy to visualize how the secondary
two-parameter family of quasi-periodic solutions fits together with
the primary two-parameter family of periodic traveling waves.

Once bifurcations are found, we use numerical continuation to explore
the secondary branches of quasi-periodic solutions. This becomes
computationally expensive, especially in the case of the pure gravity
wave problem. We formulated the problem of finding solutions on the
secondary branch as an overdetermined nonlinear least squares problem
and implemented a parallel algorithm employing MPI and ScaLapack to
carry out the trust-region minimization steps of the
Levenberg-Marquardt method. The largest-amplitude solution we
computed, namely solution $B$ of Figures~\ref{fig:grav:bif:branch}
and~\ref{fig:fourier:big}, required solving for $N_\text{tot}=53398$
independent Fourier modes in its 2D torus representation. We used the
Savio cluster at UC Berkeley and the Lawrencium cluster at Lawrence
Berkeley National Laboratory for these calculations.

An interesting feature of these solutions is that the 2D Fourier modes
$\hat\eta_{j_1,j_2}$ continue to exhibit visible effects of resonance
near the line $j_1+kj_2=0$ even though these are large-amplitude
solutions far from linear water wave theory.  We explain this by
noting that these modes correspond to long wavelengths when the torus
function is restricted to the characteristic line $\alpha_1=\alpha$,
$\alpha_2=k\alpha$, and the Euler equations are not strongly affected
by long wavelength perturbations. We also identified a symmetry
connecting solutions on one side of each secondary bifurcation branch
to the other. In particular, solution $B'$ in
Figure~\ref{fig:grav:bif:branch} is related to solution $B$ via
$\eta_{B'}(\alpha_1,\alpha_2) = \eta_{B}(\alpha_1,\alpha_2+\pi)$.  The
same symmetry was found in the gravity-capillary problem, e.g.,
solutions $G$ and $G'$ in Figure~\ref{fig:bifur_overturn} are also
related by this transformation. Figures~\ref{fig:quasi:wave}
and~\ref{fig:bifur_overturn} also show that when a large-amplitude
periodic traveling water wave is perturbed to create a quasi-periodic
traveling wave, each crest and trough of the infinite wave train will
undergo a different perturbation, so that no two are exactly the
same. Nevertheless, they fit together to form a single traveling wave
profile extending over the real line.

We took advantage of Bloch theory to express the Jacobian as a direct
sum of operators mapping $(\mc X_\sigma^\per,\mbb R,\mbb R)$ to
$X_\sigma^\per$ and $(\mc X_\sigma^\e{l_2},0,0)$ to $\mc
X_\sigma^\e{l_2}$. This greatly reduces the number of rows and columns
of $\mc J^\qua$ in \eqref{eq:J:qua} since $l_2$ could be set to 1
rather than varying over some range $-N_2\le l_2\le N_2$. Bloch theory
is also useful for studying the stability of traveling waves to
subharmonic perturbations \cite{longuet:78, mclean:82, mackay:86}, and
indeed the present work of searching for null vectors of the Jacobian
of the traveling wave equations can be thought of as a special case of
looking for zero eigenvalues of the dynamic stability problem. In the
present paper, we have shown that perturbations in null directions of
the linearization lead to branches of spatially quasi-periodic
traveling waves for the full water wave equations. In future work, it
would be interesting to investigate the Benjamin-Feir instability
\cite{benj:feir:67, zakharov1968stability} and other unstable subharmonic
perturbations \cite{oliveras:11, tiron2012linear, trichtchenko:16,
  murashige:20} by evolving them beyond the realm of validity of Bloch
stability theory using the dynamic version \cite{quasi:ivp} of our
spatially quasi-periodic water wave framework. Linearly stable
subharmonic perturbations would also be interesting to investigate as
they may lead to solutions of the full water wave equations that are
quasi-periodic in time
\cite{berti2016quasi,baldi2018time,berti2020traveling,feola2020trav,
  waterTS} as well as space.

\textbf{Acknowledgments.}  This work was supported in part by the
National Science Foundation under award number DMS-1716560 and by the
Department of Energy, Office of Science, Applied Scientific Computing
Research, under award number DE-AC02-05CH11231. This research used the
Lawrencium computational cluster resource provided by the IT Division
at the Lawrence Berkeley National Laboratory and the Savio
computational cluster provided by the Berkeley Research Computing
program at the University of California, Berkeley (supported by the UC
  Berkeley Chancellor, Vice Chancellor for Research, and Chief
  Information Officer).

\textbf{Declaration of Interest.}
The authors report no conflict of interest.

\appendix

\section{The effects of floating-point errors on the smallest singular value}
\label{sec:fp:err}

The SVD algorithm is backward stable \cite{demmel:book}, which leads
to a well-known estimate \cite{demmel:kahan}
\begin{equation}
  |\hat\sigma_i-\sigma_i|\le p(n)\|\mc J\|\veps, \qquad 1\le i\le n,
\end{equation}
where $\{\hat\sigma_i\}$ are the numerically computed singular values
of the $n\times n$ matrix $\mc J$, $p(n)$ is a slowly growing function of
the matrix dimension, $\|\mc J\|$ is the 2-norm, and $\veps$ is machine
precision. Following the approach of \cite{Achain}, one can show that
$p(n)\le O(n^{2.5})$ using standard backward stability estimates for
Householder transformations in the bidiagonalization phase
\cite{higham:book, demmel:book}; however, this is pessimistic. In
particular, it assumes worst-case $O(n\veps)$ errors when summing $n$
numbers. In practice \cite{lapack:user:guide}, $p(n)$ is often taken
to be 1.  In section~\ref{sec:num:grav}, the condition numbers of $\mc
J^\qua$ for solutions $A$ and $E$ are 1636 and 69074, respectively.
Taking $p(n)=1$, this leads to expected errors in $\chi(s)$ around
$(1636)(2.2\times10^{-16})=3.6\times10^{-13}$ and
$(69074)(2.2\times10^{-16})=1.5\times10^{-11}$, respectively.  Using
Brent's method, we reduced $|\chi(s)|$ to $2.9\times10^{-15}$ for
solution $A$.  However, Brent's method will report the result in which
floating-point errors combine most favorably to minimize $|\chi(s)|$,
so this value likely over-predicts the accuracy actually obtained.
Indeed, if we increase the size of $\mc J^\qua$ from 1537 to 2049 and
recompute $\chi$ without re-optimizing via Brent's method, we obtain
$|\chi|=1.5\times10^{-14}$, which is five times bigger. The flattening
of the high-frequency Chebyshev modes $\check\chi_m$ in panel (d) of
Figure~\ref{fig:bif:svd} suggests floating-point errors around
$10^{-11}$ or $10^{-12}$, which is consistent with the above condition
number estimate. Chebyshev interpolation seems less prone than Brent's
method to optimizing beyond the actual error, so the minimum value of
$|\chi(s)|$ obtained in problem $E$, namely $2.1\times10^{-12}$, may
be accurate. Additional calculations would have to be done in
quadruple-precision to fully quantify the effects of floating-point
arithmetic in double-precision, but this is beyond the scope of the
present work.

\section{Proof of Theorem~\ref{thm:bdd}}
\label{sec:proof}

Let us restate the theorem in a slightly more general form that
simplifies the proof. Theorem~\ref{thm:bdd} is recovered by setting
$\rho=\rho_1/2$ and $\sigma_1=\sigma$.

\begin{theorem}
  Suppose $q^\per=(\eta,\tau,b)$ with $\eta\in\Vp{\sigma}$ for some
  $\sigma>0$.  Suppose also that $\eta$ is real-valued and the
  resulting $J(\alpha_1,\alpha_2)$ in \eqref{eq:govern}, which is
  independent of $\alpha_2$, is non-zero for every $\alpha_1\in\mbb
  T$. Then there exists $\rho_1\in(0,\sigma]$ such that $D_q\mc
  R[q^\per]$ is a bounded operator from $\big(\mc V_{\sigma_1},\mbb
    C,\mbb C\big)$ to $\mc V_\rho$ provided that $0<\rho<\rho_1$ and
  $\sigma_1>\rho$.
\end{theorem}

\begin{proof}
Since $\eta\in\mc V_\sigma^\per$, we know $\eta(\alpha_1,\alpha_2)=
\td\eta(\alpha_1)$ is independent of $\alpha_2$, and $\td\eta(\alpha)$
is real analytic and $2\pi$ periodic. All the torus functions without
a dot in \eqref{eq:lin:R} are then independent of $\alpha_2$, and can
be replaced by the corresponding 1d extracted function (adding a
  tilde), evaluated at $\alpha_1$. The torus operators $\pa_\alpha$
and $H$ are replaced by their 1d variants when this is done.  For
example, $ J(\alpha_1,\alpha_2)=\td J(\alpha_1)$, where $\td
J(\alpha)=(1+H\pa_\alpha\td\eta)^2+(\pa_\alpha\td\eta)^2$. Since $\td
J(\alpha)$ is continuous and assumed non-zero at each
$\alpha\in[0,2\pi]$, it is bounded away from zero. Thus, each of the
functions
\begin{equation}\label{eq:mult:list}
  \begin{gathered}
    1/\big(2\td J\big),\quad \td\kappa,\quad
    b/\big(2\td J^2\big),\quad (1+H\pa_\alpha\td\eta),\quad \pa_\alpha\td\eta,\quad
    3\td\kappa/\big(2\td J\big), \quad \td J^{-3/2}, \quad \pa_\alpha^2\td\eta, \quad
    H\pa_\alpha^2\td\eta
  \end{gathered}
\end{equation}
that appears in \eqref{eq:lin:R} (after setting $\xi=H\eta$)
is real analytic and $2\pi$-periodic. As a result (see, e.g., Lemma
  5.6 of \cite{broer:book}) the 1d Fourier coefficients of each of
these functions decay exponentially. Thus, there exist $C,\rho_1>0$
such that $Ce^{-\rho_1|j|}$ is a common bound on each of these sets of
Fourier coefficients, i.e., if $\td\gamma(\alpha)=\sum_j
\hat\gamma_je^{ij\alpha}$ represents any of the functions in
\eqref{eq:mult:list}, then $|\hat\gamma_j|\le Ce^{-\rho_1|j|}$ for
$j\in\mbb Z$. Now fix $\rho$ in the range $0<\rho<\rho_1$ and let
$\sigma_1>\rho$.

We need to show that $D_q \mc R \big[q^\per \big]$ is a bounded operator
from $(\mc V_{\sigma_1},\mbb C,\mbb C)$ into $\mc V_\rho$. It suffices
to show that the restrictions to the subspaces $(0,\mbb C,\mbb C)$
and $(\mc V_{\sigma_1},0,0)$ are bounded. In the first case, we have
\begin{equation}\label{eq:dot:R:tau:b1}
  D_q \mc R \big(q^\per \big)\big(0,\dot \tau, \dot b \big) =
  P\Big[ \dot b/\big(2J\big)  - \kappa \dot \tau
    \Big].
\end{equation}
Note that $1/(2 J)$ and $\kappa$
are independent of $\alpha_2$. Letting $\gamma(\alpha_1,\alpha_2)=
\td\gamma(\alpha_1)$ represent either of these functions, we have
\begin{equation}
  \|\gamma\|_{\mc V_{\rho}}^2 \le
  C^2 + 2C^2\sum_{j=1}^\infty e^{-2(\rho_1-\rho)j}
  = C^2+ \frac{2C^2}{e^{2(\rho_1-\rho)}-1}
  \le C^2 + \frac{C^2}{\rho_1-\rho} = A^2 < \infty,
\end{equation}
where the last equality defines $A$.  So the $\mc V_\rho$ norm of the
left-hand side of \eqref{eq:dot:R:tau:b1} is bounded by $A(|\dot
  b|+|\dot\tau|)\le A\sqrt2(\dot b^2+\dot\tau^2)^{1/2}$.

Now let $\dot\eta\in\mc V_{\sigma_1}$. We need to compute $\dot R =
D_q\mc R\big[q^\per\big]\big(\dot\eta,0,0\big)$ and bound its norm in
$\mc V_\rho$ by a constant times $\|\dot\eta\|_{\mc V_{\sigma_1}}$. To
compute $\dot R$, we evaluate \eqref{eq:lin:R} in two steps.  First
the terms
\begin{equation}\label{eq:step1:list}
  \big(\pa_\alpha\dot{\xi}\big)=\big(H\pa_\alpha\dot{ \eta}\big), \qquad
  \big(\pa_\alpha\dot{ \eta}\big), \qquad
  \big(\pa_\alpha^2\dot{\eta}\big), \qquad
  \big(\pa_\alpha^2\dot{\xi}\big)=\big(H\pa_\alpha^2\dot{\eta}\big)
\end{equation}
are computed from $\dot\eta$.  The symbols of $H\pa_\alpha$,
$\pa_\alpha$, $\pa_\alpha^2$ and $H\pa_\alpha^2$ are, respectively,
$|l_1+kl_2|$, $i(l_1+kl_2)$, $-(l_1+kl_2)^2$ and
$i(l_1+kl_2)|l_1+kl_2|$, so these operators are bounded from $\mc
V_{\sigma_1}$ to $\mc V_\rho$ since $\sigma_1>\rho$.  For example,
\begin{equation}
  \big\| H\pa_\alpha\dot{\eta}\big\|^2_{\mc V_{\rho}} =
  \sum_{l_1,l_2} \Big| |l_1+kl_2|\hat\eta_{l_1,l_2}e^{\rho(|l_1|+|l_2|)}
  \Big|^2 \le B^2\sum_{l_1,l_2}
  \Big| \hat\eta_{l_1,l_2}e^{\sigma_1(|l_1|+|l_2|)} \Big|^2 = B^2
  \big\| \dot{\eta}\big\|^2_{\mc V_{\sigma_1}} ,
\end{equation}
where $B^2=\max(1,k)\max_{x\ge0}xe^{-(\sigma_1-\rho)x}=
\max(1,k)/[e(\sigma_1-\rho)]<\infty$. The second step is to
consider
\begin{equation}\label{eq:Rdot:again}
   \dot {\mc{R}} = 
  P\bigg[  - 
    \frac{b}{2J^2} \dot {J} + g \dot\eta
    -\tau \dot{\kappa}\bigg].
\end{equation}
The projection $P$ is bounded on $\mc V_\rho$, and $\|g\dot\eta\|_{\mc
  V_\rho}\le g\|\dot\eta\|_{\mc V_{\sigma_1}}$. Substitution of $\dot
J$ and $\dot\kappa$ in \eqref{eq:lin:R} into \eqref{eq:Rdot:again}
yields sums of products containing two factors from the list
\eqref{eq:mult:list} and one factor from the list
\eqref{eq:step1:list}. From the first step, we know each term in the
list \eqref{eq:step1:list} is bounded in $\mc V_\rho$ by a constant
times $\|\dot\eta\|_{\mc V_{\sigma_1}}$. The following lemma shows that
multiplication by any of the terms in \eqref{eq:mult:list} is a
bounded operator on $\mc V_\rho$. Multiplying by two of them is
then also bounded, which concludes the proof.
\end{proof}

\begin{lem} Suppose $\td\gamma(\alpha)=\sum_j\hat\gamma_je^{ij\alpha}$
  with $|\hat\gamma_j|\le Ce^{-\rho_1|j|}$ for positive constants $C$
  and $\rho_1$, and fix $\rho\in(0,\rho_1)$. Then multiplication by
  $\gamma(\alpha_1,\alpha_2)= \td\gamma(\alpha_1)$ is a bounded
  operator on $\mc V_\rho$.
\end{lem}
  
\begin{proof}
Let $u\in\mc V_{\rho}$ and define
$v(\alpha_1,\alpha_2)=\td\gamma(\alpha_1)u(\alpha_1,\alpha_2)$. The
Fourier modes of $v$ are related to those of $u$ by convolution along
horizontal slices, $\hat v_{l_1,l_2} = \sum_j \hat\gamma_{l_1-j}\hat u_{j,l_2}$.
Using the bound on $|\hat\gamma_j|$, we have
\begin{equation}
  \begin{aligned}
    e^{\rho(|l_1|+|l_2|)}\big|\hat v_{l_1,l_2}\big| &\le
    \sum_j Ce^{-\rho_1|l_1-j|}
    e^{\rho(|l_1|-|j|)}\left( e^{\rho(|j|+|l_2|)}\big|\hat u_{j,l_2}\big|\right) \\
    &\le \sum_j Ce^{-(\rho_1-\rho)|l_1-j|}e^{-\rho\big(|l_1-j|-|l_1|+|j|\big)}
    \left( e^{\rho(|j|+|l_2|)}\big|\hat u_{j,l_2}\big|\right).
  \end{aligned}
\end{equation}
Since $|l_1|\le|l_1-j|+|j|$, we can use Young's inequality to conclude
\begin{equation}\label{eq:young}
  \sum_{l_1} e^{2\rho(|l_1|+|l_2|)}\big|\hat v_{l_1,l_2}\big|^2 \le
  A^2 \sum_{j} e^{2\rho(|j|+|l_2|)}\big|\hat u_{j,l_2}\big|^2, \qquad
  (l_2\in\mbb Z),
\end{equation}
where $A = \sum_m Ce^{-(\rho_1-\rho)|m|}\le
C+2C/(\rho_1-\rho)<\infty$. Summing \eqref{eq:young} over $l_2$ and
square rooting gives $\|v\|_{\mc V_\rho}\le A\|u\|_{\mc V_\rho}$, as
claimed.
\end{proof}

\bibliographystyle{abbrv}


\end{document}